\documentclass[a4paper,10pt,reqno]{amsart}
\pdfoutput=1

\usepackage{subfiles}
\usepackage{amssymb}
\usepackage{amsmath}
\usepackage{amsthm}
\usepackage{latexsym}
\usepackage{mathrsfs}
\usepackage{eucal}
\usepackage[safe]{tipa}
\usepackage{slashed}
\usepackage{textalpha}
\usepackage{upgreek}
\usepackage{adjustbox}
\usepackage[numbers]{natbib}
\usepackage{graphicx}

\usepackage{tikz,tikz-3dplot}
\usetikzlibrary{decorations.pathreplacing}

\usepackage{tikz}
\usetikzlibrary{shapes,decorations,backgrounds}
\usetikzlibrary{shapes.multipart}

\usetikzlibrary{cd}
\usetikzlibrary{arrows} 
\usetikzlibrary{graphs}
\usetikzlibrary{intersections}
\usetikzlibrary{positioning}
\usetikzlibrary{decorations.pathmorphing}
\usetikzlibrary{datavisualization} 
\usetikzlibrary{datavisualization.formats.functions}
\usetikzlibrary{intersections}
\usetikzlibrary{hobby}
\usetikzlibrary{through}
\usetikzlibrary{calc}
\usetikzlibrary{fadings}
\usetikzlibrary{patterns}

\usepackage{url}
\usepackage[linktocpage=true,
            pdfauthor={L. Andersson, T. Backdahl, P. Blue and S. Ma},
            pdftitle={Nonlinear radiation gauge for near Kerr spacetimes}]{hyperref}

\usepackage{orcidlink}

\hypersetup{
    colorlinks=true,
    linkcolor=blue,
    filecolor=blue,      
    urlcolor=blue,
}

\DeclareMathOperator{\tho}{\text{\rm\textthorn}}
\DeclareMathOperator{\thop}{\text{\rm\textthorn}^\prime\negthinspace}
\DeclareMathOperator{\edt}{\eth}
\DeclareMathOperator{\edtp}{\eth^\prime\negthinspace}

\DeclareMathOperator{\thoBG}{\mathring{\tho}}
\DeclareMathOperator{\thopBG}{\mathring{\tho}{\vphantom{\tho}}^\prime\negthinspace}
\DeclareMathOperator{\edtBG}{\mathring{\edt}}
\DeclareMathOperator{\edtpBG}{\mathring{\edt}{\vphantom{\edt}}^\prime\negthinspace}

\newcommand{\defn}[1]{\textbf{#1}}

\newcounter{mnotecount}[section]

\theoremstyle{plain}
\newtheorem{theorem}{Theorem}[section]
\newtheorem{definition}[theorem]{Definition}
\newtheorem{lemma}[theorem]{Lemma}
\newtheorem{corollary}[theorem]{Corollary}

\newtheorem{remark}[theorem]{Remark}

\newcounter{step}
\renewcommand{\thestep}{\arabic{step}}
\newenvironment{steps}{\setcounter{step}{0}}{}
\newcommand{\step}[1]{\addtocounter{step}{1}\noindent\textbf{Step \thestep: #1}}

\newcommand{\Naturals}{\mathbb{N}}

\newcommand{\Reals}{\mathbb{R}}
\newcommand{\Complex}{\mathbb{C}}
\newcommand{\Circle}{\mathbb{S}^1}
\newcommand{\Sphere}{\mathbb{S}^2}
\newcommand{\di}{\mathrm{d}}
\newcommand{\metric}{g}

\newcommand{\KerrStarDomain}{\mathcal{K}^*}
\newcommand{\metricBackground}{\mathring{g}}

\newcommand{\metricPerturbation}{h}
\newcommand{\vecL}{l}
\newcommand{\vecM}{m}
\newcommand{\vecMb}{\bar{m}}
\newcommand{\vecN}{n}
\newcommand{\vecH}{e_{\Theta}}  
\newcommand{\vecP}{e_{\Phi}}    
\newcommand{\vecLBackground}{\mathring{\vecL}}
\newcommand{\vecMBackground}{\mathring{\vecM}}
\newcommand{\vecMbBackground}{\bar{\mathring{\vecM}}}
\newcommand{\vecNBackground}{\mathring{\vecN}}
\newcommand{\vecLSpecified}{\hat{\vecL}}
\newcommand{\vecNSpecified}{\hat{\vecN}}

\newcommand{\vecX}{X}

\newcommand{\Lie}{\mathcal{L}}

\newcommand{\Riem}{\mathop{\mathrm{Riem}}{}}

\newcommand{\NLnhCondition}{radiation gauge condition}
\newcommand{\LnhCondition}{linear radiation gauge condition}
\newcommand{\linearTraceCondition}{linear trace condition}
\newcommand{\ClassicalORG}{full radiation gauge of Chrzanowski}
\newcommand{\backgroundHypothesesNoRef}{background hypotheses}
\newcommand{\backgroundHypotheses}{background hypotheses of definition \ref{def:radiationGaugeHypotheses}}
\newcommand{\radiationGaugeHypothesesNoRef}{vacuum, radiation-gauge hypotheses}
\newcommand{\radiationGaugeHypotheses}{vacuum, radiation-gauge hypotheses of definition \ref{def:radiationGaugeHypotheses}}
\newcommand{\frameGaugeHypothesesNoRef}{frame-gauge hypotheses}
\newcommand{\frameGaugeHypotheses}{frame-gauge hypotheses of definition \ref{def:frameGaugeHypotheses}}

\newcommand{\assumeRadiationHypotheses}{Assume the \radiationGaugeHypotheses}
\newcommand{\assumeRadiationHypothesesAndLorentzTransformVariables}{Assume the \radiationGaugeHypotheses{} and a choice of differential Lorentz transformation variables}
\newcommand{\assumeRadiationAndFrameGaugeHypotheses}{Assume the \radiationGaugeHypotheses{} and \frameGaugeHypotheses}

\newcommand{\geometricVariables}{\mathfrak{u}}

\newcommand{\ia}{a}
\newcommand{\ib}{b}

\newcommand{\gaugeRegularityIn}{k}
\newcommand{\gaugeRegularityOut}{k'}

\newcommand{\heightInNullGaugeEnforceability}{h}
\newcommand{\vNew}{v^{\text{new}}}
\newcommand{\rNew}{r^{\text{new}}}
\newcommand{\omegaNew}{\omega^{\text{new}}}

\newcommand{\vSfc}{\hat{v}}
\newcommand{\rSfc}{\hat{r}}
\newcommand{\thetaSfc}{\hat{\theta}}
\newcommand{\phiSfc}{\hat{\phi}}
\newcommand{\omegaSfc}{\hat{\omega}}

\newcommand{\vecPSW}{\xi}
\newcommand{\Flow}[2]{\Phi[#1](#2)}
\newcommand{\Held}[1]{#1{}^{\circ}}
\DeclareMathOperator{\Heldtho}{\Held{\thoBG{}\negthinspace}{}\negthinspace}
\DeclareMathOperator{\Heldeth}{\Held{\edtBG{}\negthinspace}{}\negthinspace}
\DeclareMathOperator{\Heldethp}{\Held{\edtpBG{}}{}\negthinspace}

\def\InvGTrTilde{\tilde{\slashed{G}}{}^{\#}_{\phantom{0}}{}}

\def\LinGZero{\dot{G}_0{}}
\def\LinGOne{\dot{G}_1{}}
\def\LinGTwo{\dot{G}_2{}}

\def\LinGOneDg{\overline{\dot{G}}_1{}}
\def\LinGTwoDg{\overline{\dot{G}}_2{}}
\def\LinGTr{\dot{\slashed{G}}{}}

\allowdisplaybreaks

\title{Nonlinear radiation gauge for near Kerr spacetimes}
\author[L. Andersson]{Lars Andersson \orcidlink{0000-0002-6364-7384}}
\email{laan@aei.mpg.de}
\address{Albert Einstein Institute, Am M\"uhlenberg 1, D-14476 Potsdam, Germany }
\author[T. B\"ackdahl]{Thomas B\"ackdahl \orcidlink{0000-0003-3240-2445}} 
\email{thomas.backdahl@chalmers.se}
\address{Mathematical Sciences, Chalmers University of Technology and University of Gothenburg, SE-412~96 Gothenburg, Sweden}
\author[P. Blue]{Pieter Blue \orcidlink{0000-0002-6195-4022}}
\email{p.blue@ed.ac.uk}
\address{Maxwell Institute and The University of Edinburgh, Peter Guthrie Tait Road, Edinburgh, EH9 3FD, UK}
\author[S. Ma]{Siyuan Ma \orcidlink{0000-0003-2893-7674}}
\email{siyuan.ma@aei.mpg.de, siyuan.ma@sorbonne-universite.fr}
\address{Laboratoire Jacques-Louis Lions, Universit\'{e} Pierre et Marie Curie, 4 place Jussieu, 75005 Paris, France \and Albert Einstein Institute, Am M\"uhlenberg 1, D-14476 Potsdam, Germany}

\allowdisplaybreaks[2]

\numberwithin{equation}{section}

\begin{document}

\begin{abstract}
In this paper, we introduce and explore the properties of a new gauge choice for the vacuum Einstein equation inspired by the ingoing and outgoing radiation gauges (IRG, ORG) for the linearized vacuum Einstein equation introduced by Chrzanowski in his work on metric reconstruction \cite{Chrzanowski} on the Kerr background. It has been shown by Price, Shankar and Whiting \cite{Price:2006ke} that the IRG/ORG are consistent gauges for the linearized vacuum Einstein equation on Petrov type II backgrounds. In \cite{Andersson:2019dwi}, the ORG was used in proving linearized stability for the Kerr spacetime, and the new non-linear radiation gauge introduced here is a direct generalization of that gauge condition, and is intended to be used to study the stability of Kerr black holes under the evolution generated by the vacuum Einstein equation. 
\end{abstract}

\maketitle

\tableofcontents

\section{Introduction}

Given $M>0$ and $a\in(-M,M)$, for $(v,r,\omega)\in\Reals\times(0,\infty)\times\Sphere$ and $(\theta,\phi)$ spherical coordinates on $\Sphere$, the Kerr metric in Eddington-Finkelstein coordinates takes the form 
\begin{subequations}
\begin{align}
\metricBackground 
={}& \left(1-\frac{2Mr}{\Sigma}\right) \di v^2
+\frac{4Mra}{\Sigma}\sin^2\theta  \di v\di\phi
-\frac{(r^2+a^2)^2-a^2\Delta\sin^2\theta}{\Sigma}\sin^2\theta\di\phi^2
\nonumber\\
&-\Sigma\di\theta^2
-2\di v \di r
+2a\sin^2\theta \di r\di\phi ,
\label{eq:KerrMetric} \\
\Sigma ={}& r^2+a^2\cos^2\theta,\quad
\Delta ={} r^2-2Mr +a^2 .
\end{align}
\end{subequations}
For $M>0$ and $a\in(-M,M)$, this metric describes a subextremal black hole geometry. As explained in many textbooks (e.g. \cite{ONeill:Kerr}), the metric \eqref{eq:KerrMetric} extends smoothly to the set $\KerrStarDomain=\Reals\times(0,\infty)\times\Sphere$, in particular to the north and south poles, and there is a further analytic extension, extending beyond $v=\pm\infty$ and (for $a\not=0$) to $r<0$. The Kerr space-time is of Petrov type D (or $\{2,2\}$), which means there are two, repeated principal null directions;
a future-directed ingoing (respectively outgoing) principal null vector is a positive multiple of $\vecNSpecified$ (respectively $\vecLSpecified$), where 
\begin{subequations}
\label{eq:PNVs}
\begin{align}
\vecNSpecified ={}& -\partial_r ,
\label{eq:ingoingPNV}\\
\vecLSpecified ={}& \frac{\Delta}{2}\partial_r +\left((r^2+a^2)\partial_v+a\partial_\phi\right).
\label{eq:outgoingPNV}
\end{align}
\end{subequations}

Central to this paper is the following gauge condition: 
\begin{definition}
\label{def:NLORG}
Let $M>0$ and $a\in(-M,M)$. 
Let $\metricBackground$ be the Kerr metric on $\KerrStarDomain$, and let $\vecN$ be a future-directed, ingoing principal null vector. 
Let $U$ be an open subset of $\KerrStarDomain$.

A symmetric tensor $\metric$ on $U$ is defined to satisfy the \defn{\NLnhCondition{}} iff
\begin{align}
\label{eq:NLnh}
\vecN^a\metric_{ab} ={}& \vecN^a\metricBackground_{ab} .  
\end{align}
\end{definition}

We shall use the term diffeomorphism gauge to be synonymous with a local diffeomorphism. In order to be able to state our main results, we shall need the following, somewhat technical definition. For convenience, we define a reference Riemannian metric on $\KerrStarDomain$ from which we further define, for any $\gaugeRegularityIn\in\Naturals$, the $C^\gaugeRegularityIn$ norm with respect to the reference metric on any subset of $\KerrStarDomain$. It is well known that when dealing with diffeomorphism gauges, it is unfortunately common to lose regularity and to need to restrict to somewhat smaller sets. The relevant sets for the following definition are illustrated in figure \ref{fig:SetsForImposingGauge}. 

\begin{definition}
\label{def:diffeomorphismGauge}
In this paper, given a nonnegative integer $\gaugeRegularityIn$ and an open set $V\subset\KerrStarDomain$, a $C^{\gaugeRegularityIn}$ \defn{diffeomorphism gauge} is a map $\Phi:V\rightarrow \KerrStarDomain$ such that $\Phi$ is a $C^{\gaugeRegularityIn}$ diffeomorphism of $V$ to its image.
  
Let $(X,Y,I,J,h,U,V)$ be such that: $X$ is a bounded, open subset of $\Reals\times\Sphere$; $Y$ is a open set such that its closure is a subset of $X$; $0<J<I<\infty$; $\heightInNullGaugeEnforceability:X\rightarrow(M/2,\infty)$ is smooth; and $U$ and $V$ are the spacetime slabs $U=\{(v,r,\omega):(v,\omega)\in X,\heightInNullGaugeEnforceability(v,\omega)-I<r<\heightInNullGaugeEnforceability(v,\omega)+I\}$, $V=\{(v,r,\omega):(v,\omega)\in Y, \heightInNullGaugeEnforceability(v,\omega)-J<r<\heightInNullGaugeEnforceability(v,\omega)+J\}$. A diffeomorphism $\Phi$ is defined to be \defn{compatible with $(X,Y,I,J,\heightInNullGaugeEnforceability,U,V)$} if $\Phi(V)\subset U$. Abusing notation, we use $\heightInNullGaugeEnforceability(X)$ (and similarly for $\heightInNullGaugeEnforceability(Y)$) to denote the graph in $\Reals\times\Reals\times\Sphere$ of $\heightInNullGaugeEnforceability$ over $X$ rather than the image in $\Reals$ of $X$.
\end{definition}

Our first result is that for initial data that is close to data from the Kerr spacetime, it is possible to construct a diffeomorphism gauge so as to impose the \NLnhCondition{}.

\begin{theorem}[Enforceability of the \NLnhCondition]
\label{thm:NLnhEnforceability}
Let $M>0$ and $a\in(-M,M)$. 
Let $\metricBackground$ be the Kerr metric on $\KerrStarDomain$, and let $\vecN$ be a future-directed, ingoing principal null vector. 
Let $(X,Y,I,J,h,U,V)$ be as in definition \ref{def:diffeomorphismGauge}, and let $\gaugeRegularityOut$ be a sufficiently large integer. 

There exist $\varepsilon_0>0$, $\gaugeRegularityIn>\gaugeRegularityOut$, and $K>0$ such that, 
if $\metric_{ab}$ is a symmetric $(0,2)$ tensor satisfying $|\metric-\metricBackground|_{C^{\gaugeRegularityIn}(U)}<\varepsilon_0$, 
then
there is a $C^{\gaugeRegularityOut}$ diffeomorphism gauge $\Phi$ compatible with $(X, Y, I, J, h, U, V )$ such that $\Phi^{-1}_*\metric$ satisfies the \NLnhCondition{} on $V$. Furthermore, there is the following bound of the initial data for $\Phi^{-1}_*\metric$ in terms of the initial data for $\metric$: 
$|\Phi^{-1}_*\metric-\metricBackground|_{C^{\gaugeRegularityOut}(h(Y))}$ $\leq K|\metric-\metricBackground|_{C^{\gaugeRegularityIn}(h(X))}$. 
\end{theorem}

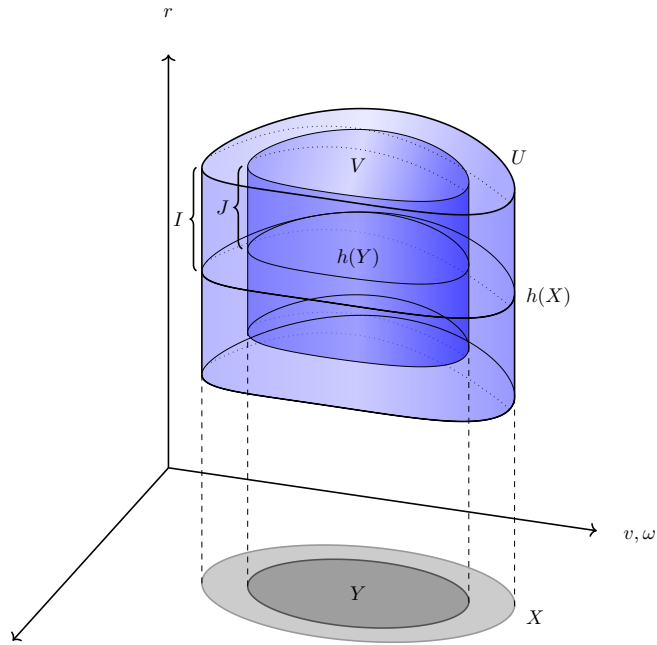
\begin{figure}
\label{fig:SetsForImposingGauge}
\scalebox{0.75}{
\newcommand{\cylinderline}[1]{ ({2+sin(#1)}, {2.5+sqrt(2)*cos(#1)}, {2 + sin(#1)^2/4}) -- ({2+sin(#1)}, {2.5+sqrt(2)*cos(#1)}, {4 + sin(#1)^2/4})}
\newcommand{\innercylinderline}[1]{ ({2+sin(#1)/sqrt(2)}, {2.5+cos(#1)}, {2.2+1/4 + sin(#1)^2/8}) -- ({2+sin(#1)/sqrt(2)}, {2.5+cos(#1)}, {3.8+1/4 + sin(#1)^2/8})}
\newcommand{\bottomcylinderline}[1]{ ({2+sin(#1)}, {2.5+sqrt(2)*cos(#1)}, 0) -- ({2+sin(#1)}, {2.5+sqrt(2)*cos(#1)}, {2 + sin(#1)^2/4})}
\newcommand{\innerbottomcylinderline}[1]{ ({2+sin(#1)/sqrt(2)}, {2.5+cos(#1)}, {0}) -- ({2+sin(#1)/sqrt(2)}, {2.5+cos(#1)}, {2.2+1/4 + sin(#1)^2/8})}
\pgfmathsetmacro{\LeftAngle}{165}
\pgfmathsetmacro{\RightAngle}{-12}
\tdplotsetmaincoords{66}{110}
\begin{tikzpicture}[tdplot_main_coords,scale=2]
 \draw[thick,->] (0,0,0) coordinate (O) -- (4,0,0) coordinate(X) node[pos=1.1]{};
 \draw[thick,->] (O) -- (0,4,0) node[pos=1.1]{$v,\omega$};
 \draw[thick,->] (O) -- (0,0,4) node[pos=1.1]{$r$};
 \path[opacity=0.6,left color=blue!50,right color=blue!80,middle color=blue!20,shading angle=75]
  plot[variable=\t,domain=0:360,smooth,samples=90] ({2+sin(\t)/sqrt(2)}, {2.5+cos(\t)}, {3.8+1/4 + sin(\t)^2/8}) -- cycle;
 \path[opacity=0.4,left color=blue!50,right color=blue!80,middle color=blue!20,shading angle=75]
  plot[variable=\t,domain=0:360,smooth,samples=90] ({2+sin(\t)}, {2.5+sqrt(2)*cos(\t)}, {4 + sin(\t)^2/4}) -- cycle;
 \draw[dotted] plot[variable=\t,domain=-1:1,smooth,samples=90] ({2}, {2.5+sqrt(2)*\t}, {3+3/2 - (\t)^2/2});
 \draw[dotted] plot[variable=\t,domain=-1:1,smooth,samples=90] ({2}, {2.5+\t}, {2.8+3/2 - (\t)^2/4});
 \draw[dotted] plot[variable=\t,domain=-1:1,smooth,samples=90] ({2}, {2.5+sqrt(2)*\t}, {2+3/2 - (\t)^2/2});
 \draw[dotted] plot[variable=\t,domain=-1:1,smooth,samples=90] ({2}, {2.5+sqrt(2)*\t}, {1+3/2 - (\t)^2/2});
 \draw[dotted] plot[variable=\t,domain=-1:1,smooth,samples=90] ({2}, {2.5+\t}, {1.2+3/2 - (\t)^2/4});
 \path[opacity=0.6, left color=blue!70, right color=blue, middle color=blue!50]
  plot[variable=\t,domain=\RightAngle:\LeftAngle,smooth,samples=90] ({2+sin(\t)/sqrt(2)}, {2.5+cos(\t)}, {2.2+1/4 + sin(\t)^2/8}) -- 
  plot[variable=\t,domain=\LeftAngle:\RightAngle,smooth,samples=90] ({2+sin(\t)/sqrt(2)}, {2.5+cos(\t)}, {3.8+1/4 + sin(\t)^2/8}) -- cycle;
 \path[opacity=0.4, left color=blue!70, right color=blue, middle color=blue!50]
  plot[variable=\t,domain=\RightAngle:\LeftAngle,smooth,samples=90] ({2+sin(\t)}, {2.5+sqrt(2)*cos(\t)}, {2 + sin(\t)^2/4}) -- 
  plot[variable=\t,domain=\LeftAngle:\RightAngle,smooth,samples=90] ({2+sin(\t)}, {2.5+sqrt(2)*cos(\t)}, {4 + sin(\t)^2/4}) -- cycle;
 \draw[thick] plot[variable=\t,domain=0:360,smooth,samples=90] ({2+sin(\t)}, {2.5+sqrt(2)*cos(\t)}, {4 + sin(\t)^2/4}) -- cycle;
 \draw[thick] plot[variable=\t,domain=\RightAngle:\LeftAngle,smooth,samples=90] ({2+sin(\t)}, {2.5+sqrt(2)*cos(\t)}, {3 + sin(\t)^2/4});
 \draw plot [variable=\t,domain=\LeftAngle+360:\RightAngle,smooth,samples=90] ({2+sin(\t)}, {2.5+sqrt(2)*cos(\t)}, {3 + sin(\t)^2/4});
 \draw[thick] plot[variable=\t,domain=\RightAngle:\LeftAngle,smooth,samples=90] ({2+sin(\t)}, {2.5+sqrt(2)*cos(\t)}, {2 + sin(\t)^2/4});
 \draw plot [variable=\t,domain=\LeftAngle+360:\RightAngle,smooth,samples=90] ({2+sin(\t)}, {2.5+sqrt(2)*cos(\t)}, {2 + sin(\t)^2/4});
 \draw[thick, fill=gray, opacity=0.4] plot[variable=\t,domain=0:360,smooth,samples=90] ({2+sin(\t)}, {2.5+sqrt(2)*cos(\t)}, 0) -- cycle;
 \draw[thick] \cylinderline{\LeftAngle};
 \draw[thick] \cylinderline{\RightAngle};
 \draw[dashed] \bottomcylinderline{\LeftAngle};
 \draw[dashed] \bottomcylinderline{\RightAngle};
 \draw \innercylinderline{\LeftAngle};
 \draw \innercylinderline{\RightAngle};
 \draw[dashed] \innerbottomcylinderline{\LeftAngle};
 \draw[dashed] \innerbottomcylinderline{\RightAngle};
 \draw plot[variable=\t,domain=0:360,smooth,samples=90] ({2+sin(\t)/sqrt(2)}, {2.5+cos(\t)}, {3.8+1/4 + sin(\t)^2/8}) ;
 \draw plot[variable=\t,domain=0:360,smooth,samples=90] ({2+sin(\t)/sqrt(2)}, {2.5+cos(\t)}, {3+1/4 + sin(\t)^2/8}) ;
 \draw plot[variable=\t,domain=0:360,smooth,samples=90] ({2+sin(\t)/sqrt(2)}, {2.5+cos(\t)}, {2.2+1/4 + sin(\t)^2/8}) ;
 \draw[thick, fill=gray, opacity=0.6] plot[variable=\t,domain=0:360,smooth,samples=90] ({2+sin(\t)/sqrt(2)}, {2.5+cos(\t)}, 0) ;
\node at (2,2.5,0) {$Y$};
\node at (2,2.5,3.25) {$h(Y)$};
\coordinate[label=right:$X$] (Xlabel) at (2,4,0);
\coordinate[label=right:$h(X)$] (hXlabel) at (2,4,3.1);
\node at (2,2.5,4.15) {$V$};
\node at (2,4,4.45) {$U$};
\draw[thick, , decorate,decoration={brace,amplitude=4pt}] ({2+sin(\LeftAngle)/sqrt(2)}, {2.45+cos(\LeftAngle)}, {3+1/4 + sin(\LeftAngle)^2/8}) -- node[anchor=east] {$J\;$} ({2+sin(\LeftAngle)/sqrt(2)}, {2.45+cos(\LeftAngle)}, {3.8+1/4 + sin(\LeftAngle)^2/8}) ;
\draw[thick, decorate,decoration={brace,amplitude=3pt} ]  ({2+sin(\LeftAngle)}, {2.45+sqrt(2)*cos(\LeftAngle)}, {3 + sin(\LeftAngle)^2/4}) -- node[anchor=east] {$I\;$} ({2+sin(\LeftAngle)}, {2.45+sqrt(2)*cos(\LeftAngle)}, {4 + sin(\LeftAngle)^2/4}) ;
\end{tikzpicture}
}
\caption{The sets arising in definition \ref{def:diffeomorphismGauge}.}
\end{figure}

Our other main result is to make the vacuum Einstein equation well-posed by constructing a first-order symmetric hyperbolic system. This involves using the Geroch-Held-Penrose (GHP) formalism \cite{GHP,PenroseRindler} to construct components of $\metric-\metricBackground$, the difference between the connection coefficients of $\metric$ and of $\metricBackground$, the difference of the corresponding curvatures, and some additional variables that describe the difference between foreground and background frames, which we call differential Lorentz transformations. In applying the GHP formalism, it is necessary to make a choice of an equivalence class of frames, which we refer to as a choice of frame gauge. This is explained in section \ref{s:FOSH}. 

\begin{theorem}[Well-posedness]
\label{thm:NLnhIsLWP}
Let $M>0$ and $a\in(-M,M)$. 
Let $\metricBackground$ be the Kerr metric on $\KerrStarDomain$, and let $\vecN$ be a future-directed, ingoing principal null vector. 
\begin{enumerate}
\item The vacuum Einstein equation, the \NLnhCondition{} and the frame gauge hypotheses in definition \ref{def:frameGaugeHypotheses} together imply a first-order symmetric hyperbolic system for the geometric variables in definition \ref{def:geometricVariables}. 
\label{pt:FOSH}
\item The geometric variables in definition \ref{def:geometricVariables} uniquely determine a metric $\metric$. 
\label{pt:systemDeterminesMetric}
\item If the initial data for first-order symmetric hyperbolic system in \ref{pt:FOSH} arise from initial data for the vacuum Einstein equation, then the metric determined by \ref{pt:FOSH}-\ref{pt:systemDeterminesMetric} satisfies the vacuum Einstein equation. 
\end{enumerate}
\end{theorem}

It is well established that first-order symmetric-hyperbolic systems are well posed in suitable function spaces \cite{MR748308}. Note that the geometric nature of our variables ensures that the first-order symmetric hyperbolic system is well-defined for all $\omega\in\Sphere$ and not just in a particular coordinate patch on the sphere. 

In the final two sections of this paper, we go further in relating the \NLnhCondition{} for the Einstein equation to previously existing results for the linearized Einstein equation. In section \ref{s:traceCondition}, we apply a residual gauge transformation to further impose a condition on the trace $\metric^{ab}\metricBackground_{ab}$ analogous to that imposed in the linear case by \cite{Chrzanowski, Price:2006ke}. In the final section of this paper, we linearize the Einstein equation with the radiation and frame gauge conditions imposed, and we show that resulting linearized metric coefficients coincide with those constructed in our previous work on the linear stability of the Kerr metric \cite{Andersson:2019dwi}. In the previous and current works, we have made different choices in decomposing the linearized connection and curvature coefficients; the different choices of linearized variables are related by a linear change of variables and satisfy equivalent PDE systems, as explained in section \ref{s:linearization}.

\subsection{Motivation and relation to existing literature}

In this paper, we introduce a new gauge choice to study the stability of Kerr black holes under the evolution generated by the vacuum Einstein equation. This gauge is inspired by what is called the ``outgoing radiation gauge (ORG)'' in \cite{Chrzanowski, Price:2006ke}, a so-called linearized gauge for the linearized Einstein equation.

The Kerr stability problem remains a central problem in the study of the Einstein equation as a hyperbolic differential equation. In brief, the problem is to show that, for any initial data that generates a solution containing a Kerr exterior, any sufficiently small perturbation of such initial data will generate a solution which contains a region that, in the future, converges to some Kerr exterior. So far, most work has focused on the linearized Einstein equation and models for it, such as the wave and Maxwell equations
\cite{bluesoffer03mora,BlueSterbenz,dafrod09red,blue:soffer:integral,blue2008decay,tataru2011localkerr,AnderssonBlue:KerrWave,larsblue15Maxwellkerr,Dafermos:2014cua,pasqualotto2019spin,Ma2017Maxwell}  and the linearized gravity \cite{DHR:SchwarzschildStability,Hung:2017qop,Ma:2017bxq,Dafermos:2017yrz,Andersson:2019dwi,Hafner:2019kov}. Quite recently, a few works \cite{klainermanszeftel2020global,dafermos2021non,Klainerman:2021qzy} have made important progress on the full nonlinear stability of Kerr spacetimes.

We are particularly interested in the following approach to proving decay of solutions to the linearization of the Einstein equation on a Kerr background: The Kerr solutions admit a pair of principal null vectors. At least locally, one can construct a basis consisting of these principal null vectors, and an oriented orthonormal basis for the plane orthogonal to them. The GHP formalism uses spinors to construct the analogue of the Cartan formalism for such bases \cite{GHP,PenroseRindler}. Of central importance, in this set up, the two extreme components of the linearized curvature each satisfy a decoupled equation known as the Teukolsky master equation (TME) \cite{Teukolsky}. Chrzanowski \cite{Chrzanowski} introduced a linearized gauge transformation, 
and showed that, in this linearized gauge, all linearized metric coefficients can be reconstructed from the linearized curvature. In the very slowly rotating case, uniform energy bounds and integrated local energy decay has been shown for the Teukolsky equation \cite{Ma:2017bxq,Dafermos:2017yrz}. Recently similar results have been obtained using physical-space methods \cite{Giorgi:2021skz}. In the full subextremal range, decay is proved for bounded frequencies in \cite{Shlapentokh-Rothman:2020vpj}. 
Higher order perturbations of the Kerr spacetime was studied in \cite{Campanelli:1998jv, Green:2019nam, Loutrel:2020wbw}.

From such results, we have shown that it follows that there are pointwise decay estimates for the linearized metric coefficients in the linear ORG \cite{Andersson:2019dwi}. In spherical symmetry, this linearized gauge choice uses the same choice of null tetrad as in the linearized gauge choice arising from double null coordinates, which has been used previously to show decay of linearized perturbations about Schwarzschild black holes \cite{DHR:SchwarzschildStability}. A significantly different approach to the linear stability problem was taken in \cite{Hafner:2019kov}.

As a geometric equation for curvature, the Einstein equation is invariant under changes of coordinate or, equivalently, diffeomorphisms.  As a consequence of the resulting freedom to choose a diffeomorphism gauge, for any solution of the Einstein equation $\metricBackground$, any vector field $\vecX$, and any solution $\metricPerturbation$ of the linearization of the Einstein equation $\metricBackground$, one finds that $\metricPerturbation +\Lie_{\vecX}\metricBackground$ is also a solution of the linearization of the Einstein equation about $\metricBackground$. The freedom to add any $\Lie_{\vecX}\metricBackground$ is called linearized gauge freedom.

For the linearized Einstein equation, the radiation gauge can be defined in the following way. 

\begin{definition}
\label{def:linearORG}
Let $M>0$ and $a\in(-M,M)$. 
Let $U$ be a subset of the maximal extension of the Kerr black hole with mass and angular momentum per unit mass $M,a$, and let $\metricBackground$ be the metric on $U$.
Let $\vecN$ denote an ingoing principal null vector on $U$.\footnote{Because equation \eqref{eq:linearnh} is homogeneous, the normalisation of $\vecN$ does not need to be specified.}
Let $\metricPerturbation$ be a symmetric $(0,2)$ tensor field on $U$. 

$\metricPerturbation$ is defined to satisfy the \defn{\LnhCondition}
\footnote{Note that \cite{Price:2006ke} calls this the \defn{$\vecN\cdot\metricPerturbation$ gauge}.} if 
\begin{subequations}
\begin{align}
\vecN^a\metricPerturbation_{ab} ={}& 0  ,
\label{eq:linearnh}
\end{align}
and to satisfy the \defn{\linearTraceCondition} if
\begin{align}
\metricBackground^{ab} \metricPerturbation_{ab} ={}& 0 .
\label{eq:linearTraceCondition}
\end{align}
\end{subequations}
$\metricPerturbation$ is defined to satisfy the \defn{\ClassicalORG} (\defn{ORG}) if it satisfies both the radiation gauge and the linear trace conditions. 
\end{definition}
Essentially, this was first introduced in \cite{Chrzanowski} and then clarified in \cite{Price:2006ke}. 
\cite{Price:2006ke} has shown that if $\metricPerturbation$ satisfies the \LnhCondition{}, then there is a linearized gauge transformation so that $\metricPerturbation +\Lie_{\vecX}\metricBackground$ satisfies the \ClassicalORG. From the perspective of naive function counting, it is surprising that all five of the conditions can be imposed, not merely the four of the linear null condition. A careful reading of \cite{Price:2006ke} shows that for any linearized metric (i.e. symmetric $(0,2)$ tensor), one can construct a linear gauge transformation so that the \LnhCondition{} is satisfied on open sets. Furthermore, one can apply further residual gauge transformations that maintain the \LnhCondition{}. From the perspective of naive function counting, it is convenient to consider residual gauge transformations as diffeomorphisms of the initial data set that can be applied in addition to the four gauge conditions that are applied within the spacetime and that generate a well-posed dynamics when combined with the Einstein equation and a frame gauge condition.

While it is clear that if one has a smooth family of gauge transformations $\Phi_t$ then the linearization of this family determines a linear gauge transformation $\frac{\di}{\di t}\Phi_t^* h$, it is not clear that any so-called linear gauge transformation genuinely arises from the linearization of a family of gauge transformations, nor that, even if they did, the family of gauge transformations would have desirable properties. The main results of this paper, theorems \ref{thm:NLnhEnforceability}-\ref{thm:NLnhIsLWP}, show that the \LnhCondition{} does arise from the linearization of a gauge for the full Einstein equation, namely the \NLnhCondition, and that this gauge together with a frame gauge choice gives a locally well-posed Cauchy problem for the Einstein equation. Furthermore, in section \ref{s:traceCondition}, we show that for the full Einstein equation, one can make use of the diffeomorphism gauge freedom to find a gauge that both satisfies the \NLnhCondition{} and such that the trace $\metric^{ab}\metricBackground_{ab}-4$ vanishes quadratically, and the frame gauge hypotheses in definition \ref{def:frameGaugeHypotheses} can be further imposed such that the well-posedness theorem \ref{thm:NLnhIsLWP} holds additionally. Thus, the linearization of this system can be seen as satisfying \ClassicalORG.

The formalism we use to treat the \NLnhCondition{} has important similarities with and differences from the formalism based on principal geodesic structures in \cite{Klainerman:2021qzy}. Both formalisms specify one null vector field that is tangent to null geodesics. They are both frame formalisms based on a choice of a pair of null vector fields such that the orthogonal plane fails to be integrable in the sense of Frobenius. 
By exclusively using properly weighted quantities, we can use the GHP formalism without specifying a choice of basis for the orthogonal plane and, hence, avoid the ``artificial gauge singularities'' noted in \cite[p27]{Klainerman:2021qzy}.
Perhaps in most striking contrast to the previous literature, both formalisms use not one but two classes of frame. In obtaining the first-order symmetric-hyperbolic form of the Einstein equations under the \NLnhCondition{} and the frame gauge hypotheses in definition \ref{def:frameGaugeHypotheses}, we use the background principal null vectorfields $\vecLBackground,\vecNBackground$ of the background Kerr geometry $\metricBackground$ and a foreground pair of vectorfields $\vecL,\vecN=\vecNBackground$ that are null with respect to the new, foreground geometry $\metric$. To each pair of null vectors, we associate the plane that is orthogonal in the relevant geometry. In contrast, the two frames used in the principal geodesic structures of \cite{Klainerman:2021qzy} share the same null legs, but one frame is completed by adjoining a basis for the (non-integrable) orthogonal plane while the other frame is completed by adjoining a basis for the (integrable) tangent space of the spheres that are $r,v$ level sets. Our two classes of frames coincide when the metric is exactly the Kerr metric, which suggests the possibility that the formalism based on the \NLnhCondition{} will provide significant simplifications, in addition to connecting with the previously existing physics literature.

\subsection{Structure of the proofs and of the paper}
Section \ref{s:ProofOfEnforecabilityOfTheNLORG} proves theorem \ref{thm:NLnhEnforceability} about the existence of a gauge transformation to impose the \NLnhCondition. Section \ref{s:FOSH} proves theorem \ref{thm:NLnhIsLWP} on the existence of a first-order symmetric hyperbolic system for the metric components and other geometric quantities; this section includes the definition of the frame gauge and the relevant geometric variables in terms of the GHP formalism. Section \ref{s:traceCondition} proves that perturbations of the trace, $\metric^{ab}\metricBackground_{ab}-4$, can be made to vanish quadratically, in a quantifiable sense introduced in that section; this section is heavily inspired by \cite{Price:2006ke}. Section \ref{s:linearization} treats the linearization of the Einstein equation under our gauge choices and makes a comparison with our earlier work \cite{Andersson:2019dwi}.

\section{Imposing the \NLnhCondition{}} 
\label{s:ProofOfEnforecabilityOfTheNLORG}
This section begins with some definitions to simplify discussion of the geometry in the directions orthogonal to the principal null vectors. There is then a lemma about metrics satisfying the \NLnhCondition, in particular that the flow along $\vecN=-\partial_r$ generates affinely parameterized null geodesics, as is the case in the Kerr spacetime. Finally, there is a proof of the enforceability of the \NLnhCondition, which is based on appropriately constructing null geodesics. This completes the proof of theorem \ref{thm:NLnhEnforceability}. 

Recall the notions of real null tetrad and complex null tetrads. These are given in appendix \ref {sec:worry}. Unless otherwise specified, a null tetrad is understood to mean an oriented complex null tetrad.

\begin{definition}
Let $M>0$ and $a\in(-M,M)$. Let $U$ be an open subset of $\KerrStarDomain$ parameterized by $(v,r,\omega)$. In the domain of the standard spherical coordinates, define 
\begin{subequations}
\begin{align}
\vecH ={}&\partial_\theta ,\\
\vecP ={}&\partial_\phi-a\sin^2\theta\partial_v .
\end{align}
\end{subequations}
\end{definition}

\begin{lemma}[Necessary results of the \NLnhCondition]
\label{lem:NLnhConsequences}
Let $M>0$ and $a\in(-M,M)$. Let $U$ be an open subset of $\KerrStarDomain$ parameterized by $(v,r,\omega)$. 

If $\metric$ is a Lorentzian metric on $U$ that satisfies the \NLnhCondition, then
\begin{enumerate}
\item $\partial_r$ is null. 
\item In the portion of $U$ covered by spherical coordinates, 
$\vecP$ and $\vecP$ are orthogonal to $\partial_r$. 
\item At each point in the domain of the spherical coordinates, if $\vecN=-\partial_r$ and $\vecM$ is a complex linear combination of $\vecH$ and $\vecP$ such that $\vecM$ and its complex conjugate $\vecMb$  are a complex basis for the space spanned by $\vecH$ and $\vecP$ such that $\metric(\vecM,\vecM)=0$ and $\metric(\vecM,\vecMb)=-1$, then there is a unique, future-directed null vector $\vecL$ that is orthogonal to $\vecM$ and $\vecMb$ and that satisfies $\metric(\vecL,\vecN)=1$. Furthermore, if $\metric(\vecM,\vecM)=\metric(\vecMb,\vecMb)=0$ and $\metric(\vecM,\vecMb)=-1$, then $(\vecL,\vecN,\vecM,\vecMb)$ form a null tetrad. 
\label{pt:NLnhHelper:ExistenceOfTetrad}
\item For all $(v_0,\omega_0)\in\Reals\times\Sphere$, the curve $\gamma(s)=(v_0,s,\omega_0)$ is a (not necessarily affinely parameterized) geodesic. 
\item If $\Sigma$ is $3$-submanifold of $U$ parameterized by $(v,\omega)$, and if $(\vSfc,\omegaSfc)$ are the restrictions of $(v,\omega)$ to $\Sigma$, $(\thetaSfc,\phiSfc)$ denote the values of the standard spherical coordinate corresponding to $\omegaSfc$, and $\rSfc$ is the restriction of $r$ to $\Sigma$, then, in the domain of the standard spherical coordinates, $\partial_{\vSfc}, \partial_{\thetaSfc}, \partial_{\phiSfc} \in T\Sigma\subset TU$ satisfy 
\begin{subequations}
\begin{align}
\metric(\partial_r,\partial_{\vSfc}) ={}& -1 ,\\
\metric(\partial_r,\partial_{\thetaSfc}) ={}& 0, \\
\metric(\partial_r,\partial_{\phiSfc}) ={}& -a\sin^2\theta .
\end{align}
\label{eq:surfaceInnerProducts}
\end{subequations}
\end{enumerate}
\end{lemma}
\begin{proof}
Unless otherwise specified, in this proof, we work in the domain of the spherical coordinates and then extend by continuity. 
Since $\metric_{r\ia}\di x^\ia=(\di v+a\sin^2\theta\di\phi)$, it follows that $\metric(\partial_r,\partial_r)$ $=\metric_{r\ia}\di x^\ia(\partial_r)$ $=(\di v+a\sin^2\theta\di\phi)(\partial_r)$ $=0$, that $\metric(\partial_r,\vecP)$ $=(\di v+a\sin^2\theta\di\phi)(\partial_\theta)$ $=0$, and that $\metric(\partial_r,\vecP)$ $=(\di v+a\sin^2\theta\di\phi)(\partial_\phi-a\sin^2\theta\partial_v)$ $=0$, which establishes the first two claims in the domain of the spherical coordinates. By continuity, $\partial_r$ remains null at the poles of the spherical coordinates. The plane orthogonal to $\vecM$ and $\vecMb$ is a $1+1$-dimensional Lorentzian vector space with a time orientation, and, since $\vecN$ is null but not zero, the existence of a unique $\vecL$ as in the statement of point \ref{pt:NLnhHelper:ExistenceOfTetrad} holds. 

To show that the curves $(v_0,s,\omega_0)$ are (not necessarily affinely parameterized) geodesics it is sufficient to show that $\ddot\gamma^\ib=\vecN^\ia\nabla_\ia\vecN^\ib$ is parallel to $\vecN$. This is equivalent to $\ddot\gamma^\ib\vecN_\ib$ $=\ddot\gamma^\ib\vecM_\ib$ $=\ddot\gamma^\ib\vecMb_\ib=0$. Trivially, 
\begin{align}
\vecN_\ib\vecN^\ia\nabla_\ia\vecN^\ib
={}&\frac12\vecN^\ia\nabla_\ia(\vecN_\ib\vecN^\ib)
=0, 
\end{align}
since $\vecN_\ib\vecN^\ib=0$. Before continuing, first observe that the commutator $[\vecN,\vecM]$ satisfies
\begin{subequations}
\begin{align}
[\vecN,\vecM]
={}&-[\partial_r,\vecM^\Theta\vecH +\vecM^\Phi\vecP] \nonumber\\
={}&-[\partial_r,\vecM^\Theta\partial_\theta +\vecM^\Phi(\partial_\phi-a\sin^2\theta\partial_v)] \nonumber\\
={}&-(\partial_r\vecM^\Theta)\vecH -(\partial_r\vecM^\Phi)\vecP ,\\
\metric(\vecN,[\vecN,\vecM])
={}& 0. 
\end{align}
\end{subequations}
Now, observe, from the orthogonality conditions and from properties of the commutator, that 
\begin{subequations}
\begin{align}
\vecM_\ib\vecN^\ia\nabla_\ia\vecN^\ib
={}& \vecN^\ia\nabla_\ia(\vecM_\ib\vecN^\ib) -\vecN_\ib\vecN^\ia\nabla_\ia\vecM^\ib
=-\vecN_\ib\vecN^\ia\nabla_\ia\vecM^\ib,\\
\vecN_\ib\vecN^\ia\nabla_\ia\vecM^\ib
={}& 0 +\vecN_\ib\vecM^\ia\nabla_\ia\vecN^\ib +\vecN_\ib[\vecN,\vecM]^\ib\nonumber\\
={}&\frac12\vecM^\ia\nabla_\ia(\vecN_\ib\vecN^\ib) +0 
=0 
\label{eq:llnablam}. 
\end{align}
\end{subequations}
Observe that $\cos\phi\partial_\theta+\frac{\sin\phi}{\sin\theta}(\partial_\phi-a\sin^2\theta\partial_v)$ and $\sin\phi\partial_\theta-\frac{\cos\phi}{\sin\theta}(\partial_\phi-a\sin^2\theta\partial_v)$ form a basis for the planes they span, and that this combination extends smoothly to $\theta=0$ and to $\theta=\pi$. Thus, the results extend from the domain of the spherical coordinates to all of $U$. 

From the chain rule, one finds $\partial_{\vSfc}= \partial_v +\frac{\partial \rSfc}{\partial \vSfc}\partial_r$. From this and the fact that $\partial_r$ is null, it follows that $\metric(\partial_r,\partial_{\vSfc}) = \metric(\partial_r,\partial_{v})$, which is equal to $-1$ by the \NLnhCondition. This proves the first equation of \eqref{eq:surfaceInnerProducts}. Replacing $\vSfc$ by $\thetaSfc$ and $\phiSfc$, one obtains the remaining two equations. 
\end{proof}

\begin{proof}[Proof of the enforceability of the \NLnhCondition, theorem \ref{thm:NLnhEnforceability}]
To begin we construct the gauge transformation. In this paragraph $(v,r,\omega)$ denotes the original parameterization in $V$. On $h(X)$, define $(\vSfc,\omegaSfc)$ and $\rSfc$ to be the restrictions of $(v,\omega)$ and $r$ respectively. By the closeness (in $C^0$) of $\metric$ to $\metricBackground$, at each point $p\in h(X)$ in the domain of the spherical coordinates, there is a unique vector $\vecN$ in $T_p W$ such that $\vecN$ is null and satisfies the analogue of \eqref{eq:surfaceInnerProducts}, i.e.  
\begin{subequations}
\label{eq:surfaceInnerProducts:imposed}
\begin{align}
\metric(\vecN,\partial_{\vSfc}) ={}& 1 
\label{eq:surfaceInnerProducts:imposed:v},\\
\metric(\vecN,\partial_{\thetaSfc}) ={}& 0, \\
\metric(\vecN,\partial_{\phiSfc}) ={}& a\sin^2\theta .
\end{align}
\end{subequations}
Furthermore, since $a\sin^2\theta\di\phi$ extends smoothly to $0$ in $T\Sphere$, the vector field $\vecN$ has a unique continuous extension from the portion of $\heightInNullGaugeEnforceability(X)$ covered by spherical coordinates to all of $\heightInNullGaugeEnforceability(X)$. To avoid overloading notation, let $\vecN$ denote this extension. 
By the $C^2$ closeness of $\metric$ and $\metricBackground$, $\vecN$ is $C^2$ on $h(X)$, and there is an $\varepsilon_0>0$ and an open neighbourhood $W$ of $h(X)$ such that the geodesic flow defines a diffeomorphism $(-\varepsilon_0,\varepsilon_0)\times h(Y)\rightarrow U$. At $q\in W$, define $(\vNew,\omegaNew)$ to be the value of $(\vSfc,\omegaSfc)$ at the unique point $p\in h(X)$ such that $q$ is on the geodesic launched by $\vecN$ at $p$. (The diffeomorphism guarantees the existence of such a point.) Let $\tilde{\gamma}_{(\vNew,\omegaNew)}(s)$ denote the geodesic corresponding to the values $(\vNew,\omegaNew)$ with, on $h(X)$, the initial conditions $s=r$ and $\frac{\di}{\di s}\tilde{\gamma} =-\vecN$. Set $\rNew=s$. Thus, $(\vNew,\rNew,\omegaNew)$ is a gauge choice. 
In this parameterization, $-\partial_{\rNew}$ is null, since it is the tangent to a geodesic launched from a null vector. 

It remains to show the \NLnhCondition{} holds in this diffeomorphism gauge. For the remainder of this proof $(v,r,\omega)$ denotes the parameters in the new parameterization. In the domain of the spherical coordinates, the form $\lambda =i_\vecN\metric =-\metric_{\ia\ib}(\partial_r^\ib)\di x^\ia$ can be expanded, in $\di v$, $\di r$, $\di\theta$, and $\di\phi$. It is sufficient to show that $\metric(\partial_r,\partial_r)$ $=0$ $=\metric(\partial_r,\vecH)$ $=\metric(\partial_r,\vecP)$ and $\metric(-\partial_r,\partial_v)=1$. 
Since $\partial_r$ is null, clearly $\metric(\partial_r,\partial_r)=0$. 
From \eqref{eq:surfaceInnerProducts:imposed}, $\metric(\partial_r,\vecH)$ has the desired value on $h(X)$. Let $\vecN$ denote $-\partial_r$. Observe that since $\partial_r$ is tangent to an affinely parameterized geodesic, $\nabla_{\partial_r}\partial_r=0$. Observe further that $[\vecN,\vecH]$ $=[-\partial_r,\partial_\theta]$ $=0$. Thus, 
\begin{align}
0
={}&\metric(\nabla_{\vecN}\vecN,\vecH) \nonumber\\
={}& \nabla_{\vecN}\left(\metric(\vecN,\vecH)\right)
-\metric(\vecN,\nabla_{\vecN}\vecH) \nonumber\\
={}& \nabla_{\vecN}\left(\metric(\vecN,\vecH)\right)
-\metric(\vecN,\nabla_{\vecH}\vecN) \nonumber\\
={}& \nabla_{\vecN}\left(\metric(\vecN,\vecH)\right)
-\frac12\nabla_{\vecH}\left(\metric(\vecN,\vecN)\right) .
\end{align}
The final term vanishes since $\vecN$ is always a null vector. Thus, $\metric(\vecN,\vecH)$ is constant, and, in particular, since it is initially zero, it remains zero along the entire geodesic. Since $[\vecN,\vecP]$ $=[-\partial_r,\partial_\phi-a\sin^2\theta\partial_v]$ $=0$, the same argument applies with $\vecP$. Since $\metric(\partial_r,\partial_\theta)=0$, the $\di\theta$ component of $\lambda$ vanishes. Since $\metric(\partial_r,\vecP)=0$, the $\di\phi$ component of $\lambda$ is $a\sin^2\theta$ times the coefficient of $\di v$. Since $[\vecN,\partial_v]=0$, a similar calculation shows that $\metric(\vecN,\partial_v)$ is constantly $-1$. Since the parameterization is constructed smoothly, the construction extends from the domain of the spherical coordinates to the full sphere. Since the Kerr metric is itself a solution, from the continuity of solutions of ODE, it follows that for any $V\subset U$, if the initial data is sufficiently close (in a sufficiently high regularity class), the gauge transformation maps $V$ to a subset of $U$. Observe that the new metric on the initial hypersurface $h(X)$ depends only on the old metric on $h(X)$, which gives the desired norm property. This completes the proof. 
\end{proof}

\section{Field equations}
\label{s:FOSH}

Within this section, we introduce geometric variables and a frame gauge condition, which are used to construct a first-order symmetric-hyperbolic system. 

\subsection{GHP Notation}
In this subsection, we review the GHP notation \cite{GHP} for connection and curvature components, which we will use throughout this paper. Appendix \ref{sec:worry} explains the nature of GHP scalars and recalls the definitions of tetrads and properly weighted scalars. All calculations for this paper were done using the \emph{xAct} suite for Mathematica \cite{xAct}, and in particular the \emph{SpinFrames} package \cite{SpinframesPackage}. 

\begin{definition}
\label{def:spincoeff}
Given any null tetrad $(l^a, n^a, m^a, \bar{m}^a)$ and the Levi-Civita connection $\nabla_a$ with respect to the corresponding metric, the spin coefficients are
\begin{subequations}
\begin{align}
\kappa ={}&l^{a} m^{b} \nabla_{a}l_{b},&
\kappa '={}&\bar{m}^{a} n^{b} \nabla_{b}n_{a},\\
\rho ={}&m^{a} \bar{m}^{b} \nabla_{b}l_{a},&
\rho '={}&m^{a} \bar{m}^{b} \nabla_{a}n_{b},\\
\sigma ={}&m^{a} m^{b} \nabla_{a}l_{b},&
\sigma '={}&\bar{m}^{a} \bar{m}^{b} \nabla_{b}n_{a},\\
\tau ={}&m^{a} n^{b} \nabla_{b}l_{a},&
\tau '={}&l^{a} \bar{m}^{b} \nabla_{a}n_{b},
\end{align}
\end{subequations}
and 
\begin{subequations}
\begin{align}
\beta ={}&- \tfrac{1}{2} m^{a} \bar{m}^{b} \nabla_{a}m_{b}
 -  \tfrac{1}{2} l^{a} m^{b} \nabla_{b}n_{a},&
\beta '={}&\tfrac{1}{2} \bar{m}^{a} \bar{m}^{b} \nabla_{b}m_{a}
 + \tfrac{1}{2} l^{a} \bar{m}^{b} \nabla_{b}n_{a},\\
\epsilon ={}&- \tfrac{1}{2} l^{a} \bar{m}^{b} \nabla_{a}m_{b}
 -  \tfrac{1}{2} l^{a} l^{b} \nabla_{b}n_{a},&
\epsilon '={}&\tfrac{1}{2} \bar{m}^{a} n^{b} \nabla_{b}m_{a}
 + \tfrac{1}{2} l^{a} n^{b} \nabla_{b}n_{a}.
\end{align}
\end{subequations}
\end{definition}

\begin{definition}
Given any null tetrad $(l^a, n^a, m^a, \bar{m}^a)$ and the Weyl tensor $C_{abcd}$ with respect to the corresponding metric, we define the Weyl scalars 
\begin{subequations}
\begin{align}
\Psi_{0}{}={}&l^{a} l^{c} m^{b} m^{d} C_{abcd},&
\Psi_{1}{}={}&l^{a} l^{c} m^{b} n^{d} C_{abcd},&
\Psi_{2}{}={}&l^{a} m^{b} \bar{m}^{c} n^{d} C_{abcd},\\
\Psi_{3}{}={}&l^{a} \bar{m}^{c} n^{b} n^{d} C_{abcd},&
\Psi_{4}{}={}&\bar{m}^{a} \bar{m}^{c} n^{b} n^{d} C_{abcd}.
\end{align}
\end{subequations}
\end{definition}
One of the central results of the GHP framework is that $\kappa,\tau,\rho,\sigma,\kappa',\tau',\rho',\sigma'$ and all the $\Psi_i$ are properly weighted, but $\beta,\epsilon,\beta',\epsilon'$ are not.

\subsection{Background and foreground metrics}

To begin our analysis of perturbations of the Kerr metric, we introduce the following hypotheses, which we typically use throughout the rest of this section. 

\begin{definition}[The \radiationGaugeHypothesesNoRef]
\label{def:radiationGaugeHypotheses}
The \defn{\backgroundHypothesesNoRef} are defined as follows: 
``Let $M>0$ and $a\in(-M,M)$. 
Let $\mathring{g}_{ab}$ be the background Kerr metric as in equation \eqref{eq:KerrMetric} with parameters $(M, a)$. 
Let $U$ be an open subset of $\KerrStarDomain$. 
Let $(\vecLBackground^{a},\vecNBackground^{a},\vecMBackground^{a},\vecMbBackground^{a})$ denote an arbitrary element of the set of local complex null tetrads such that $\vecLBackground$ and $\vecNBackground$ are outgoing and ingoing, future-directed principal null vectors. 
Let $(\vecLBackground_{a},\vecNBackground_{a},\vecMBackground_{a},\vecMbBackground_{a})$ be the corresponding co-frame. 
The spin coefficients and Weyl scalars with respect to this tetrad are indicated with the accent $\mathring{}$.'' 

The \defn{\radiationGaugeHypothesesNoRef} are defined to be the \backgroundHypothesesNoRef{} 
together with the assumption that $g_{ab}$ is a Lorentzian metric satisfying the vacuum Einstein equation and the \NLnhCondition{} 
\begin{align}
\vecNBackground^{b} (g_{ab} -  \mathring{g}_{ab})={}&0.
\end{align}
The \defn{background and foreground metrics} are defined to be $\metricBackground_{ab}$ and $\metric_{ab}$ respectively with inverses $\metricBackground{}^{ab}$ and $g^{\#}{}^{ab}$. 
\end{definition}

Because $(\vecLBackground,\vecNBackground,\vecMBackground,\vecMbBackground)$ is used to denote an arbitrary element of the set of local tetrads in Kerr aligned with $(\vecLSpecified,\vecNSpecified)$, there is a freedom to apply spin and boost transformations. As long as our variables and operators are made so that they transform properly under such transformation, this allows us to introduce properly weighted quantities, which are globally defined. In the language of principal-$G$ bundles, as long as our variables transform equivariantly, we may use local tetrads to construct a globally defined section of an associated complex line bundle. In the language of gauge theory, we have a gauge freedom corresponding to choice of boost and spin transformation, and, as long as our variables transform correctly under such gauge transformations, they are globally defined gauge fields. This allows us to avoid problems at the poles in spherical coordinates that might arise from, for example, taking $\vecMBackground=2^{-1/2}(r-ia\cos\theta)^{-1}(\partial_\theta+i(\sin\theta)^{-1}(\partial_\phi+a\sin^2\theta\partial_v))$ or any other explicit combination of $\vecH$ and $\vecP$.

\begin{definition}[Foreground metric coefficients in the background frame]
\assumeRadiationHypotheses{}. 
Define the \defn{foreground metric coefficients in the background frame} to be 
\begin{subequations}
\begin{align}
G_{2}={}&\vecMbBackground^{a} \vecMbBackground^{b}(g_{ab} - \mathring{g}_{ab})
= \vecMbBackground^{a} \vecMbBackground^{b}g_{ab},\\
G_{1}={}& \mathring{l}^{a} \vecMbBackground^{b}(g_{ab} - \mathring{g}_{ab})
= \mathring{l}^{a} \vecMbBackground^{b}g_{ab},\\
G_{0}={}& \mathring{l}^{a} \mathring{l}^{b}(g_{ab} - \mathring{g}_{ab})
= \mathring{l}^{a} \mathring{l}^{b}g_{ab},\\
\slashed{G}_{}={}&\mathring{g}^{ab} (g_{ab} - \mathring{g}_{ab})
=\mathring{g}^{ab} g_{ab} - 4,\\
G^{\#}_{2}={}& \vecMbBackground_{a} \vecMbBackground_{b}(g^{\#}{}^{ab} - \mathring{g}^{ab})
= \vecMbBackground_{a} \vecMbBackground_{b} g^{\#}{}^{ab},\\
G^{\#}_{1}={}&\mathring{l}_{a} \vecMbBackground_{b}(g^{\#}{}^{ab} - \mathring{g}^{ab}) 
= \mathring{l}_{a} \vecMbBackground_{b} g^{\#}{}^{ab},\\
G^{\#}_{0}={}&\mathring{l}_{a} \mathring{l}_{b}(g^{\#}{}^{ab} - \mathring{g}^{ab})
= \mathring{l}_{a} \mathring{l}_{b} g^{\#}{}^{ab},\\
\slashed{G}^{\#}={}&\mathring{g}_{ab}(g^{\#}{}^{ab} - \mathring{g}^{ab})
=\mathring{g}_{ab} g^{\#}{}^{ab} - 4.
\end{align}
\end{subequations}
\end{definition}
Observe that they vanish if the perturbation vanishes. They are all properly weighted with respect to background boost and spin transformations. The remaining metric coefficients vanish by the \NLnhCondition. The set $(G^{\#}_{2}, G^{\#}_{1}, G^{\#}_{0}, \slashed{G}^{\#})$ can be algebraically computed from the set $(G_{2}, G_{1}, G_{0}, \slashed{G}_{})$ and vice versa via
\begin{subequations}
\label{eq:GToInvG}
\begin{align}
G^{\#}_{2}={}&- \frac{G_{2}}{(1 + \tfrac{1}{2} \slashed{G}_{})^2 - |G_{2}|^2},\\
G^{\#}_{1}={}&- \frac{(1 + \tfrac{1}{2} \slashed{G}_{}) G_{1} + \overline{G_{1}} G_{2}}{(1 + \tfrac{1}{2} \slashed{G}_{})^2 - |G_{2}|^2},\\
G^{\#}_{0}={}&- G_{0}
 -  \frac{2 (1 + \tfrac{1}{2} \slashed{G}_{}) \overline{G_{1}} G_{1} + \overline{G_{2}} G_{1}{}^2 + \overline{G_{1}}{}^2 G_{2}}{(1 + \tfrac{1}{2} \slashed{G}_{})^2 - |G_{2}|^2},\\
\slashed{G}^{\#}={}&
 \frac{1}{1 -  |G_{2}| + \tfrac{1}{2} \slashed{G}_{}}
 + \frac{1}{1 + |G_{2}| + \tfrac{1}{2} \slashed{G}_{}}
 - 2,
\end{align}
\end{subequations}
\begin{subequations}
\label{eq:InvGToG}
\begin{align}
G_{2}={}&- \frac{G^{\#}_{2}}{(1 + \tfrac{1}{2} \slashed{G}^{\#})^2- |G^{\#}_{2}|^2},\\
G_{1}={}&- \frac{(1 + \tfrac{1}{2} \slashed{G}^{\#}) G^{\#}_{1} + \overline{G^{\#}_{1}} G^{\#}_{2}}{(1 + \tfrac{1}{2} \slashed{G}^{\#})^2 - |G^{\#}_{2}|^2},\\
G_{0}={}&- G^{\#}_{0}
 -  \frac{2 (1 + \tfrac{1}{2} \slashed{G}^{\#}) \overline{G^{\#}_{1}} G^{\#}_{1} + \overline{G^{\#}_{2}} G^{\#}_{1}{}^2 + \overline{G^{\#}_{1}}{}^2 G^{\#}_{2}}{(1 + \tfrac{1}{2} \slashed{G}^{\#})^2- |G^{\#}_{2}|^2},\\
\slashed{G}_{}={}&
 \frac{1}{1 -  |G^{\#}_{2}| + \tfrac{1}{2} \slashed{G}^{\#}}
 + \frac{1}{1 + |G^{\#}_{2}| + \tfrac{1}{2} \slashed{G}^{\#}} 
 -2.
\end{align}
\end{subequations}

\subsection{Frame choice}
Given the set $(G^{\#}_{2}, G^{\#}_{1}, G^{\#}_{0}, \slashed{G}^{\#})$ of background frame components of the inverse foreground metric $g^{\#}{}^{ab}$, we can construct a null tetrad for the foreground metric. However, due to Lorentz gauge freedom this frame is not unique. Due to the fact that the \NLnhCondition{} singles out $\mathring{n}^{a}$, we choose to use it also in the foreground tetrad, i.e. $n^a=\mathring{n}^{a}$.
With this leg fixed, the remaining group of Lorentz transformations are described by one real differential spin rotation parameter $\nu$ with $(p,q)$-weight $(0,0)$ and a complex parameter $\eta$ with $(p,q)$-weight $(2,0)$. 
\begin{remark}
In principle one could instead demand that $n^a$ is merely proportional to $\mathring{n}^{a}$. Doing this would introduce a real differential boost parameter $\mu$ to the group of Lorentz transformations, so that $n^{a}=\mu^{-1}\mathring{n}^{a}$. However, as we later would like to set the  $\tilde{\epsilon}'=\epsilon ' -  \mu^{-1}\mathring{\epsilon}'$ to zero, and we find that $\tilde{\epsilon}' + \overline{\tilde{\epsilon}'}=- \mu^{-1}\thop \mu$, we conclude that $\mu=1$, i.e. $n^a=\mathring{n}^{a}$ is sensible. 
\end{remark}

\begin{definition}[Foreground frame]
\label{def:ForegroundFrame}
\assumeRadiationHypotheses. 

A choice of \defn{differential Lorentz transformation variables} is a choice of 
$(\nu,\eta)$ with $(p,q)$-weights $(0,0)$ and $(2,0)$ respectively. 

Assuming a choice of differential Lorentz transformation variables, 
define the \defn{foreground frame} to be
\begin{subequations}
\begin{align}
l^{a}={}&\mathring{l}^{a}
 + (\eta \bar{\eta} + \tfrac{1}{2} G^{\#}_{0}) \mathring{n}^{a}
 - \Bigl( G^{\#}_{1} +  \frac{\eta G^{\#}_{2}}{2 e^{i \nu} \varsigma^{\#}{}} - e^{i \nu} \bar{\eta} \varsigma^{\#}{}\Bigr) \vecMBackground^{a}
 - \Bigl( \overline{G^{\#}_{1}} +  \frac{e^{i \nu} \bar{\eta} \overline{G^{\#}_{2}}}{2 \varsigma^{\#}{}} - \frac{\eta \varsigma^{\#}{}}{e^{i \nu}}\Bigr) \vecMbBackground^{a} ,\\
n^{a}={}&\mathring{n}^{a},\\
m^{a}={}&\eta \mathring{n}^{a}
+e^{i \nu} \varsigma^{\#}{} \vecMBackground^{a}
 -  \frac{e^{i \nu} \overline{G^{\#}_{2}} \vecMbBackground^{a}}{2 \varsigma^{\#}{}},
\end{align}
\end{subequations}
along with the auxiliary variables
\begin{subequations}
\label{eq:defVarsigma}
\begin{align}
\varsigma ={}&\tfrac{1}{2} \sqrt{1 -  |G_{2}| + \tfrac{1}{2} \slashed{G}_{}}
 + \tfrac{1}{2} \sqrt{1 + |G_{2}| + \tfrac{1}{2} \slashed{G}_{}},\\
\varsigma^{\#}{}={}&\tfrac{1}{2} \sqrt{1 - |G^{\#}_{2}| + \tfrac{1}{2}\slashed{G}^{\#}}
 + \tfrac{1}{2} \sqrt{1 + |G^{\#}_{2}| + \tfrac{1}{2}\slashed{G}^{\#}}.
\end{align}
\end{subequations}
\end{definition}

\begin{lemma}
\assumeRadiationHypothesesAndLorentzTransformVariables. 

The foreground frame is a null tetrad for the foreground metric $g_{ab}$, i.e.
\begin{align}
g^{\#}{}^{ab}={}&2 l^{(a}n^{b)} - 2 m^{(a}\bar{m}^{b)}.
\end{align}
The corresponding co-frame is
\begin{subequations}
\begin{align}
l_{a}={}&\mathring{l}_{a}
 + \Bigl(\eta \bar{\eta} -  \tfrac{1}{2} G^{\#}_{0} -  e^{i \nu} \bar{\eta} \overline{G^{\#}_{1}} \varsigma -  \frac{\eta G^{\#}_{1} \varsigma}{e^{i \nu}} -  \frac{e^{i \nu} \bar{\eta} \overline{G^{\#}_{2}} G^{\#}_{1} \varsigma}{2 \varsigma^{\#}{}^2} -  \frac{\eta \overline{G^{\#}_{1}} G^{\#}_{2} \varsigma}{2 e^{i \nu} \varsigma^{\#}{}^2}\Bigr)\mathring{n}_{a}\nonumber\\
&  +  \Bigl(e^{i \nu} \bar{\eta} \varsigma +  \frac{\eta G^{\#}_{2} \varsigma}{2 e^{i \nu} \varsigma^{\#}{}^2}\Bigr)\vecMBackground_{a}
+ \Bigl( \frac{\eta \varsigma}{e^{i \nu}} +  \frac{e^{i \nu} \bar{\eta} \overline{G^{\#}_{2}} \varsigma}{2 \varsigma^{\#}{}^2}\Bigr)\vecMbBackground_{a},\\
n_{a}={}&\mathring{n}_{a},\\
m_{a}={}& \Bigl( \eta - e^{i \nu} \overline{G^{\#}_{1}} \varsigma - \frac{e^{i \nu} \overline{G^{\#}_{2}} G^{\#}_{1} \varsigma}{2 \varsigma^{\#}{}^2}\Bigr)\mathring{n}_{a}
 + e^{i \nu}  \varsigma \vecMBackground_{a}
 + \frac{e^{i \nu} \overline{G^{\#}_{2}} \varsigma}{2 \varsigma^{\#}{}^2}\vecMbBackground_{a} .
\end{align}
\end{subequations}
\end{lemma}

\begin{definition}
\assumeRadiationHypothesesAndLorentzTransformVariables. 

Define the \defn{foreground metric coefficients} to be 
\begin{subequations}
\begin{align}
\tilde{\slashed{G}}={}&\mathring{g}_{ab} g^{\#}{}^{ab} - 4,&
\tilde{G}_{2}{}={}&\mathring{g}_{ab} \bar{m}^{a} \bar{m}^{b},&
\tilde{G}_{1}{}={}&\mathring{g}_{ab} l^{a} \bar{m}^{b},&
\tilde{G}_{0}{}={}&\mathring{g}_{ab} l^{a} l^{b},\\
\InvGTrTilde={}&\mathring{g}^{ab} g_{ab} - 4,&
\tilde{G}^{\#}_2{}={}&\mathring{g}^{ab} \bar{m}_{a} \bar{m}_{b},&
\tilde{G}^{\#}_1{}={}&\mathring{g}^{ab} l_{a} \bar{m}_{b},&
\tilde{G}^{\#}_0{}={}&\mathring{g}^{ab} l_{a} l_{b}.
\end{align}
\end{subequations}
Unless otherwise specified, define \defn{metric coefficients} to be the foreground metric coefficients. 
\end{definition}

Note that the background metric coefficients are the components of the foreground metric with respect to the background tetrad, and, conversely, the (foreground) metric coefficients are the components of the background metric with respect to the foreground frame. 

We have the following useful relations 
\begin{subequations}
\label{eq:GslashRelations}
\begin{align}
\slashed{G}{}^{\#}={}&\tilde{\slashed{G}}
=-2 + 4 \varsigma^{\#}{}^2 -  \frac{2 \varsigma^{\#}{}}{\varsigma},&
| G^{\#}_{2}|^2={}&|\tilde{G}_{2}{}|^2
=4 \varsigma^{\#}{}^4 -  \frac{4 \varsigma^{\#}{}^3}{\varsigma},\\
\slashed{G}_{}={}&\InvGTrTilde
=-2 + 4 \varsigma^2 -  \frac{2 \varsigma}{\varsigma^{\#}{}},&
|G_{2}|^2={}&|\tilde{G}^{\#}_2{}|^2
=4 \varsigma^4 - \frac{4 \varsigma^3}{\varsigma^{\#}{}}.
\end{align}
\end{subequations}
The relations between $(\tilde{\slashed{G}}, \tilde{G}_{2}{}, \tilde{G}_{1}{}, \tilde{G}_{0}{})$, $(\InvGTrTilde\negthinspace, \tilde{G}^{\#}_2{}, \tilde{G}^{\#}_1{}, \tilde{G}^{\#}_0{})$ follows the pattern \eqref{eq:GToInvG}. 
Given $\nu$ and $\eta$ we can express the sets $(\tilde{\slashed{G}}, \tilde{G}_{2}{}, \tilde{G}_{1}{}, \tilde{G}_{0}{})$, $(\InvGTrTilde\negthinspace, \tilde{G}^{\#}_2{}, \tilde{G}^{\#}_1{}, \tilde{G}^{\#}_0{})$, $(G^{\#}_{2}, G^{\#}_{1}, G^{\#}_{0}, \slashed{G}^{\#}\hspace{-1pt})$ and $(G_{2}, G_{1}, G_{0}, \slashed{G}_{})$ in terms of each other. For instance 
\begin{subequations}
\label{eq:InvGTildeInvGRelations}
\begin{align}
\tilde{G}^{\#}_2{}={}&\frac{G_{2}}{e^{2i \nu}} ,\\
\tilde{G}^{\#}_1{}={}&
  \frac{G_{1} \varsigma^{\#}{}}{e^{i \nu}}
 - \tfrac{1}{2} \bar{\eta} \slashed{G}_{}
 + \frac{\eta G_{2}}{e^{2i \nu}}
 + \frac{\overline{G_{1}} G_{2} \varsigma^{\#}{}}{2 e^{i \nu} \varsigma^2} ,\\
\tilde{G}^{\#}_0{}={}&
 G_{0}
 - \eta \bar{\eta} \slashed{G}_{}
 + e^{2i \nu} \bar{\eta}^2 \overline{G_{2}}
 + \frac{\eta^2 G_{2}}{e^{2i \nu}}
 + e^{i \nu} \bar{\eta} \varsigma^{\#}{} \Bigl(2 \overline{G_{1}} + \frac{\overline{G_{2}} G_{1}}{\varsigma^2}\Bigr)
 + \frac{\eta \varsigma^{\#}{}}{e^{i \nu}} \Bigl(2 G_{1} + \frac{\overline{G_{1}} G_{2}}{\varsigma^2}\Bigr)\nonumber\\
& + \frac{\varsigma^{\#}{}^2}{\varsigma^2} \bigl((2 + \slashed{G}_{}) \overline{G_{1}} G_{1} + \overline{G_{2}} G_{1}{}^2 + \overline{G_{1}}{}^2 G_{2}\bigr),\\
G^{\#}_{2}={}&- \frac{e^{2i \nu} \tilde{G}^{\#}_2{} \varsigma^{\#}{}^2}{\varsigma^2},\\
G^{\#}_{1}={}&- e^{i \nu} \tilde{G}^{\#}_1{} \varsigma^{\#}{}
 + e^{i \nu} \bar{\eta} (\varsigma^{\#}{} -  \varsigma)
 -  \frac{e^{i \nu} \overline{\tilde{G}}{}^{\#}_1{} \tilde{G}^{\#}_2{} \varsigma^{\#}{}}{2 \varsigma^2}
 + \frac{e^{i \nu} \eta \tilde{G}^{\#}_2{} (\varsigma^{\#}{} + \varsigma)}{2 \varsigma^2},\\
G^{\#}_{0}={}&- \tilde{G}^{\#}_0{}
 + 2 \eta \tilde{G}^{\#}_1{}
 + 2 \bar{\eta} \overline{\tilde{G}}{}^{\#}_1{}
 -  \eta^2 \tilde{G}^{\#}_2{}
 -  \bar{\eta}^2 \overline{\tilde{G}}{}^{\#}_2{}
 + \eta \bar{\eta} \InvGTrTilde.
\end{align}
\end{subequations}

\subsection{Geometric variables and operators}
In this section, we define differential spin coefficients and differential curvature components. The foreground spin coefficients carry all the information about the connection. However, several are not small for a small metric perturbation, because several of the background components are non-vanishing. Furthermore, not all of them are properly weighted with respect to spin and boost transformations of the background frame. Our choice of differential spin coefficients compensate for both of these issues. While the foreground curvature components are properly weighted, the middle curvature component is not small, since the middle curvature component in the background is non-vanishing. Our choice of differential curvature components compensates for this problem. 

\begin{definition}
\label{def:DiffSpinCoeff}
\label{def:geometricVariables}
\assumeRadiationHypothesesAndLorentzTransformVariables. 

Define the \defn{differential spin coefficients}
\begin{subequations}
\begin{align}
\tilde{\beta}={}&\beta
- e^{i \nu} \varsigma^{\#}{} \mathring{\beta}
 -  \frac{e^{i \nu} \overline{G^{\#}_{2}} \mathring{\beta}'}{2 \varsigma^{\#}{}}
 + \eta \mathring{\epsilon}',\\
\tilde{\beta}'={}&\beta '
- \frac{G^{\#}_{2} \mathring{\beta}}{2 e^{i \nu} \varsigma^{\#}{}}
 -  \frac{\varsigma^{\#}{} \mathring{\beta}'}{e^{i \nu}}
 -  \bar{\eta} \mathring{\epsilon}',\\
\tilde{\epsilon}={}&\epsilon -  \mathring{\epsilon}
 + \Bigl(G^{\#}_{1} + \frac{\eta G^{\#}_{2}}{2 e^{i \nu} \varsigma^{\#}{}} -  e^{i \nu} \bar{\eta} \varsigma^{\#}{}\Bigr) \mathring{\beta}
 - \Bigl( \overline{G^{\#}_{1}} +  \frac{e^{i \nu} \bar{\eta} \overline{G^{\#}_{2}}}{2 \varsigma^{\#}{}} - \frac{\eta \varsigma^{\#}{}}{e^{i \nu}}\Bigr) \mathring{\beta}'
 + (\eta \bar{\eta} + \tfrac{1}{2} G^{\#}_{0}) \mathring{\epsilon}',\\
\tilde{\epsilon}'={}&\epsilon ' - \mathring{\epsilon}',\\
\tilde{\kappa}={}&\kappa ,\\
\tilde{\kappa}'={}&\kappa ',\\
\tilde{\rho}={}&\rho - \mathring{\rho},\\
\tilde{\rho}'={}&\rho ' - \mathring{\rho}',\\
\tilde{\sigma}={}&\sigma ,\\
\tilde{\sigma}'={}&\sigma ',\\
\tilde{\tau}={}&\tau - \mathring{\tau},\\
\tilde{\tau}'={}&\tau ' - \mathring{\tau}'.
\end{align}
\end{subequations}
Define the \defn{differential curvature coefficients} as
\begin{align}
\tilde{\Psi}_{0}{}={}&\Psi_{0}{},&
\tilde{\Psi}_{1}{}={}&\Psi_{1}{},&
\tilde{\Psi}_{2}{}={}&\Psi_{2}{} - \mathring{\Psi}_{2}{},&
\tilde{\Psi}_{3}{}={}&\Psi_{3}{},&
\tilde{\Psi}_{4}{}={}&\Psi_{4}{}.
\end{align}
The \defn{geometric variables} are defined to be
\begin{align*}
\geometricVariables=&{}(
\eta,\nu,
\tilde{G}^{\#}_2{}, \InvGTrTilde, \tilde{G}^{\#}_1{}, \tilde{G}^{\#}_0{}, 
\tilde{\sigma}', \tilde{\rho}', \tilde{\tau}', \tilde{\beta},\tilde{\beta}', \tilde{\epsilon}, \tilde{\rho}, \tilde{\sigma}, \tilde{\kappa}, 
\tilde{\Psi}_{0},\tilde{\Psi}_{1},\tilde{\Psi}_{2},\tilde{\Psi}_{3},\tilde{\Psi}_{4}
)^T.
\end{align*}
\end{definition}

The differential variables are chosen so that they are properly weighted with respect to the background tetrad. This may initially seem surprising, since $\mathring{\beta},\mathring{\epsilon},\mathring{\beta}',\mathring{\epsilon}'$ are not. It may be helpful to recall this is similar to the fact that the Christoffel symbols for a connection do not transform as a tensor, but the difference between the Christoffel symbols for two different connections does transform as a tensor. 

This choice of variables is not unique, and not all of them are properly weighted under differential Lorentz transformations. However, they are the simplest choices of variables that are properly weighted under spin and boost transformations of the background tetrad. We are going to use the differential Lorentz transformations to eliminate some of the differential spin coefficients. This would have been impossible if they were properly weighted under the differential Lorentz transformations. 

\begin{definition}
\assumeRadiationHypothesesAndLorentzTransformVariables. 

Define the \defn{foreground GHP operators} acting on a $(p,q)$-weighted scalar $\varphi$ to be 
\begin{subequations}
\begin{align}
\tho \varphi ={}&\thoBG \varphi
 + (\eta \bar{\eta}
 + \tfrac{1}{2} G^{\#}_{0}) \thopBG \varphi
 - \Bigl(G^{\#}_{1}
 +  \frac{\eta G^{\#}_{2}}{2 e^{i \nu} \varsigma^{\#}{}}
 - e^{i \nu} \bar{\eta} \varsigma^{\#}{}\Bigr) \edtBG \varphi
 - \Bigl(\overline{G^{\#}_{1}}
 +  \frac{e^{i \nu} \bar{\eta} \overline{G^{\#}_{2}}}{2 \varsigma^{\#}{}}
 - \frac{\eta \varsigma^{\#}{}}{e^{i \nu}}\Bigr) \edtpBG \varphi\nonumber\\
& -  p \tilde{\epsilon} \varphi
 -  q \overline{\tilde{\epsilon}} \varphi ,\\
\thop \varphi ={}&\thopBG \varphi
 + p \tilde{\epsilon}' \varphi
 + q \overline{\tilde{\epsilon}'} \varphi ,\\
\edt \varphi ={}&\eta \thopBG \varphi
 + e^{i \nu} \varsigma^{\#}{} \edtBG \varphi
 -  \frac{e^{i \nu} \overline{G^{\#}_{2}} \edtpBG \varphi}{2 \varsigma^{\#}{}}
 -  p \tilde{\beta} \varphi
 + q \overline{\tilde{\beta}'} \varphi ,\\
\edtp \varphi ={}&\bar{\eta} \thopBG \varphi
 -  \frac{G^{\#}_{2} \edtBG \varphi}{2 e^{i \nu} \varsigma^{\#}{}}
 + \frac{\varsigma^{\#}{} \edtpBG \varphi}{e^{i \nu}}
 -  q \overline{\tilde{\beta}} \varphi
 + p \tilde{\beta}' \varphi,
\end{align}
\end{subequations}
where $\thoBG$, $\thopBG$, $\edtBG$ and $\edtpBG$ are the classical GHP operators as defined in \cite{GHP} with respect to the background tetrad.
\end{definition}

\begin{remark}
Observe that we define weight to be with respect to the background tetrad. Any background spin and boost transformation will induce the same spin and boost transformation on the foreground tetrad. Hence, any quantity which is properly weighted with respect to the foreground tetrad will become properly weighted with the same weights with respect to the background tetrad, when we have tied the frames together as in definition~\ref{def:ForegroundFrame}. 
For any quantity which is properly weighted with respect to the foreground tetrad, our definition corresponds to the classical definition of GHP operators. Our definition can therefore be seen as an extension to quantities which are weighted only in terms of the background tetrad.
\end{remark}
These GHP operators satisfy the commutator relations in appendix \ref{sec:commutators}.

\subsection{Structure equations}
We now choose $\eta$ and $\nu$ so that two differential spin coefficients are eliminated (in addition to $\kappa'$, which vanishes as a result of $\vecN$ being tangent to null geodesics in the \NLnhCondition{}) and so that the remaining connection coefficients satisfy transport equations. 

\begin{definition}[The \frameGaugeHypothesesNoRef]
\label{def:frameGaugeHypotheses}
\assumeRadiationHypotheses. 

The \defn{\frameGaugeHypothesesNoRef} are defined to hold if there is a choice of differential Lorentz transformation variables satisfying
\begin{subequations}
\label{eq:nuetaevoleq}
\begin{align}
\thop \nu ={}&- \frac{i}{2 \varsigma} (\varsigma^{\#}{}
 -  \varsigma) (\mathring{\rho}'
 -  \bar{\mathring{\rho}}')
 + \frac{i \overline{\tilde{G}}{}^{\#}_2{} \tilde{\sigma}'}{4 \varsigma^2}
 -  \frac{i \tilde{G}^{\#}_2{} \overline{\tilde{\sigma}'}}{4 \varsigma^2}
\label{eq:nuevoleq} ,\\
\thop \eta ={}&\tilde{\beta}
 -  \overline{\tilde{\beta}'}
 + \eta (\mathring{\rho}'
 + \tilde{\rho}')
 + \bar{\eta} \overline{\tilde{\sigma}'}
 -  \mathring{\tau}
 + e^{i \nu} \varsigma^{\#}{} \mathring{\tau}
 + \frac{\overline{\tilde{G}}{}^{\#}_2{} \varsigma^{\#}{} \bar{\mathring{\tau}}}{2 e^{i \nu} \varsigma^2}.
\label{eq:etaevoleq}
\end{align}
\end{subequations}
\end{definition}

\begin{lemma}[Structure equations]
\label{lem:structureEquations}
\assumeRadiationAndFrameGaugeHypotheses. 

The structure equations take the form of a transport system for the metric coefficients 
\begin{subequations}
\label{eq:thopmetric}
\begin{align}
(\thop{} + \mathring{\rho}' + 2 \tilde{\rho}' -  \bar{\mathring{\rho}}')\tilde{G}^{\#}_2{}={}&(2 + \InvGTrTilde) \tilde{\sigma}',
\label{eq:ThopInvGTilde2}\\
(\thop{} + \tilde{\rho}' + \overline{\tilde{\rho}'})\InvGTrTilde={}&-2 \tilde{\rho}'
 - 2 \overline{\tilde{\rho}'}
 + 2 \overline{\tilde{G}}{}^{\#}_2{} \tilde{\sigma}'
 + 2 \tilde{G}^{\#}_2{} \overline{\tilde{\sigma}'},
\label{eq:ThopInvGTrTilde}\\
(\thop{} + \tilde{\rho}' + \bar{\mathring{\rho}}')\tilde{G}^{\#}_1{}={}&
 2 \tilde{\tau}'
-\Bigl(\overline{\tilde{G}}{}^{\#}_1{} \tilde{G}^{\#}_2{} \frac{\varsigma^{\#}{}}{\varsigma} + 2 \tilde{G}^{\#}_1{} \varsigma^{\#}{} \varsigma + 2 \bar{\eta} \varsigma^2\Bigr) (\mathring{\rho}' -  \bar{\mathring{\rho}}')
 -  \bar{\eta} (2 + \InvGTrTilde) \bar{\mathring{\rho}}'
  -  \overline{\tilde{G}}{}^{\#}_1{} \tilde{\sigma}'\nonumber\\
& + \eta \tilde{G}^{\#}_2{} (\mathring{\rho}' + \bar{\mathring{\rho}}')
 + \tfrac{1}{2} \InvGTrTilde (\bar{\mathring{\tau}} + \mathring{\tau}' + \tilde{\tau}')
 - \tilde{G}^{\#}_2{} ( \mathring{\tau} +  \overline{\tilde{\tau}'} +  \bar{\mathring{\tau}}')
   +\frac{e^{i \nu} \tilde{G}^{\#}_2{}\bar{\mathring{\tau}}'}{\varsigma}\nonumber\\
& + 2  (1 -  \varsigma e^{-i \nu}) \mathring{\tau}',
\label{eq:ThopInvGTilde1}\\
\thop \tilde{G}^{\#}_0{}={}&-2 \tilde{\epsilon} - 2 \overline{\tilde{\epsilon}}
+ \bigl(\eta^2 \tilde{G}^{\#}_2{} + \bar{\eta}^2 \overline{\tilde{G}}{}^{\#}_2{} -  \eta \bar{\eta} (2 + \InvGTrTilde)\bigr) (\mathring{\rho}' + \bar{\mathring{\rho}}') \nonumber\\
& -  \bigl(2 \eta \overline{\tilde{G}}{}^{\#}_1{} \tilde{G}^{\#}_2{} - 2 \bar{\eta} \tilde{G}^{\#}_1{} \overline{\tilde{G}}{}^{\#}_2{} + (\eta \tilde{G}^{\#}_1{} -  \bar{\eta} \overline{\tilde{G}}{}^{\#}_1{}) (2 + \InvGTrTilde)\bigr) \varsigma^{\#}{} \varsigma^{-1} (\mathring{\rho}' -  \bar{\mathring{\rho}}')\nonumber\\
& + e^{-i \nu} (2 \overline{\tilde{G}}{}^{\#}_1{} \varsigma^{\#}{} + \tilde{G}^{\#}_1{} \overline{\tilde{G}}{}^{\#}_2{} \varsigma^{\#}{} \varsigma^{-2} + \bar{\eta} \overline{\tilde{G}}{}^{\#}_2{} \varsigma^{-1} - 2 \eta \varsigma) \mathring{\tau}'
 - 2 \overline{\tilde{G}}{}^{\#}_1{} (\bar{\mathring{\tau}} + \mathring{\tau}' + \tilde{\tau}')\nonumber\\
& + e^{i \nu} (2 \tilde{G}^{\#}_1{} \varsigma^{\#}{} + \overline{\tilde{G}}{}^{\#}_1{} \tilde{G}^{\#}_2{} \varsigma^{\#}{} \varsigma^{-2} + \eta \tilde{G}^{\#}_2{} \varsigma^{-1} - 2 \bar{\eta} \varsigma) \bar{\mathring{\tau}}'
 - 2 \tilde{G}^{\#}_1{} (\mathring{\tau} + \overline{\tilde{\tau}'} + \bar{\mathring{\tau}}'),
\label{eq:ThopInvGTilde0}
\end{align}
\end{subequations}
algebraic relations for the spin coefficients
\begin{subequations}
\label{eq:AlgebraicSpincoeff1}
\begin{align}
\tilde{\kappa}'={}&0,\qquad
\tilde{\epsilon}'={}0,\qquad
\tilde{\tau}={}0,\\
\tilde{\rho}' -  \overline{\tilde{\rho}'}={}&\frac{(\varsigma^{\#}{} -  \varsigma) (\mathring{\rho}' -  \bar{\mathring{\rho}}')}{\varsigma},
\label{eq:algebraicrho1}\\
\tilde{\beta}-\overline{\tilde{\beta}'}={}&
 -  \frac{(\overline{\tilde{G}}{}^{\#}_1{} \varsigma^{\#}{} -  \bar{\eta} \overline{\tilde{G}}{}^{\#}_2{} \varsigma^{\#}{} -  \eta \varsigma + 2 \eta \varsigma^{\#}{} \varsigma^2) (\mathring{\rho}' -  \bar{\mathring{\rho}}')}{\varsigma}
 + \frac{\overline{\tilde{G}}{}^{\#}_2{} \varsigma^{\#}{} \mathring{\tau}'}{2 e^{i \nu} \varsigma^2}
 -  \overline{\tilde{\tau}'}
 + (e^{i \nu} \varsigma^{\#}{} - 1) \bar{\mathring{\tau}}',
\label{eq:algebraicbeta1}
\end{align}
\end{subequations}
and a supplementary set of equations displayed in \eqref{eq:StructureSpincoeff1}.
\end{lemma} 
\begin{proof}
The foreground Levi-Civita connection $\nabla$ and background $\mathring{\nabla}$ connections are related via
\begin{align}
\widetilde\Gamma^{a}{}_{bc}={}&\tfrac{1}{2} g^{\#}{}^{ad} (\mathring{\nabla}_{b}g_{cd} + \mathring{\nabla}_{c}g_{bd} -  \mathring{\nabla}_{d}g_{bc}).
\end{align}
Definition~\ref{def:spincoeff} lets us express the foreground spin coefficients in terms of the foreground $\nabla$ acting on the foreground tetrad. We can re-express this in terms of the background $\mathring{\nabla}$ as
\begin{subequations}
\begin{align}
\beta ={}&- \tfrac{1}{2} m^{a} \bar{m}^{b} (- \widetilde\Gamma^{c}{}_{ab} m_{c} + \mathring{\nabla}_{a}m_{b})
 -  \tfrac{1}{2} l^{a} m^{b} (- \widetilde\Gamma^{c}{}_{ba} n_{c} + \mathring{\nabla}_{b}n_{a}),\\
\beta '={}&\tfrac{1}{2} \bar{m}^{a} \bar{m}^{b} (- \widetilde\Gamma^{c}{}_{ba} m_{c} + \mathring{\nabla}_{b}m_{a})
 + \tfrac{1}{2} l^{a} \bar{m}^{b} (- \widetilde\Gamma^{c}{}_{ba} n_{c} + \mathring{\nabla}_{b}n_{a}),\\
\epsilon ={}&- \tfrac{1}{2} l^{a} \bar{m}^{b} (- \widetilde\Gamma^{c}{}_{ab} m_{c} + \mathring{\nabla}_{a}m_{b})
 -  \tfrac{1}{2} l^{a} l^{b} (- \widetilde\Gamma^{c}{}_{ba} n_{c} + \mathring{\nabla}_{b}n_{a}),\\
\epsilon '={}&\tfrac{1}{2} \bar{m}^{a} n^{b} (- \widetilde\Gamma^{c}{}_{ba} m_{c} + \mathring{\nabla}_{b}m_{a})
 + \tfrac{1}{2} l^{a} n^{b} (- \widetilde\Gamma^{c}{}_{ba} n_{c} + \mathring{\nabla}_{b}n_{a}),\\
\kappa ={}&l^{a} m^{b} (- \widetilde\Gamma^{c}{}_{ab} l_{c} + \mathring{\nabla}_{a}l_{b}),\\
\kappa '={}&\bar{m}^{a} n^{b} (- \widetilde\Gamma^{c}{}_{ba} n_{c} + \mathring{\nabla}_{b}n_{a}),\\
\rho ={}&m^{a} \bar{m}^{b} (- \widetilde\Gamma^{c}{}_{ba} l_{c} + \mathring{\nabla}_{b}l_{a}),\\
\rho '={}&m^{a} \bar{m}^{b} (- \widetilde\Gamma^{c}{}_{ab} n_{c} + \mathring{\nabla}_{a}n_{b}),\\
\sigma ={}&m^{a} m^{b} (- \widetilde\Gamma^{c}{}_{ab} l_{c} + \mathring{\nabla}_{a}l_{b}),\\
\sigma '={}&\bar{m}^{a} \bar{m}^{b} (- \widetilde\Gamma^{c}{}_{ba} n_{c} + \mathring{\nabla}_{b}n_{a}),\\
\tau ={}&m^{a} n^{b} (- \widetilde\Gamma^{c}{}_{ba} l_{c} + \mathring{\nabla}_{b}l_{a}),\\
\tau '={}&l^{a} \bar{m}^{b} (- \widetilde\Gamma^{c}{}_{ab} n_{c} + \mathring{\nabla}_{a}n_{b}).
\end{align}
\end{subequations}
Using the relation between the background and foreground tetrads, and expressing all background derivatives of background frame components in terms of background spin coefficients, we get expressions of all background tetrad components of $\widetilde\Gamma^{a}{}_{bc}$ in terms of the metric components $(G_{2}, G_{1}, G_{0}, \slashed{G}_{})$ or $(\InvGTrTilde, \tilde{G}^{\#}_2{}, \tilde{G}^{\#}_1{}, \tilde{G}^{\#}_0{})$. See \eqref{eq:GammaComponents} below for explicit expressions. Putting it all together, we can express all differential spin coefficients in terms of background spin coefficients and GHP derivatives of the above mentioned metric components.

For instance, we get
\begin{subequations}
\label{eq:rhoTildePrimeSigmaTildePrime}
\begin{align}
\tilde{\rho}'={}&\frac{1}{2 \varsigma} (\varsigma^{\#}{} -  \varsigma) (\mathring{\rho}' -  \bar{\mathring{\rho}}')
 -  \frac{\varsigma^{\#}{} \thopBG \slashed{G}_{}}{4 \varsigma} (2 \varsigma^{\#}{} \varsigma - 1)
 + \frac{G_{2} \varsigma^{\#}{}^2 \thopBG \overline{G_{2}}}{4 \varsigma^2}
 + \frac{\overline{G_{2}} \varsigma^{\#}{}^2 \thopBG G_{2}}{4 \varsigma^2} ,\\
\tilde{\sigma}'={}&\frac{G_{2} \varsigma^{\#}{}}{2 e^{2i \nu} \varsigma} (\mathring{\rho}' -  \bar{\mathring{\rho}}')
 -  \frac{G_{2} \varsigma^{\#}{}^2 \thopBG \slashed{G}_{}}{4 e^{2i \nu} \varsigma^2}
 + \frac{G_{2}{}^2 \varsigma^{\#}{}^2 \thopBG \overline{G_{2}}}{8 e^{2i \nu} \varsigma^4}
 + \frac{\varsigma^{\#}{}^2 \thopBG G_{2}}{2 e^{2i \nu}} .
\end{align}
\end{subequations}
This is equivalent to 
\begin{subequations}
\label{eq:ThopBGG2GTr}
\begin{align}
\thopBG \slashed{G}_{}={}&- (2 + \slashed{G}_{}) (\tilde{\rho}' + \overline{\tilde{\rho}'})
 + 2 e^{2i \nu} \overline{G_{2}} \tilde{\sigma}'
 + \frac{2 G_{2} \overline{\tilde{\sigma}'}}{e^{2i \nu}} ,\\
\thopBG G_{2}={}&G_{2} (\bar{\mathring{\rho}}' - \mathring{\rho}' -  \tilde{\rho}' -  \overline{\tilde{\rho}'})
 + 2 e^{2i \nu} \varsigma^2 \tilde{\sigma}'
 + \frac{G_{2}{}^2 \overline{\tilde{\sigma}'}}{2 e^{2i \nu} \varsigma^2} ,\\
\overline{\tilde{\rho}'}={}&\tilde{\rho}'
 -  \frac{1}{\varsigma} (\varsigma^{\#}{} -  \varsigma) (\mathring{\rho}' -  \bar{\mathring{\rho}}').
\end{align}
\end{subequations}
Similarly, we get
\begin{subequations}
\begin{align}
\tilde{\epsilon}'={}&- \frac{\slashed{G}_{} \varsigma^{\#}{} \mathring{\rho}'}{8 \varsigma}
 + \frac{\slashed{G}_{} \varsigma^{\#}{} \bar{\mathring{\rho}}'}{8 \varsigma}
 -  \tfrac{1}{2} i \thopBG \nu
 + \frac{G_{2} \varsigma^{\#}{} \thopBG \overline{G_{2}}}{16 \varsigma^3}
 -  \frac{\overline{G_{2}} \varsigma^{\#}{} \thopBG G_{2}}{16 \varsigma^3} \\
={}&\frac{1}{4 \varsigma} (\varsigma^{\#}{} -  \varsigma) (\mathring{\rho}' -  \bar{\mathring{\rho}}')
 -  \frac{e^{2i \nu} \overline{G_{2}} \tilde{\sigma}'}{8 \varsigma^2}
 + \frac{G_{2} \overline{\tilde{\sigma}'}}{8 e^{2i \nu} \varsigma^2}
 -  \tfrac{1}{2} i \thopBG \nu .
\end{align}
\end{subequations}
As we have not yet fixed the differential spin rotation parameter $\nu$, we can use it to set $\tilde{\epsilon}'=0$. Translated into the foreground operators and $\tilde{G}^{\#}_2$, this condition is equivalent to the evolution equation \eqref{eq:nuevoleq} in the \frameGaugeHypothesesNoRef. 
Using this relation one can express the system \eqref{eq:ThopBGG2GTr} in terms of $\tilde{G}^{\#}_2$ and $\InvGTrTilde$ to get \eqref{eq:ThopInvGTilde2}, \eqref{eq:ThopInvGTrTilde} and \eqref{eq:algebraicrho1}.
A similar calculation for $\tilde{\tau}'$ yields the following after reduction with \eqref{eq:ThopBGG2GTr}
\begin{align}
\tilde{\tau}'={}&\frac{\mathring{\rho}'}{4 e^{i \nu} \varsigma^2} (\overline{G_{1}} G_{2} \varsigma^{\#}{} + 4 e^{i \nu} \bar{\eta} \varsigma^2 + 4 G_{1} \varsigma^{\#}{} \varsigma^2)
 + \frac{\tilde{\rho}'}{2 e^{i \nu} \varsigma^2} (\overline{G_{1}} G_{2} \varsigma^{\#}{} + 2 e^{i \nu} \bar{\eta} \varsigma^2 + 2 G_{1} \varsigma^{\#}{} \varsigma^2)
 -  \frac{G_{1} \varsigma^{\#}{} \bar{\mathring{\rho}}'}{2 e^{i \nu}}\nonumber\\
& + \frac{\tilde{\sigma}'}{2 \varsigma^2} (e^{i \nu} \overline{G_{2}} G_{1} \varsigma^{\#}{} + 2 \eta \varsigma^2 + 2 e^{i \nu} \overline{G_{1}} \varsigma^{\#}{} \varsigma^2)
 + \frac{G_{2} \mathring{\tau}}{4 e^{i \nu} \varsigma^2} (\varsigma^{\#}{} + \varsigma)
 + \frac{\bar{\mathring{\tau}}}{2 e^{i \nu}} (\varsigma^{\#}{} -  \varsigma)\nonumber\\
& -  \frac{\mathring{\tau}'}{2 e^{i \nu}} (2 e^{i \nu} -  \varsigma^{\#}{} -  \varsigma)
 + \frac{G_{2} \bar{\mathring{\tau}}'}{4 e^{i \nu} \varsigma^2} (\varsigma^{\#}{} -  \varsigma)
 + \frac{G_{2} \varsigma^{\#}{} \thopBG \overline{G_{1}}}{4 e^{i \nu} \varsigma^2}
 + \frac{\varsigma^{\#}{} \thopBG G_{1}}{2 e^{i \nu}} .
\end{align}
This equation can be used to solve for $\thopBG G_{1}$.

Similarly $\tilde{\tau}$ can be expressed as follows after substitution of the expressions for $\thopBG \slashed{G}_{}$, $\thopBG G_{2}$, $\thopBG G_{1}$ and $\thopBG \nu$ above
\begin{align}
\tilde{\tau}={}&- \frac{e^{i \nu} \varsigma^{\#}{}^2}{2 \varsigma^3} (\overline{G_{2}} G_{1} + 2 \overline{G_{1}} \varsigma^2) (\mathring{\rho}' -  \bar{\mathring{\rho}}')
 -  \eta (2 \tilde{\epsilon}' -  \overline{\tilde{\rho}'} -  \bar{\mathring{\rho}}')
 + \bar{\eta} \overline{\tilde{\sigma}'}
 + \frac{e^{i \nu} \overline{G_{2}} \varsigma^{\#}{}}{2 \varsigma^2} (\bar{\mathring{\tau}} + \mathring{\tau}')
 -  \overline{\tilde{\tau}'}\nonumber\\
& + (e^{i \nu} \varsigma^{\#}{} - 1) (\mathring{\tau} + \bar{\mathring{\tau}}')
 -  \thopBG \eta .
\label{eq:tauexpr1}
\end{align}
In the \frameGaugeHypotheses, equation \eqref{eq:etaevoleq} was chosen so that $\tilde{\tau}=0$. This gives an expression for $\thopBG \eta$. Using this in the expression for $\tilde{G}{}^{\#}_1$, we can derive the evolution equation for $\tilde{G}{}^{\#}_1$, i.e. \eqref{eq:ThopInvGTilde1}.

Using all the previous relations, one can express $\tilde{\beta}-\overline{\tilde{\beta}'}$ as 
\begin{align}
\tilde{\beta} -  \overline{\tilde{\beta}'}={}&- \frac{\varsigma^{\#}{}}{2 \varsigma^3} (e^{i \nu} \overline{G_{2}} G_{1} \varsigma^{\#}{} + 2 \eta \varsigma^2 + 2 e^{i \nu} \overline{G_{1}} \varsigma^{\#}{} \varsigma^2) (\mathring{\rho}' -  \bar{\mathring{\rho}}')
 + \frac{e^{i \nu} \overline{G_{2}} \varsigma^{\#}{} \mathring{\tau}'}{2 \varsigma^2}
 -  \overline{\tilde{\tau}'}
 -  \bar{\mathring{\tau}}'
 + e^{i \nu} \varsigma^{\#}{} \bar{\mathring{\tau}}'.
\end{align}
Using this relation, we can eliminate $G_{1}$ from \eqref{eq:tauexpr1} to obtain the evolution \eqref{eq:etaevoleq}. Translating to the $\tilde{G}{}^{\#}_i$ variables, we also get \eqref{eq:algebraicbeta1}.

Similarly, using the previous relations, we get a long expression 
\begin{align}
\tilde{\epsilon} + \overline{\tilde{\epsilon}}={}&- \tfrac{1}{2} \thopBG G_{0} + \dots
\end{align}
where the dots indicates an expression depending on $G_{1}, G_{2}, \slashed{G}_{}, \tilde{\sigma }', \tilde{\tau }', \tilde{\rho }'$ and the background spin coefficients. Translating this to the $\tilde{G}{}^{\#}_i$ variables, we get \eqref{eq:ThopInvGTilde0}.

Similarly, one can express $\tilde{\beta} + \overline{\tilde{\beta}'}$ and $\tilde{\epsilon} - \overline{\tilde{\epsilon}}$ to obtain \eqref{eq:ethnuTobeta} and \eqref{eq:thoetaToepsilon}. Here however, the $\tilde{G}{}_i$ variables turned out to give shorter expressions, so we used them instead.
The remaining equations in \eqref{eq:StructureSpincoeff1} were derived in the same way using the expressions for $\tilde{\rho}$, $\tilde{\sigma}$ and $\tilde{\kappa}$. As the direct expressions for these spin coefficients became long and complicated, we found that solving for the left hand sides of \eqref{eq:StructureSpincoeff1} gave us shorter expressions. The expressions can be inverted though, so all spin coefficients are expressible in terms of derivatives of $\nu$, $\eta$ and the metric components.
\end{proof}

\subsection{Ricci relations}
\begin{lemma}[Ricci relations]
\assumeRadiationAndFrameGaugeHypotheses. 

The Ricci relations take the form
\begin{subequations}
\label{eq:thopRicci}
\begin{align}
(\thop{} -  \mathring{\rho}' -  \tilde{\rho}' -  \overline{\tilde{\rho}'} -  \bar{\mathring{\rho}}')\tilde{\sigma}'={}&\tilde{\Psi}_{4}{},\\
(\thop{} - 2 \mathring{\rho}' -  \tilde{\rho}')\tilde{\rho}'={}&\tilde{\sigma}' \overline{\tilde{\sigma}'},
\label{eq:transport:rhoTildePrime}\\
(\thop{} -  \mathring{\rho}' -  \tilde{\rho}')\tilde{\tau}'={}&\tilde{\Psi}_{3}{}
 + (\overline{\tilde{\tau}'} -  \mathring{\tau} + \bar{\mathring{\tau}}') \tilde{\sigma}'
 + \tilde{\rho}' (- \bar{\mathring{\tau}} + \mathring{\tau}'),\\
(\thop{} -  \mathring{\rho}' -  \tilde{\rho}')\tilde{\beta}={}&- \tilde{\beta}' \overline{\tilde{\sigma}'}
 + (1 -  e^{i \nu} \varsigma^{\#}{}) \mathring{\rho}' \mathring{\tau}
 + \tilde{\rho}' \mathring{\tau},
\label{eq:thopbeta}\\
(\thop{} -  \overline{\tilde{\rho}'} -  \bar{\mathring{\rho}}')\tilde{\beta}'={}&\tilde{\Psi}_{3}{}
 -  \tilde{\beta} \tilde{\sigma}'
 -  \frac{G^{\#}_{2} \mathring{\rho}' \mathring{\tau}}{2 e^{i \nu} \varsigma^{\#}{}}
 -  \tilde{\sigma}' \mathring{\tau},
\label{eq:thopbetap}\\
\thop \tilde{\epsilon}={}&- \tilde{\Psi}_{2}{}
 + (\tilde{\tau}' -  \bar{\mathring{\tau}} + \mathring{\tau}') \tilde{\beta}
 + \Bigl(G^{\#}_{1} + \frac{\eta G^{\#}_{2}}{2 e^{i \nu} \varsigma^{\#}{}} -  e^{i \nu} \bar{\eta} \varsigma^{\#}{}\Bigr) \mathring{\rho}' \mathring{\tau}
 + \mathring{\tau} \tilde{\tau}'\nonumber\\
& + \tilde{\beta}' (\mathring{\tau} -  \overline{\tilde{\tau}'} -  \bar{\mathring{\tau}}'),\\
(\thop{} -  \overline{\tilde{\rho}'} -  \bar{\mathring{\rho}}')\tilde{\rho}={}&- \tilde{\Psi}_{2}{}
 + \mathring{\rho} \overline{\tilde{\rho}'}
 + \tilde{\sigma} \tilde{\sigma}'
 + (\overline{\tilde{\beta}} + \tilde{\beta}') \mathring{\tau}
 + 2 \bar{\eta} \mathring{\rho}' \mathring{\tau}
 -  \frac{G^{\#}_{2} \mathring{\tau}^2}{2 e^{i \nu} \varsigma^{\#}{}}
 -  \Bigl(1 -  \frac{\varsigma^{\#}{}}{e^{i \nu}}\Bigr) \edtpBG \mathring{\tau},\\
(\thop{} -  \mathring{\rho}' -  \tilde{\rho}')\tilde{\sigma}={}&(\mathring{\rho} + \tilde{\rho}) \overline{\tilde{\sigma}'}
 -  (\tilde{\beta} + \overline{\tilde{\beta}'}) \mathring{\tau}
 + 2 \eta \mathring{\rho}' \mathring{\tau}
 -  (1 -  e^{i \nu} \varsigma^{\#}{}) \mathring{\tau}^2
 -  \frac{e^{i \nu} \overline{G^{\#}_{2}} \edtpBG \mathring{\tau}}{2 \varsigma^{\#}{}} ,\\
\thop \tilde{\kappa}={}&- \tilde{\Psi}_{1}{}
 + (\overline{\tilde{\tau}'} -  \mathring{\tau} + \bar{\mathring{\tau}}') \tilde{\rho}
 + (\tilde{\tau}' -  \bar{\mathring{\tau}} + \mathring{\tau}') \tilde{\sigma}
 -  (\tilde{\epsilon} -  \overline{\tilde{\epsilon}}) \mathring{\tau}
 + (2 \eta \bar{\eta} + G^{\#}_{0}) \mathring{\rho}' \mathring{\tau}\nonumber\\
& -  \Bigl(G^{\#}_{1} + \frac{\eta G^{\#}_{2}}{2 e^{i \nu} \varsigma^{\#}{}} -  e^{i \nu} \bar{\eta} \varsigma^{\#}{}\Bigr) \mathring{\tau}^2
 + \mathring{\rho} \overline{\tilde{\tau}'}
 -  \Bigl(\overline{G^{\#}_{1}} + \frac{e^{i \nu} \bar{\eta} \overline{G^{\#}_{2}}}{2 \varsigma^{\#}{}} -  \frac{\eta \varsigma^{\#}{}}{e^{i \nu}}\Bigr) \edtpBG \mathring{\tau},
\end{align}
\end{subequations}
together with the supplementary relations \eqref{eq:ExtraRicci}. Here $G^{\#}_{i}$ can be interpreted in terms of $\tilde{G}{}^{\#}_{i}$ via \eqref{eq:InvGTildeInvGRelations} and there is the background formula
\begin{align}
\label{eq:edtpBGtauBG}
\edtpBG \mathring{\tau}={}&\tfrac{1}{2} \mathring{\Psi}_{2}{}
 -  \frac{\bar{\mathring{\Psi}}_{2}{} \bar{\kappa}_{1'}{}}{2 \kappa_{1}{}}
 + \mathring{\rho} \mathring{\rho}'
 -  \mathring{\rho} \bar{\mathring{\rho}}'
 + \mathring{\tau} \mathring{\tau}', 
\end{align}
where
\begin{align}
\label{eq:defKappa1}
\kappa_{1}={}&-\frac13(r-ia\cos\theta) .
\end{align}
\end{lemma}
\begin{proof}
To prove these relations, we begin with the Newman-Penrose (NP) version of the Ricci relations equations (4.11.12) in \cite{PenroseRindler} for both the foreground spin coefficients and operators. The foreground spin coefficients can then be written in terms of the differential spin coefficients from definition \ref{def:DiffSpinCoeff}. 
When the foreground NP operators acts on the background spin coefficients, we express the operators in terms of the background operators, via the relations in definition~\ref{def:ForegroundFrame}. The resulting background NP operators acting on background spin coefficients can then be eliminated using the background Ricci relations. After this procedure, all non-properly weighted quantities have been eliminated, and the operators can be translated into the foreground GHP operators yielding \eqref{eq:thopRicci} and \eqref{eq:ExtraRicci} after reduction with the structure equations \eqref{eq:nuetaevoleq}, \eqref{eq:thopmetric}, \eqref{eq:AlgebraicSpincoeff1}, \eqref{eq:StructureSpincoeff1}. Here some background derivatives of background spin coefficients have been simplified due to the vacuum Bianchi type D property of the Kerr spacetime. 

As an example we derive \eqref{eq:thopbetap} starting with the foreground NP-Ricci relation
\begin{align}
\Delta \beta '={}&\Psi_{3}{}
 + \bar{\beta} \epsilon '
 -  \beta ' \overline{\epsilon '}
 + \beta ' \overline{\rho '}
 -  \beta \sigma '
 -  \sigma ' \tau
 -  \epsilon ' \bar{\tau}
 + \bar\delta \epsilon '.
\end{align}
Translating to the differential spin coefficients and expressing the foreground NP derivatives in terms of the background NP derivatives when acting on background spin coefficients, we get
\begin{align}
\Delta \tilde{\beta}'={}&\tilde{\Psi}_{3}{}
 + \overline{\tilde{\beta}} \tilde{\epsilon}'
 + e^{-i \nu} \varsigma^{\#}{} \bar{\mathring{\beta}} (\mathring{\epsilon}'
 + \tilde{\epsilon}')
 + \tfrac{1}{2} e^{-i \nu} G^{\#}_{2} \varsigma^{\#}{}^{-1} \bar{\mathring{\beta}}' (\mathring{\epsilon}'
 + \tilde{\epsilon}')
 -  \tilde{\beta}' \overline{\tilde{\epsilon}'}
 - ( \tilde{\beta}'
 +  \bar{\eta} \tilde{\epsilon}') \bar{\mathring{\epsilon}}'
 + \tilde{\beta}' \overline{\tilde{\rho}'}
 + \bar\delta \tilde{\epsilon}'\nonumber\\
& + \tilde{\beta}' \bar{\mathring{\rho}}'
 -  \tilde{\beta} \tilde{\sigma}'
 -  \tilde{\sigma}' \mathring{\tau}
 -  \tilde{\sigma}' \tilde{\tau}
 -  \tilde{\epsilon}' \overline{\tilde{\tau}}
 -  \tilde{\epsilon}' \bar{\mathring{\tau}}
 - \tfrac{1}{2} e^{-i \nu} G^{\#}_{2} \varsigma^{\#}{}^{-1} (\mathring{\Delta} \mathring{\beta}
 +  \mathring{\delta} \mathring{\epsilon}')
 - e^{-i \nu} \varsigma^{\#}{} ( \mathring{\Delta} \mathring{\beta}'
 - \mathring{\bar\delta} \mathring{\epsilon}')\nonumber\\
&- \mathring{\beta} \bigl( e^{i \nu} \varsigma^{\#}{} \tilde{\sigma}'
 +  \tfrac{1}{2} e^{-i \nu} \varsigma^{\#}{}^{-1} \Delta G^{\#}_{2}
  - \tfrac{1}{2} e^{-i \nu} G^{\#}_{2} \varsigma^{\#}{}^{-1} (\varsigma^{\#}{}^{-1} \Delta \varsigma^{\#}{} + i \Delta \nu   - \overline{\tilde{\epsilon}'} -  \bar{\mathring{\epsilon}}' + \overline{\tilde{\rho}'} + \bar{\mathring{\rho}}')\bigr)
 \nonumber\\
&+ \mathring{\beta}' \bigl(
   e^{-i \nu} \varsigma^{\#}{} (i \Delta \nu - \overline{\tilde{\epsilon}'} -  \bar{\mathring{\epsilon}}' + \overline{\tilde{\rho}'} + \bar{\mathring{\rho}}' )
  - \tfrac{1}{2} e^{i \nu} \overline{G^{\#}_{2}} \varsigma^{\#}{}^{-1} \tilde{\sigma}'
 -  e^{-i \nu} \Delta \varsigma^{\#}{}\bigr)
 \nonumber\\
& + \mathring{\epsilon}' (\overline{\tilde{\beta}}
 -  \bar{\eta} \overline{\tilde{\epsilon}'}
 - 2 \bar{\eta} \bar{\mathring{\epsilon}}'
 + \bar{\eta} \overline{\tilde{\rho}'}
 + \bar{\eta} \bar{\mathring{\rho}}'
 + \eta \tilde{\sigma}'
 -  \overline{\tilde{\tau}}
 -  \bar{\mathring{\tau}}
 -  \Delta \bar{\eta})  .
\end{align}
Using the background Ricci relations, transforming the foreground NP derivatives into foreground GHP operators, and translating the metric coefficients to the $\tilde{G}{}^{\#}_i$ variables yield
\begin{align}
\thop \tilde{\beta}'={}&\tilde{\Psi}_{3}{}
 + \tilde{\beta}' \overline{\tilde{\rho}'}
 + \tilde{\beta}' \bar{\mathring{\rho}}'
 -  \tilde{\beta} \tilde{\sigma}'
 + \tfrac{1}{2} e^{i \nu} \tilde{G}^{\#}_2{} \varsigma^{\#}{} \varsigma^{-2} \mathring{\rho}' \mathring{\tau}
 -  \tilde{\sigma}' \mathring{\tau}
 -  \tilde{\sigma}' \tilde{\tau}
 -  \tilde{\epsilon}' \overline{\tilde{\tau}}
 -  \tilde{\epsilon}' \bar{\mathring{\tau}}
 + \edtp \tilde{\epsilon}'
 \nonumber\\
&+ \mathring{\epsilon}' (\overline{\tilde{\beta}}
 -  \tilde{\beta}'
 + \bar{\eta} \tilde{\epsilon}'
 + \bar{\eta} \overline{\tilde{\epsilon}'}
  + \bar{\eta} \overline{\tilde{\rho}'}
 + \bar{\eta} \bar{\mathring{\rho}}'
 + \eta \tilde{\sigma}'
 + \tfrac{1}{2} e^{i \nu} \tilde{G}^{\#}_2{} \varsigma^{\#}{} \varsigma^{-2} \mathring{\tau}
 -  \overline{\tilde{\tau}}
 -  \bar{\mathring{\tau}}
 + e^{-i \nu} \varsigma^{\#}{} \bar{\mathring{\tau}}
 -  \thop \bar{\eta})
\nonumber\\
& + \mathring{\beta}' (e^{-i \nu} \varsigma^{\#}{} \tilde{\epsilon}'
 -  e^{-i \nu} \varsigma^{\#}{} \overline{\tilde{\epsilon}'}
  + e^{-i \nu} \varsigma^{\#}{} \overline{\tilde{\rho}'}
 + \tfrac{1}{2} e^{-i \nu} \overline{\tilde{G}}^{\#}_2{} \varsigma^{\#}{} \varsigma^{-2} \tilde{\sigma}'
 + i e^{-i \nu} \varsigma^{\#}{} \thop \nu
 -  e^{-i \nu} \thop \varsigma^{\#}{})
\nonumber\\
& + \mathring{\beta} \bigl(- e^{i \nu} \varsigma^{\#}{} \tilde{\sigma}'
 + \tfrac{1}{2} e^{i \nu} \varsigma^{\#}{} \varsigma^{-2} \thop \tilde{G}^{\#}_2{}
 + \tfrac{1}{2} e^{i \nu} \tilde{G}^{\#}_2{} \varsigma^{-2} \thop \varsigma^{\#}{}
 + \tfrac{1}{2} e^{i \nu} \tilde{G}^{\#}_2{} \varsigma^{\#}{} \varsigma^{-2} (\tilde{\epsilon}'
 -  \overline{\tilde{\epsilon}'}
 + \mathring{\rho}'
 -  \bar{\mathring{\rho}}'
\nonumber\\
&\qquad 
 -  \overline{\tilde{\rho}'}
+ i \thop \nu
  - 2 \varsigma^{-1} \thop \varsigma)\bigr).
\end{align}
The evolution equations \eqref{eq:nuetaevoleq} together with the structure equations \eqref{eq:thopmetric} and \eqref{eq:AlgebraicSpincoeff1} will reduce this to \eqref{eq:thopbetap}. Observe that the equation \eqref{eq:thopbetap} is properly weighted even though we started with a non-properly weighted equation. 
 The other equations can be derived in the same way.
\end{proof}

\subsection{Bianchi system}
\begin{lemma}[Bianchi identities]
\assumeRadiationAndFrameGaugeHypotheses. 

The Bianchi identities take the form
\begin{subequations}
\label{eq:Bianchi}
\begin{align}
0={}&(\tho{} - 4 \mathring{\rho} - 4 \tilde{\rho})\tilde{\Psi}_{1}{}
 -  (\edtp{} -  \mathring{\tau}' -  \tilde{\tau}')\tilde{\Psi}_{0}{}
 + 3 (\mathring{\Psi}_{2}{} + \tilde{\Psi}_{2}{}) \tilde{\kappa},\\
0={}&(\tho{} - 3 \mathring{\rho} - 3 \tilde{\rho})\tilde{\Psi}_{2}{}
 -  (\edtp{} - 2 \mathring{\tau}' - 2 \tilde{\tau}')\tilde{\Psi}_{1}{}
 + 2 \tilde{\Psi}_{3}{} \tilde{\kappa}
 - 3 \mathring{\Psi}_{2}{} \tilde{\rho}
 + \tfrac{3}{2} (2 \eta \bar{\eta} + G^{\#}_{0}) \mathring{\Psi}_{2}{} \mathring{\rho}'
 -  \tilde{\Psi}_{0}{} \tilde{\sigma}'\nonumber\\
& - 3 \Bigl(G^{\#}_{1} + \frac{\eta G^{\#}_{2}}{2 e^{i \nu} \varsigma^{\#}{}} -  e^{i \nu} \bar{\eta} \varsigma^{\#}{}\Bigr) \mathring{\Psi}_{2}{} \mathring{\tau}
 - 3 \Bigl(\overline{G^{\#}_{1}} + \frac{e^{i \nu} \bar{\eta} \overline{G^{\#}_{2}}}{2 \varsigma^{\#}{}} -  \frac{\eta \varsigma^{\#}{}}{e^{i \nu}}\Bigr) \mathring{\Psi}_{2}{} \mathring{\tau}',\\
0={}&(\tho{} - 2 \mathring{\rho} - 2 \tilde{\rho})\tilde{\Psi}_{3}{}
 -  (\edtp{} - 3 \mathring{\tau}' - 3 \tilde{\tau}')\tilde{\Psi}_{2}{}
 + \tilde{\Psi}_{4}{} \tilde{\kappa}
 - 3 \bar{\eta} \mathring{\Psi}_{2}{} \mathring{\rho}'
 - 2 \tilde{\Psi}_{1}{} \tilde{\sigma}'
 + \frac{3 G^{\#}_{2} \mathring{\Psi}_{2}{} \mathring{\tau}}{2 e^{i \nu} \varsigma^{\#}{}}\nonumber\\
& + 3 \Bigl(1 -  \frac{\varsigma^{\#}{}}{e^{i \nu}}\Bigr) \mathring{\Psi}_{2}{} \mathring{\tau}'
  + 3 \mathring{\Psi}_{2}{} \tilde{\tau}',\\
0={}&(\tho{} -  \mathring{\rho} -  \tilde{\rho})\tilde{\Psi}_{4}{}
 -  (\edtp{} - 4 \mathring{\tau}' - 4 \tilde{\tau}')\tilde{\Psi}_{3}{}
 - 3 (\mathring{\Psi}_{2}{} + \tilde{\Psi}_{2}{}) \tilde{\sigma}',\\
0={}&(\thop{} -  \mathring{\rho}' -  \tilde{\rho}')\tilde{\Psi}_{0}{}
 -  (\edt{} - 4 \mathring{\tau})\tilde{\Psi}_{1}{}
 - 3 (\mathring{\Psi}_{2}{} + \tilde{\Psi}_{2}{}) \tilde{\sigma},\\
0={}&(\thop{} - 2 \mathring{\rho}' - 2 \tilde{\rho}')\tilde{\Psi}_{1}{}
 -  (\edt{} - 3 \mathring{\tau})\tilde{\Psi}_{2}{}
 - 3 \eta \mathring{\Psi}_{2}{} \mathring{\rho}'
 - 2 \tilde{\Psi}_{3}{} \tilde{\sigma}
 + 3 (1 -  e^{i \nu} \varsigma^{\#}{}) \mathring{\Psi}_{2}{} \mathring{\tau}
 + \frac{3 e^{i \nu} \overline{G^{\#}_{2}} \mathring{\Psi}_{2}{} \mathring{\tau}'}{2 \varsigma^{\#}{}} ,\\
0={}&(\thop{} - 3 \mathring{\rho}' - 3 \tilde{\rho}')\tilde{\Psi}_{2}{}
 -  (\edt{} - 2 \mathring{\tau})\tilde{\Psi}_{3}{}
 - 3 \mathring{\Psi}_{2}{} \tilde{\rho}'
 -  \tilde{\Psi}_{4}{} \tilde{\sigma},\\
0={}&(\thop{} - 4 \mathring{\rho}' - 4 \tilde{\rho}')\tilde{\Psi}_{3}{}
 -  (\edt{} -  \mathring{\tau})\tilde{\Psi}_{4}{}.
\end{align}
\end{subequations}
Here $G^{\#}_{i}$ can be interpreted in terms of $\tilde{G}{}^{\#}_{i}$ via \eqref{eq:InvGTildeInvGRelations}.
\end{lemma}
\begin{proof}
A direct translation of the standard GHP Bianchi identities in \cite{GHP} to our differential variables gives the relations \eqref{eq:Bianchi}. Here we have also used the background type D Bianchi identities to handle the derivatives of $\mathring{\Psi}_{2}$.
\end{proof}

\begin{remark}
It is important to note that the full set of equations, i.e. the evolution equations for the differential Lorentz transformation variables \eqref{eq:nuetaevoleq}, the structure equations \eqref{eq:thopmetric}\eqref{eq:AlgebraicSpincoeff1}\eqref{eq:StructureSpincoeff1}, the Ricci relations \eqref{eq:thopRicci}\eqref{eq:ExtraRicci} and the Bianchi identities \eqref{eq:Bianchi}, are all properly weighted equations. This means that they make sense for the entire family of background principal null tetrads. 
\end{remark}

\begin{remark}
From the Bianchi and Ricci equations, one can derive non-linear versions of the Teukolsky master equations (TME) \cite{Teukolsky}
\begin{subequations}
\begin{align}
\hspace{6ex}&\hspace{-6ex}\bigl((\thop{} - 4 \mathring{\rho}' - 4 \tilde{\rho}' -  \overline{\tilde{\rho}'} -  \bar{\mathring{\rho}}')(\tho{} -  \mathring{\rho} -  \tilde{\rho})
 -  (\edtp{} -  \bar{\mathring{\tau}} - 4 \mathring{\tau}' - 4 \tilde{\tau}')(\edt{} -  \mathring{\tau})
 - 3 (\mathring{\Psi}_{2}{} + \tilde{\Psi}_{2}{} + \tilde{\sigma} \tilde{\sigma}')\bigr)\tilde{\Psi}_{4}{}\nonumber\\
={}&4 \tilde{\Psi}_{3}{} (\edt{} -  \overline{\tilde{\tau}'} -  \bar{\mathring{\tau}}')\tilde{\sigma}'
 + 4 \tilde{\sigma}' \edt \tilde{\Psi}_{3}{}
 - 10 \tilde{\Psi}_{3}{}^2,\\
\hspace{6ex}&\hspace{-6ex}\bigl((\tho{} - 4 \mathring{\rho} - 4 \tilde{\rho} -  \overline{\tilde{\rho}} -  \bar{\mathring{\rho}})(\thop{} -  \mathring{\rho}' -  \tilde{\rho}')
 -  (\edt{} - 4 \mathring{\tau} -  \overline{\tilde{\tau}'} -  \bar{\mathring{\tau}}')(\edtp{} -  \mathring{\tau}' -  \tilde{\tau}')
 - 3 (\mathring{\Psi}_{2}{} + \tilde{\Psi}_{2}{} + \tilde{\sigma} \tilde{\sigma}')\bigr)\tilde{\Psi}_{0}{}\nonumber\\
={}&-4 \tilde{\Psi}_{1}{} (\thop{} -  \overline{\tilde{\rho}'} -  \bar{\mathring{\rho}}')\tilde{\kappa}
 + 4 \tilde{\Psi}_{1}{} (\edtp{} -  \bar{\mathring{\tau}})\tilde{\sigma}
 - 4 \tilde{\kappa} \thop \tilde{\Psi}_{1}{}
 + 4 \tilde{\sigma} \edtp \tilde{\Psi}_{1}{}
 - 10 \tilde{\Psi}_{1}{}^2.
\end{align}
\end{subequations}
\end{remark}

From the Bianchi equations, it follows that the differential curvatures satisfy the evolution system given in the following corollary. 

\begin{corollary}[Evolution system for the differential curvature components]
\label{cor:evolsys}
\assumeRadiationAndFrameGaugeHypotheses. 

Let $(t, x, y, z)$ be a real coordinate system such that constant $t$ hypersurfaces are spacelike. 

The differential curvature components satisfy
\begin{subequations}
\label{eq:HyperbolicSystemForCurvature}
\begin{align}
B^t
\partial_t
\begin{pmatrix}
\tilde{\Psi}_{0}\\\tilde{\Psi}_{1}\\\tilde{\Psi}_{2}\\\tilde{\Psi}_{3}\\\tilde{\Psi}_{4}
\end{pmatrix}
={}&
-\sum_{i\in\{x,y,z\}}B^i\partial_i\begin{pmatrix}
\tilde{\Psi}_{0}\\\tilde{\Psi}_{1}\\\tilde{\Psi}_{2}\\\tilde{\Psi}_{3}\\\tilde{\Psi}_{4}
\end{pmatrix}
+
F,
\end{align}
where $F=F(\geometricVariables)$ is a function of the geometric variables and
\begin{align}
B^i={}&
\begin{pmatrix}
n^t{}^3 n^i & - n^t{}^3 m^i  & 0 & 0 & 0\\
- n^t{}^3\bar{m}^i  &  n^t{}^3l^i + l^t n^t{}^2 n^i & - l^t  n^t{}^2m^i & 0 & 0\\
0 & - l^t  n^t{}^2\bar{m}^i & l^t  n^t{}^2l^i + l^t{}^2 n^t n^i & - l^t{}^2  n^t m^i & 0\\
0 & 0 & - l^t{}^2 n^t \bar{m}^i  & l^t{}^2 n^t l^i  + l^t{}^3 n^i & - l^t{}^3 m^i\\
0 & 0 & 0 & - l^t{}^3 \bar{m}^i & l^t{}^3 l^i
\end{pmatrix}  \;\forall\;  i\in \{t,x,y,z\} .
\end{align}
\end{subequations}
\end{corollary}
\begin{proof}
Consider the components of the foreground frame in terms of the coordinate co-frame, i.e. $l^t=l^a(dt)_a$ etc. The spacelike nature of the hypersurfaces means that the co-normal $(dt)_a$ is time-like, i.e. $0<g^{\#}{}^{ab}(dt)_a(dt)_b=2l^tn^t-2m^t\bar{m}^t$. In particular, we get $l^tn^t > m^t\bar{m}^t\geq 0$. Furthermore as we assume that $l^a$ and $n^a$ are future pointing, we get that $l^t>0$ and $n^t>0$. 

We can write \eqref{eq:Bianchi} in the form 
\begin{align*}
&\begin{pmatrix}
- \bar{m}^t & l^t & 0 & 0 & 0\\
0 & - \bar{m}^t & l^t & 0 & 0\\
0 & 0 & - \bar{m}^t & l^t & 0\\
0 & 0 & 0 & - \bar{m}^t & l^t\\
n^t & - m^t & 0 & 0 & 0\\
0 & n^t & - m^t & 0 & 0\\
0 & 0 & n^t & - m^t & 0\\
0 & 0 & 0 & n^t & - m^t
\end{pmatrix}\partial_t
\begin{pmatrix}
\tilde{\Psi}_{0}\\\tilde{\Psi}_{1}\\\tilde{\Psi}_{2}\\\tilde{\Psi}_{3}\\\tilde{\Psi}_{4}
\end{pmatrix}
\\&={}\sum_{i\in\{x,y,z\}}\begin{pmatrix}
\bar{m}^i & - l^i & 0 & 0 & 0\\
0 & \bar{m}^i & - l^i & 0 & 0\\
0 & 0 & \bar{m}^i & - l^i & 0\\
0 & 0 & 0 & \bar{m}^i & - l^i\\
- n^i & m^i & 0 & 0 & 0\\
0 & - n^i & m^i & 0 & 0\\
0 & 0 & - n^i & m^i & 0\\
0 & 0 & 0 & - n^i & m^i
\end{pmatrix}\partial_i\begin{pmatrix}
\tilde{\Psi}_{0}\\\tilde{\Psi}_{1}\\\tilde{\Psi}_{2}\\\tilde{\Psi}_{3}\\\tilde{\Psi}_{4}
\end{pmatrix}
+
l.o.
\end{align*}
where $l.o.$ denotes a function of the geometric variables $\geometricVariables$ but not their derivatives. The corollary follows from multiplying this by
\begin{align*}
\begin{pmatrix}
0 & 0 & 0 & 0 & n^t{}^3 & 0 & 0 & 0\\
n^t{}^3 & 0 & 0 & 0 & 0 & l^t n^t{}^2 & 0 & 0\\
0 & l^t n^t{}^2 & 0 & 0 & 0 & 0 & l^t{}^2 n^t & 0\\
0 & 0 & l^t{}^2 n^t & 0 & 0 & 0 & 0 & l^t{}^3\\
0 & 0 & 0 & l^t{}^3 & 0 & 0 & 0 & 0
\end{pmatrix} .
\end{align*}
\end{proof}

\subsection{First-order symmetric-hyperbolicity}

\begin{theorem}[First-order symmetric-hyperbolic system]
\label{thm:FOSH}
\assumeRadiationAndFrameGaugeHypotheses. 

Assume $\metric$ is a solution of the vacuum Einstein equation. 
Let $(t, x, y, z)$ be a real coordinate system such that constant $t$ hypersurfaces are spacelike. 

The system \eqref{eq:nuetaevoleq},
\eqref{eq:thopmetric},
\eqref{eq:thopRicci}, and
\eqref{eq:HyperbolicSystemForCurvature} 
forms a first-order symmetric-hyperbolic system for the geometric variables $\geometricVariables$, 
and where $\tilde{G}_i$ and $\tilde{\slashed{G}}$ are given in terms of the geometric variables by equations \eqref{eq:InvGTildeInvGRelations} and \eqref{eq:GslashRelations} and where $\varsigma$ and $\varsigma^\#$ are given by equation \eqref{eq:defVarsigma}. 
\end{theorem}
\begin{proof} 
The goal is to show that, using the algebraic relations for $\tilde{G}_i$ and $\tilde{\slashed{G}}$ in \eqref{eq:InvGTildeInvGRelations} and for $\varsigma$ and $\varsigma^{\#}$ in \eqref{eq:defVarsigma}, 
the system \eqref{eq:nuetaevoleq},
\eqref{eq:thopmetric},
\eqref{eq:thopRicci}, and
\eqref{eq:HyperbolicSystemForCurvature} can be written in the form 
\begin{align}
A(\geometricVariables)^t \partial_t \geometricVariables ={}& \sum_{i\in\{x,y,z\}} A(\geometricVariables)^i \partial_i \geometricVariables + F(\geometricVariables),
\label{eq:FullFOSH} 
\end{align} 
where $A^t$ and $A^i$ are Hermitian matrices, where $A^t$ is positive definite, and where $A^t$, each $A^i$, and $F$ are functions of $\geometricVariables$ and $\overline\geometricVariables$. Note also that in this equation $F(0)=0$, so $\geometricVariables=0$ is a solution of this system. 

Since $\vecN$ is future-directed, for any $\varphi\in\geometricVariables$, any transport equation of the form $\thop\varphi=f_1(\geometricVariables)$ can be written in coordinates as $n^t\partial_t\varphi=-\sum_{i\in\{x,y,z\}}n^i\partial_i\varphi +f_2(\geometricVariables)$, where $f_2$ is constructed from $f_1$ and from products of the connection coefficients appearing in $\thop$ and of $\varphi$. The equations \eqref{eq:nuetaevoleq},
\eqref{eq:thopmetric}, and \eqref{eq:thopRicci} are all of the form $\thop\varphi=f(\geometricVariables)$. Therefore, the right hand side of the entire transport system has a diagonal principal part. Since $\vecN$ is real, these diagonal parts are trivially Hermitian. The spacelike nature of the slice implies $n^t>0$, so the left hand side matrix is diagonal and positive definite. This gives equations for differential Lorentzian transformations, the metric components, and the spin components. 

It remains to obtain equations for the curvature components. In equation \eqref{eq:HyperbolicSystemForCurvature}, the $B^i$ are clearly symmetric. It remains to show $B^t$ is positive definite. The determinant and sub-determinants of $B^t$ are $4 l^t{}^8 n^t{}^8 (l^t n^t -  m^t \bar{m}^t) (2l^t n^t -  m^t \bar{m}^t)>0$, $l^t{}^4 n^t{}^8 \bigl(l^t{}^2 n^t{}^2 + 6 l^t n^t (l^t n^t -  m^t \bar{m}^t) + (l^t n^t -  m^t \bar{m}^t)^2\bigr)>0$, $l^t{}^2 n^t{}^8 \bigl(l^t n^t + 3 (l^t n^t -  m^t \bar{m}^t)\bigr)>0$, $n^t{}^6 (2 l^t n^t - m^t \bar{m}^t)>0$, $n^t{}^4>0$. Hence, $B^t$ is positive definite. This gives equations for the curvature components and hence all components of $\geometricVariables$. Thus, equations \eqref{eq:nuetaevoleq},
\eqref{eq:thopmetric},
\eqref{eq:thopRicci}, and
\eqref{eq:HyperbolicSystemForCurvature} form a first-order symmetric-hyperbolic system for $\geometricVariables$. 
\end{proof}

\subsection{Completing the proof of theorem \ref{thm:NLnhIsLWP}}
\begin{proof}[Proof of theorem \ref{thm:NLnhIsLWP}]
The symmetric hyperbolicity in point \ref{pt:FOSH} is proved in theorem \ref{thm:FOSH}. 

The geometric variables $\geometricVariables$ include the foreground metric coefficients $\tilde{G}^{\#}_i$ and $\tilde{\slashed{G}}^{\#}$ and the differential Lorentz transformations $\nu$ and $\eta$. From these, the components with respect to the background metric $G^{\#}$ and $\slashed{G}$ can be calculated using equations \eqref{eq:GslashRelations} and \eqref{eq:InvGTildeInvGRelations}. From these and the background tetrad $(\vecLBackground,\vecNBackground,\vecMBackground,\vecMbBackground)$, the original metric $\metric_{ab}$ can be calculated. This completes the proof of point \ref{pt:systemDeterminesMetric}. 

Finally, suppose one has a set of initial data for the vacuum Einstein equation. Choose also a set of initial data for $\nu$ and $\eta$. On the one hand, the initial data for the Einstein equation launches a unique solution of the vacuum Einstein equation. From the results in section \ref{s:ProofOfEnforecabilityOfTheNLORG}, coordinates can be chosen so that this metric satisfies the \NLnhCondition. Let $\nu$ and $\eta$ satisfy the evolution equations \eqref{eq:nuetaevoleq} from the \frameGaugeHypothesesNoRef. The lemmas from this section give that the geometric variables constructed from differential Lorentz transforms, the foreground metric, its connection coefficients, and its curvature (and from the background quantities) satisfy the system \eqref{eq:nuetaevoleq},
\eqref{eq:thopmetric},
\eqref{eq:thopRicci}, and
\eqref{eq:HyperbolicSystemForCurvature}. Call this solution $\geometricVariables_1$. On the other hand, the initial data for the vacuum Einstein equation, together with the choice of initial data for $\nu$ and $\eta$, launch a unique solution of the system \eqref{eq:nuetaevoleq},
\eqref{eq:thopmetric},
\eqref{eq:thopRicci}, and
\eqref{eq:HyperbolicSystemForCurvature}. Call this solution $\geometricVariables_2$. Since $\geometricVariables_1$ and $\geometricVariables_2$ have the same initial data, since they satisfy the same system, and since there is uniqueness of solutions to first-order symmetric-hyperbolic  systems, it follows that $\geometricVariables_1$ and $\geometricVariables_2$ are the same. This means that the metric components coincide. In particular, the metric constructed from the solution of the first-order symmetric-hyperbolic system $\geometricVariables_2$ coincides with the solution of the vacuum Einstein equation launched from the corresponding initial data. In particular, the solution of the first-order symmetric-hyperbolic system determines a metric which satisfies the vacuum Einstein equation. This completes the proof of the final point in the theorem. 
\end{proof}

\subsection{Initial data and residual gauge}
Before concluding this section, we make a few remarks about the initial data and the residual gauge. 

\begin{remark}[Propagation of constraints]
With a coordinate system as in corollary~\ref{cor:evolsys}, one can interpret the equations \eqref{eq:AlgebraicSpincoeff1}, \eqref{eq:StructureSpincoeff1}, \eqref{eq:ExtraRicci} and the remaining Bianchi identities as a set of constraint equations, by expressing the derivatives in terms of coordinate derivatives and eliminating the time derivatives with \eqref{eq:FullFOSH}. By applying a $\thop$ derivative to this set of equations, commuting the $\thop$ inside, using the evolution equations, and again the constraints, one finds that the constraints propagate. 
\end{remark}

\begin{remark}[Initial data for spin coefficients]
If one is given initial data only for the metric coefficients, $\nu$, and $\eta$, one can construct initial data for the differential spin coefficients via the full set of structure equations. Initial data for the curvature can be constructed from a subset of the Ricci relations. Note that the values for $\nu$ and $\eta$ on the initial slice are not constrained if we interpret \eqref{eq:StructureSpincoeff1} as equations giving initial data for differential spin coefficients. The initial data for the metric coefficients are constrained due to the fact that we are only considering vacuum perturbations. 
\end{remark}

\begin{remark}[Residual gauge]
In this section, we use the \NLnhCondition{} in the open set on which we construct solutions. From the perspective of naive function counting, these specify the four free functions that can be specified by a gauge choice in an open set. This gives a unique solution for each choice of initial data. However, there remains a residual gauge freedom that can be treated as a diffeomorphism of the initial data. 
In the next section we will see that the diffeomorphism part of the initial data gauge freedom can be partially fixed by making $\InvGTrTilde$ small in an appropriate sense. The initial data part of the differential frame gauge can be fixed by choosing the initial data for $\nu$ and $\eta$. 
As discussed above, this can be done in an arbitrary way, but it is convenient to choose the initial data for $\nu$ to be $0$. As we will see below, $\nu$ will then stay quadratically small. 
To also set the initial data for $\eta$ to zero is also possible, but it will not stay quadratically small during the evolution. 
An alternative is to set the initial data for $\eta$ so that the initial data for $\tilde{\beta}$ is quadratically small. This has the advantage that $\tilde{\beta}$ will stay quadratically small. For details see section~\ref{sec:SmallnessBeta}. 
\end{remark}

\section{Imposing the trace condition}
\label{s:traceCondition} 
This section can be summarized as follows: Price-Shankar-Whiting \cite{Price:2006ke} have shown that, for the linearized Einstein equation, a linearized gauge transformation that satisfies the \LnhCondition{} can be further transformed to satisfy the \linearTraceCondition{} and hence the \ClassicalORG; we show that the same result can be shown to quadratic order for the full Einstein equation.

The main result of this section is the following refinement of theorem \ref{thm:NLnhEnforceability}. 

\begin{theorem}[Enforceability of the trace condition to quadratic order]
\label{thm:NLORGEnforceability}
\assumeRadiationHypotheses. Let $\gaugeRegularityOut$ be a sufficiently large integer and let $(X,Y,I,J,h,U,V)$ be as in definition \ref{def:diffeomorphismGauge} for a diffeomorphism. 

There exist $\varepsilon_0>0$, $\gaugeRegularityIn>\gaugeRegularityOut$, and $K>0$ such that if $\metric_{ab}$ is a symmetric $(0,2)$ tensor $\metric_{ab}$ satisfying the vacuum Einstein equation and $|\metric-\metricBackground|_{C^{\gaugeRegularityIn}(U)}<\varepsilon_0$, then:
\begin{enumerate}
\item 
\label{pt:NLnhEnforceability}
There is a $C^{\gaugeRegularityOut}$ diffeomorphism gauge transform $(U,V,\Phi)$ such that
$|\Phi^{-1}_*\metric-\metricBackground|_{C^{\gaugeRegularityOut}(h(Y))}$ $\leq K|\metric-\metricBackground|_{C^{\gaugeRegularityIn}(h(X))}$, and $\Phi^{-1}_*\metric$ satisfies the \NLnhCondition{} on $V$. 
\item 
\label{pt:NLORGEnforceability}
Furthermore, $\Phi$ can be chosen such that 
\begin{align}
\label{eq:NLTraceConditionToQuadraticOrder}
|\metricBackground^{ab}(\Phi^{-1}_*\metric)_{ab}-\metricBackground^{ab}\metricBackground_{ab}|_{C^{\gaugeRegularityOut}(V)}
\leq K |\metric-\metricBackground|_{C^{\gaugeRegularityIn}(U)}^2 .
\end{align}
\end{enumerate}
\end{theorem}

\begin{remark}
In the proof of the above theorem, we use only the diffeomorphism gauge and the background operators,  independent of the choice of foreground frame gauge, thus leaving the freedom of choosing a frame gauge. In particular, the statements in both the well-posedness theorem \ref{thm:NLnhIsLWP} and the above theorem \ref{thm:NLORGEnforceability} can  simultaneously hold by the above diffeomorphism gauge choice and the  frame gauge choice in definition \ref{def:frameGaugeHypotheses}.
\end{remark}

\subsection{Review of the \LnhCondition{} from Price-Shankar-Whiting \texorpdfstring{\cite{Price:2006ke}}{}}
In this subsection, we review the results of \cite{Price:2006ke} on the linear radiation gauge and linear trace conditions as well as the ORG, which appear in definition \ref{def:linearORG}. The radiation gauge in \cite{Price:2006ke} is based on the vector field $\vecL$, while ours is based on $\vecN$. Therefore, many of the formulas interchange primed and unprimed. We state the following result for a Kerr background, although \cite{Price:2006ke} show these results hold in the wider class of metrics.

To explain the linear theory, following \cite{Price:2006ke}, we introduce the Held integration technique first described in \cite{Held1974}. 
We have re-derived all equations and made slight modifications to make sure that all expressions are properly weighted. 
A spinor $\alpha$ is defined to be a Held spinor if $\thop\alpha=0$. For a spinor $\alpha$, the notation $\Held{\alpha}$ indicates that $\alpha$ is a Held spinor. For a vector field $\vecX$ and a point $p$, define $\Flow{\vecX}{s}(p)$ to be the flow along $\vecX$, i.e. such that for any $p$, the function $\Flow{\vecX}{s}(p)$ is the solution of $\frac{d}{ds}\Flow{\vecX}{s}(p)=\vecX$ and $\Flow{\vecX}{0}(p)=p$; for sets $S$ and $P$ of $\Reals$ and the manifold respectively, define $\Flow{\vecX}{S}(P)=\cup_{s\in S,p\in P}\Flow{\vecX}{s}(p)$. For a spinor $\alpha$ defined on a hypersurface $\Sigma$ which is given as the graph of $r$ as a function of $(v,\omega)$, there is a unique extension of $\alpha$ as a Held spinor on $\Flow{\vecN}{\Reals}(\Sigma)$, which we will denote by $\Held{\alpha}$. For Held spinors defined on an open set, the operators $\Heldtho$, $\Heldeth$, and $\Heldethp$ are defined to be
\begin{subequations}
\begin{align}
\Heldtho \varphi ={}&- \frac{p \mathring{\Psi}_{2}{} \varphi}{2 \mathring{\rho}'}
 -  \frac{q \bar{\mathring{\Psi}}_{2}{} \varphi}{2 \bar{\mathring{\rho}}'}
 + \thoBG \varphi
 -  \frac{\mathring{\tau}' \edtBG \varphi}{\mathring{\rho}'}
 -  \frac{\bar{\mathring{\tau}}' \edtpBG \varphi}{\bar{\mathring{\rho}}'} ,\\
\Heldeth \varphi ={}&- \frac{p \bar{\mathring{\tau}}' \varphi}{\bar{\mathring{\rho}}'}
 + \frac{\edtBG \varphi}{\mathring{\rho}'} ,\\
\Heldethp \varphi ={}&- \frac{q \mathring{\tau}' \varphi}{\mathring{\rho}'}
 + \frac{\edtpBG \varphi}{\bar{\mathring{\rho}}'} .
\end{align}
\end{subequations}
For a spinor $\alpha$ defined on $\Sigma$, the operator $\Heldtho$ denotes the operator defined by extending $\alpha$ to $\Held{\alpha}$, applying $\Heldtho$, and then restricting to $\Sigma$ again. For a spinor $\alpha$ defined on $\Sigma$, the operators $\Heldeth$ and $\Heldethp$ are defined analogously. Note that, when acting on Held spinors, the operator $\thop$ commutes with $\Heldtho$, $\Heldeth$, and $\Heldethp$. 

The following lemma encapsulates the key results of \cite{Price:2006ke} regarding the \LnhCondition. Equations \eqref{eq:NullGaugevcondGHP} and \eqref{eq:PSW23} correspond to equations (15) and (23) of \cite{Price:2006ke}. 

\begin{lemma}[The \LnhCondition{}  \cite{Price:2006ke}]
Let $0<r_1<r_2<\infty$ and $v_1<v_2$. Let the \backgroundHypotheses{} hold with  $U=(r_1,r_2)\times(v_1,v_2)\times\Sphere$.

Let $h_{ab}$ be a symmetric $(0,2)$ tensor that satisfies the \LnhCondition{} of definition \ref{def:linearORG}. Let $\vecPSW$ be a vector field. 
\begin{enumerate}
\item The tensor field $h_{ab}+\Lie_\vecPSW\metricBackground_{ab}$ satisfies the \LnhCondition{} if
\begin{subequations}
\label{eq:NullGaugevcondGHP}
\begin{align}
\thopBG \vecPSW_{\mathring{l}}={}&- \vecPSW_{\vecMBackground} (\bar{\mathring{\tau}} + \mathring{\tau}')
 -  \vecPSW_{\vecMbBackground} (\mathring{\tau} + \bar{\mathring{\tau}}')
 -  \thoBG \vecPSW_{\mathring{n}},\\
\thopBG \vecPSW_{\mathring{n}}={}&0,\\
\thopBG \vecPSW_{\vecMBackground}={}&- \vecPSW_{\vecMBackground} \mathring{\rho}'
 -  \vecPSW_{\mathring{n}} \mathring{\tau}
 -  \edtBG \vecPSW_{\mathring{n}},\\
\thopBG \vecPSW_{\vecMbBackground}={}&- \vecPSW_{\vecMbBackground} \bar{\mathring{\rho}}'
 -  \vecPSW_{\mathring{n}} \bar{\mathring{\tau}}
 -  \edtpBG \vecPSW_{\mathring{n}}.
\end{align}
\end{subequations}
\item The general solution of  \eqref{eq:NullGaugevcondGHP} is given in terms of arbitrary Held spinors $\Held{\vecPSW}_{\vecLBackground}$, $\Held{\vecPSW}_{\vecNBackground}$, $\Held{\vecPSW}_{\vecMBackground}$, $\Held{\vecPSW}_{\vecMbBackground}$ by 
\begin{subequations}
\label{eq:PSW23}
\begin{align}
\vecPSW_{\mathring{l}}={}&\Held{\vecPSW}_{\mathring{l}}
 + \frac{\Held{\vecPSW}_{\vecMBackground} \mathring{\tau}'}{\mathring{\rho}'^2}
 + \frac{\Held{\vecPSW}_{\vecMbBackground} \bar{\mathring{\tau}}'}{\bar{\mathring{\rho}}'^2}
 + \Held{\vecPSW}_{\mathring{n}} \Bigl(\frac{\mathring{\Psi}_{2}{}}{2 \mathring{\rho}'^2} + \frac{\bar{\mathring{\Psi}}_{2}{}}{2 \bar{\mathring{\rho}}'^2} + \frac{\mathring{\tau}' \bar{\mathring{\tau}}'}{\mathring{\rho}' \bar{\mathring{\rho}}'}\Bigr)
 + \frac{1}{2}\Bigl(\frac{1}{\mathring{\rho}'} + \frac{1}{\bar{\mathring{\rho}}'}\Bigr) \Heldtho \Held{\vecPSW}_{\mathring{n}}
 -  \frac{\mathring{\tau}' \Heldeth \Held{\vecPSW}_{\mathring{n}}}{\mathring{\rho}'}
 -  \frac{\bar{\mathring{\tau}}' \Heldethp \Held{\vecPSW}_{\mathring{n}}}{\bar{\mathring{\rho}}'} ,\\
\vecPSW_{\mathring{n}}={}&\Held{\vecPSW}_{\mathring{n}},\\
\vecPSW_{\vecMBackground}={}&\frac{\Held{\vecPSW}_{\vecMBackground}}{\mathring{\rho}'}
 + \frac{\Held{\vecPSW}_{\mathring{n}} \bar{\mathring{\tau}}'}{\bar{\mathring{\rho}}'}
 -  \Heldeth \Held{\vecPSW}_{\mathring{n}},\\
\vecPSW_{\vecMbBackground}={}&\frac{\Held{\vecPSW}_{\vecMbBackground}}{\bar{\mathring{\rho}}'}
 + \frac{\Held{\vecPSW}_{\mathring{n}} \mathring{\tau}'}{\mathring{\rho}'}
 -  \Heldethp \Held{\vecPSW}_{\mathring{n}}.
\end{align}
\end{subequations}
\end{enumerate}
\end{lemma}

The approach of \cite{Price:2006ke} to the \ClassicalORG{} proceeds as follows. Since the \LnhCondition{} has already been treated, it remains to treat the \linearTraceCondition. In the linearization of the Einstein equation, the trace is the linearization of $\InvGTrTilde$, so it satisfies a linearized version of \eqref{eq:ThopInvGTrTilde}, the linearization of which is a transport equation driven by the linearization of $\tilde{\rho}'$. In turn, $\tilde{\rho}'$ satisfies \eqref{eq:transport:rhoTildePrime}, the linearization of which is a homogeneous transport equation. Thus, the trace satisfies a second-order ordinary differential equation, which has a general solution involving two free parameters, denoted $\Held{a}$ and $\Held{b}$. If a linearized gauge transformation satisfies the linearized radiation gauge, then $\Held{a}$ and $\Held{b}$ can be expressed in terms of $\Held{\vecPSW}_{\vecLBackground}$, $\Held{\vecPSW}_{\vecMBackground}$, $\Held{\vecPSW}_{\vecMbBackground}$, $\Held{\vecPSW}_{\vecNBackground}$. Furthermore, $\Held{\vecPSW}_{\vecLBackground}$, $\Held{\vecPSW}_{\vecMBackground}$, $\Held{\vecPSW}_{\vecMbBackground}$, $\Held{\vecPSW}_{\vecNBackground}$ can be chosen so that the linear trace condition holds. 

It is convenient for us to take a slightly different perspective on imposing the \linearTraceCondition. This is based on considering the initial value problem for the second-order ODE satisfied by the linearized trace, rather than analyzing the general solution in terms of $\Held{a}$ and $\Held{b}$. A linearized gauge transformation takes $h_{ab}$ to $h_{ab}+\Lie_\vecPSW\metricBackground_{ab}$. Thus, to impose the trace condition, it is sufficient to be able to specify $\mathring{g}^{ab} \mathring{\nabla}_{(a}\vecPSW_{b)}$. As noted in \cite{Price:2006ke}, this trace is given by
\begin{subequations}
\begin{align}
\mathring{g}^{ab} \mathring{\nabla}_{(a}\vecPSW_{b)}={}&- \vecPSW_{\mathring{n}} (\mathring{\rho} + \bar{\mathring{\rho}})
 -  \vecPSW_{\mathring{l}} (\mathring{\rho}' + \bar{\mathring{\rho}}')
 -  \edtBG \vecPSW_{\vecMbBackground}
 -  \edtpBG \vecPSW_{\vecMBackground}. 
\end{align}
Assuming that $\vecPSW$ satisfies \eqref{eq:NullGaugevcondGHP}, the derivative along $\vecN$ can be calculated as
\begin{align}
\thopBG (\mathring{g}^{ab} \mathring{\nabla}_{(a}\vecPSW_{b)})={}& \vecPSW_{\mathring{n}} (\mathring{\Psi}_{2}{} + \bar{\mathring{\Psi}}_{2}{} - 2 \mathring{\rho} \bar{\mathring{\rho}}')
 -  \vecPSW_{\mathring{l}} (\mathring{\rho}'^2 + \bar{\mathring{\rho}}'^2)
 + 2 \vecPSW_{\vecMBackground} \mathring{\rho}' \mathring{\tau}'
 + 2 \vecPSW_{\vecMbBackground} \bar{\mathring{\rho}}' \bar{\mathring{\tau}}'\nonumber\\
& +  (\mathring{\rho}' + \bar{\mathring{\rho}}') \thoBG \vecPSW_{\mathring{n}}
 -  (\mathring{\rho}' -  \bar{\mathring{\rho}}') \edtBG \vecPSW_{\vecMbBackground}
 +  \edtBG \edtpBG \vecPSW_{\mathring{n}}
 +  (\mathring{\rho}' -  \bar{\mathring{\rho}}') \edtpBG \vecPSW_{\vecMBackground}
 +  \edtpBG \edtBG \vecPSW_{\mathring{n}}.
\end{align}
\end{subequations}
Applying the general solution of the \LnhCondition{} and further calculation leads to the pair of equations
\begin{subequations}
\label{eq:PSWHeldab}
\begin{align}
\hspace{6ex}&\hspace{-6ex}- \frac{\mathring{g}^{ab} \mathring{\nabla}_{(a}\vecPSW_{b)}}{4 \kappa_{1}{} \bar{\kappa}_{1'}{}} (\kappa_{1}{}^2 + \bar{\kappa}_{1'}{}^2)
 + \frac{\thopBG (\mathring{g}^{ab} \mathring{\nabla}_{(a}\vecPSW_{b)})}{4 \kappa_{1}{} \mathring{\rho}'} (\kappa_{1}{} + \bar{\kappa}_{1'}{})\nonumber\\
={}&\Heldtho \Held{\vecPSW}_{\mathring{n}}
 + \tfrac{1}{2} \Heldeth \Held{\vecPSW}_{\vecMbBackground}
 + \tfrac{1}{2} \Heldethp \Held{\vecPSW}_{\vecMBackground},
\label{eq:PSWHeldab:ForN}\\
\hspace{6ex}&\hspace{-6ex}\frac{ (\kappa_{1}{} -  \bar{\kappa}_{1'}{})^2\mathring{g}^{ab} \mathring{\nabla}_{(a}\vecPSW_{b)}}{4 \kappa_{1}{}^2 \bar{\kappa}_{1'}{} \mathring{\rho}'} (\kappa_{1}{} + \bar{\kappa}_{1'}{})
 -  \frac{\thopBG (\mathring{g}^{ab} \mathring{\nabla}_{(a}\vecPSW_{b)})}{4 \kappa_{1}{}^2 \mathring{\rho}'^2} (\kappa_{1}{}^2 + \bar{\kappa}_{1'}{}^2)\nonumber\\
={}&\Held{\vecPSW}_{\mathring{l}}
 -  \tfrac{1}{2} \Heldeth \Heldethp \Held{\vecPSW}_{\mathring{n}}
  -  \tfrac{1}{2} \Heldethp \Heldeth \Held{\vecPSW}_{\mathring{n}}
   + \frac{(\kappa_{1}{} -  \bar{\kappa}_{1'}{})}{2 \kappa_{1}{} \mathring{\rho}'}( \Heldethp \Held{\vecPSW}_{\vecMBackground}-\Heldeth \Held{\vecPSW}_{\vecMbBackground})
\nonumber\\
& + \frac{\Held{\vecPSW}_{\mathring{n}}}{2 \mathring{\rho}'^2} \Bigl( \frac{\kappa_{1}{}}{\bar{\kappa}_{1'}{}}\mathring{\Psi}_{2}{} + \mathring{\Psi}_{2}{} + 2  \mathring{\rho} \mathring{\rho}' - 2 \mathring{\tau} \mathring{\tau}'\Bigr),
\label{eq:PSWHeldab:ForL}
\end{align}
\end{subequations}
where the Killing spinor coefficient $\kappa_{1}$ is given in equation \eqref{eq:defKappa1}. 
The right-hand sides of these two equations loosely correspond to the quantities $\Held{a}$ and $\Held{b}$ from \cite{Price:2006ke}. 
Set $\Held{\vecPSW}_{\vecMBackground}=0$ and $\Held{\vecPSW}_{\vecMbBackground}=0$. Set $\Held{\vecPSW}_{\vecNBackground}$ to satisfy the analogue of \eqref{eq:PSWHeldab:ForN} where $g^{ab}\mathring{\nabla}_{(a}\vecPSW_{b)}$ and its $\thopBG$ derivative have been replaced by $\frac{1}{2}\metricBackground^{ab}h_{ab}$  and its $\thopBG$ derivative on an initial hypersurface $\Sigma$. In a similar way, set $\Held{\vecPSW}_{\vecLBackground}$ to satisfy the analogue of \eqref{eq:PSWHeldab:ForL}. From this choice of $\Held{\vecPSW}$, set $\vecPSW$ to be the corresponding general solution of the \LnhCondition. This has been chosen so that the trace of $h_{ab}+\Lie_\vecPSW\metricBackground_{ab}$ and the $\thopBG$ derivative of this trace both vanish on the initial hypersurface $\Sigma$. From the second-order, linear ODE that it satisfies, the trace remains zero. This imposes the \linearTraceCondition, and hence the \ClassicalORG.

\subsection{Proof of theorem \ref{thm:NLORGEnforceability}}

\begin{proof}[Proof of theorem \ref{thm:NLORGEnforceability}]
\begin{steps}
\step{Preliminaries.}
Note that the first point of the theorem is simply a restatement of theorem \ref{thm:NLnhEnforceability}. Thus, we may assume that a diffeomorphism gauge has already been chosen to impose that result. Within the proof, we will impose a pair of further diffeomorphism gauges. The first will impose the trace condition to quadratic order while potentially violating the \NLnhCondition, and the second will reimpose the \NLnhCondition{} while preserving the quadratic smallness of the trace term. 

We assume the hypotheses of the theorem and initially consider what can be uniformly controlled. By taking $\gaugeRegularityIn$ sufficiently large with respect to $\gaugeRegularityOut$, there is a constant $K$ such that $|\Riem[\metric]-\Riem[\metricBackground]|_{C^{\gaugeRegularityOut}(U)}$ $\leq K|\metric-\metricBackground|_{C^{\gaugeRegularityIn}(U)}$. In this case, we can take $\gaugeRegularityIn=\gaugeRegularityOut+2$, but this illustrates that to control any quantity to desired regularity $\gaugeRegularityOut$, we can choose $\gaugeRegularityIn$ sufficiently large. We will use the notation $\varepsilon=|\metric-\metricBackground|_{C^{\gaugeRegularityIn}(U)}$ and, for an exponent $p$, the notation $\alpha=\beta+O(\varepsilon^p)$ to indicate that there is a constant $K$, possibly depending on the open sets and regularity constants $\gaugeRegularityIn$ and $\gaugeRegularityOut$, such that $|\alpha-\beta|_{C^{\gaugeRegularityOut}(V)}\leq K|\metric-\metricBackground|_{C^{\gaugeRegularityIn}(U)}^p$. We use $\alpha$ is $O(\varepsilon^p)$ to mean $\alpha=0+O(\varepsilon^p)$. 

Within this proof, we shall use the ``noncurvature quantities'' to refer to the differential Lorentz transforms, metric, and spin coefficient components. The geometric variables as given before the diffeomorphism gauge is applied are called the geometric variables in the original gauge; the geometric variables after the diffeomorphism gauge has been applied are called the the regauged geometric variables. 

There are three subtleties to address in this proof, all of which are resolved through the use of the smallness of the norms. The first subtlety is that, when constructing the diffeomorphisms, it is necessary that the image of $V$ remains in $U$. 

The first diffeomorphism is generated by the flow along a vector field, and the image property is ensured by the $\varepsilon$ smallness of this vector field. The second diffeomorphism is generated using the argument from the geodesic flow from section \ref{s:ProofOfEnforecabilityOfTheNLORG}, and the image property is ensured by $\varepsilon$ smallness of the perturbation of the initial data in the geodesic flow.

The second subtlety is that the domain $V$ depends on the norm of the geometric variables, but the $C^{\gaugeRegularityOut}(V)$ norm of the geometric variables depends on the choice of $V$. The regauged noncurvature quantities satisfy transport equations that are driven by both the regauged noncurvature quantities and the regauged $\tilde{\Psi}_i$; the regauged noncurvature quantities are determined by this evolution, while the $\tilde{\Psi}_i$ can be viewed as being calculated from the curvatures $\Riem[\metric]$ and $\Riem[\metricBackground]$ in the original diffeomorphism gauge and from the regauged foreground tetrad, which is determined by the regauged differential Lorentz transformation variables. Since the curvatures in the original gauge are already given on $U$, the regauged noncurvature quantities can be determined from the transport equations from their initial data and from the curvature in the original gauge. Since the regauged noncurvature quantities are $\varepsilon$ small on the initial hypersurface, $\heightInNullGaugeEnforceability(X)$, it is possible to pass to a subset $\heightInNullGaugeEnforceability(Y)$ so that both the image under the transport equations remains in $V$ and the regauged noncurvature quantities remain $\varepsilon$ small on $V$, provided that the regauged $\tilde{\Psi}_i$ remain $\varepsilon$ small.

The third subtlety is that, a priori, the $C^{\gaugeRegularityIn}(U)$ norm of $\Riem[\metric]-\Riem[\metricBackground]$ need not control the $C^{\gaugeRegularityIn}(U)$ of the $\tilde{\Psi}_i$ because there is not an a priori bound on the lengths of the foreground tetrad with respect to the reference Riemannian metric used to define the $C^{\gaugeRegularityIn}$ norms. This third subtlety is resolved by observing that as long as the Lorentz transformation variables remain bounded, the norms of the regauged $\tilde{\Psi}_i$ are controlled by the corresponding norms of $\Riem[\metric]-\Riem[\metricBackground]$ and the norms of the Lorentz transformation variables. Since the proof shows that the Lorentz transformation variables remain $\varepsilon$ small, we trivially recover the bootstrap assumption that they are bounded for the third subtlety, which then provides the necessary conditions for the second and first subtlety to be resolved. 

Within this proof, we shall define a Held spinor to be a Held spinor with respect to $\metricBackground$ and again use the notation $\Held{\alpha}$ to denote that $\alpha$ is a Held spinor. 

\step{Define $\vecPSW$.} 
Set $\Held{\vecPSW}_{\vecM}=0$ and $\Held{\vecPSW}_{\vecMb}=0$. Set $\Held{\vecPSW}_{\vecN}$ to satisfy the analogue of \eqref{eq:PSWHeldab:ForN} where, on the initial hypersurface $h(X)$, the quantities $g^{ab}\mathring{\nabla}_{(a}v_{b)}$ and its $\thop$ derivative have been replaced by $\frac{1}{2}\metricBackground^{ab}\metric_{ab}-2$ and its $\thop$ derivative respectively. In a similar way, set $\Held{\vecPSW}_{\vecL}$ to satisfy the analogue of \eqref{eq:PSWHeldab:ForL}. From this choice of $\Held{\vecPSW}$, set $\vecPSW$ to be the corresponding general solution of the \LnhCondition{} given in equation \eqref{eq:PSW23}. This has been chosen so that the trace of $\metric_{ab}+\Lie_\vecPSW\metricBackground_{ab}$ and the $\thop$ derivative of this trace both vanish on the initial hypersurface $h(X)$. Note that from the smallness of $\metric$, the components of $\vecPSW$ are $O(\varepsilon)$.

\step{Construct an initial gauge transformation from the flow along $\vecPSW$.}
Recall $\Flow{\vecPSW}{s}(p)$ denotes the flow along $\vecPSW$, and that this defines a local diffeomorphism. For simplicity, denote by $\Phi_1$ the diffeomorphism such that $\Phi_1(p)=\Flow{\vecPSW}{1}(p)$ for all $p$ for which this is defined. In particular, if $\varepsilon$ is sufficiently small on a scale dictated by $U$ and $V$, then $\Phi_1$ will define a bijection from $V$ to a subset of $U$. Since $\vecPSW$ and $\metric-\metricBackground$ are $O(\varepsilon)$, it follows that $\Phi_1^{*}\metric$ is also $O(\varepsilon)$.

Since $\vecPSW$ and $\metric-\metricBackground$ are both $O(\varepsilon)$, it follows that $\Lie_\vecPSW\metric_{ab}-\Lie_\vecPSW\metricBackground_{ab}$ is $O(\varepsilon^2)$. From the Price-Shankar-Whiting lemma on the linear theory, it follows that $\Lie_\vecPSW\metricBackground_{ab}$ satisfies the \LnhCondition, so $\vecN^a\Lie_\vecPSW\metric_{ab}$ is $O(\varepsilon^2)$. Similarly, on the initial hypersurface $h(X)$, the vector field $\vecPSW$ was chosen so that $\metricBackground^{ab}\Lie_{\vecPSW}(\metric-\metricBackground)_{ab}$ and its $\thop$ derivative vanish. Thus, on the image of $h(X)$, they are $O(\varepsilon^2)$. From the transport equations \eqref{eq:ThopInvGTrTilde} and \eqref{eq:transport:rhoTildePrime} satisfied by $\InvGTrTilde$ and $\tilde{\rho}'$, it follows that the perturbed trace $\InvGTrTilde$ satisfies a second-order ODE in which all the terms that appear are either linear in $\InvGTrTilde$ or of size $O(\varepsilon^2)$. Since the initial data is $O(\varepsilon^2)$ on the image of $h(X)$, this means that $\tilde{\slashed{G}}^{\#}$ remains $O(\varepsilon^2)$. Thus, $\metricBackground^{ab}\Phi_{1}^{*}\metric_{ab}-4$ is also $O(\varepsilon^2)$.

\step{Reimpose the radiation gauge.}
From the enforceability of the \NLnhCondition{} in theorem \ref{thm:NLnhEnforceability}, it follows that there is a local diffeomorphism $\Phi_2$ such that $\Phi_2^{*}(\Phi_{1}^{*}\metric)$ satisfies the \NLnhCondition. From the previous step, we know that $\Phi_{1}^{*}\metric$ is already very close to satisfying the \NLnhCondition. In particular, following the proof of the enforceability of the \NLnhCondition{} in section \ref{s:ProofOfEnforecabilityOfTheNLORG}, one observes that the size of $\Phi_2^{*}(\Phi_{1}^{*}\metric)-\Phi_{1}^{*}\metric$ depends not on the size of all components of $\Phi_{1}^{*}\metric-\metricBackground$, but only upon the size of the components of $\vecN^a(\Phi_{1}^{*}\metric)_{ab}$. From this, it follows that $\Phi_2^{*}(\Phi_{1}^{*}\metric) -\Phi_{1}^{*}\metric$ is $O(\varepsilon^2)$. In particular, $\metricBackground^{ab}(\Phi_2^{*}\Phi_{1}^{*}\metric)_{ab}-4$ is $O(\varepsilon^2)$. Defining $\Phi^{*}=\Phi_2^{*}\circ\Phi_{1}^{*}$, one obtains a $C^{\gaugeRegularityOut}$ diffeomorphism of $V$ to a subset of $U$.
This completes the proof. 
\end{steps}
\end{proof}

\section{Linearization}
\label{s:linearization}

In this section, we begin by linearizing the results in this paper in subsection \ref{sec:linearizeThisPaper}, 
then compare with our previous results in \cite{Andersson:2019dwi} in subsection \ref{sec:comparisonOfLinearization}, 
and conclude with some further remarks on how the initial data for the frame gauge can be used to set $\tilde{\beta}=0$ to linear order in subsection \ref{sec:SmallnessBeta}.

\subsection{Linearization of results in this paper} 
\label{sec:linearizeThisPaper}
The linearization of the systems considered in section \ref{s:FOSH} can now be computed. Dropping the nonlinear terms, one obtains the following result. 

\begin{theorem}[Linearization of the equations of section \ref{s:FOSH}]
\label{thm:linearizeSectionFOSH}
Assume the \backgroundHypotheses. 

The linearization of the relations between the different versions of the metric components yield
\begin{subequations}
\label{eq:TildeGToInvTildeG:linearized}
\begin{align}
\dot{\tilde{G}}^{\#}_2{}={}&\dot{G}_2{}
=- \dot{\tilde{G}}_2{}
=- \dot{G}^{\#}_2{},\\
\dot{\tilde{G}}^{\#}_1{}={}&\dot{G}_1{}
=- \dot{\tilde{G}}_1{}
=- \dot{G}^{\#}_1{},\\
\dot{\tilde{G}}^{\#}_0{}={}&\dot{G}_0{}
=- \dot{\tilde{G}}_0{}
=- \dot{G}^{\#}_0{},\\
\dot{\tilde{\slashed{G}}}{}^{\#}{}={}&\dot{\slashed{G}}{}
=- \dot{\tilde{\slashed{G}}}{}
=- \dot{\slashed{G}}{}^{\#}{}.
\end{align}
\end{subequations}
Furthermore, the linearization of the system \eqref{eq:nuetaevoleq}, \eqref{eq:thopmetric}, \eqref{eq:thopRicci}, and \eqref{eq:Bianchi} consists of 
\begin{subequations}
\label{eq:ThopFGDiffSpinLin}
\begin{align}
\thopBG \dot{\nu}={}&\tfrac{1}{4} i \LinGTr (\mathring{\rho}' -  \bar{\mathring{\rho}}') \label{eq:ThopFGDiffSpinLinEq1},\\
\thopBG \dot{\eta}={}&\dot{\tilde{\beta}}
 -  \overline{\dot{\tilde{\beta}}'}
 + \dot{\eta} \mathring{\rho}'
 + i \dot{\nu} \mathring{\tau}
 -  \tfrac{1}{4} \LinGTr \mathring{\tau}
 + \tfrac{1}{2} \LinGTwoDg \bar{\mathring{\tau}} \label{eq:ThopFGDiffLLinEq1},
\end{align}
\end{subequations}
\begin{subequations}
\label{eq:ThopLinG}
\begin{align}
\thopBG \LinGTwo={}&- \LinGTwo \mathring{\rho}'
 + \LinGTwo \bar{\mathring{\rho}}'
 + 2 \dot{\tilde{\sigma}}' \label{eq:ThopFGInvGTilde2LinEq1},\\
\thopBG \LinGTr={}&-2 (\dot{\tilde{\rho}}' + \overline{\dot{\tilde{\rho}}'}) \label{eq:ThopFGInvGTrTildeLinEq1},\\
\thopBG \LinGOne={}&-2 \overline{\dot{\eta}} \mathring{\rho}'
 - 2 \LinGOne \mathring{\rho}'
 + \LinGOne \bar{\mathring{\rho}}'
 -  \LinGTwo \mathring{\tau}
 + \tfrac{1}{2} \LinGTr \bar{\mathring{\tau}}
 + 2i \dot{\nu} \mathring{\tau}'
 + 2 \dot{\tilde{\tau}}' \label{eq:ThopFGInvGTilde1LinEq1},\\
\thopBG \LinGZero={}&-2 \dot{\tilde{\epsilon}}
 - 2 \overline{\dot{\tilde{\epsilon}}}
 - 2 \LinGOne \mathring{\tau}
 - 2 \LinGOneDg \bar{\mathring{\tau}}
 - 2 \dot{\eta} \mathring{\tau}'
 - 2 \overline{\dot{\eta}} \bar{\mathring{\tau}}' \label{eq:ThopFGInvGTilde0LinEq1},
\end{align}
\end{subequations}
\begin{subequations}
\label{eq:ThopLinConnection}
\begin{align}
\thopBG \dot{\tilde{\sigma}}'={}&\dot{\Psi}_{4}{}
 + \mathring{\rho}' \dot{\tilde{\sigma}}'
 + \bar{\mathring{\rho}}' \dot{\tilde{\sigma}}' \label{eq:ThopFGSigmapLinEq1},\\
\thopBG \dot{\tilde{\rho}}'={}&2 \mathring{\rho}' \dot{\tilde{\rho}}' \label{eq:ThopFGRhopLinEq1},\\
\thopBG \dot{\tilde{\tau}}'={}&\dot{\Psi}_{3}{}
 -  \dot{\tilde{\sigma}}' \mathring{\tau}
 -  \dot{\tilde{\rho}}' \bar{\mathring{\tau}}
 + \dot{\tilde{\rho}}' \mathring{\tau}'
 + \mathring{\rho}' \dot{\tilde{\tau}}'
 + \dot{\tilde{\sigma}}' \bar{\mathring{\tau}}' \label{eq:ThopFGTaupLinEq1},\\
\thopBG \dot{\tilde{\beta}}={}&\dot{\tilde{\beta}} \mathring{\rho}'
 - i \dot{\nu} \mathring{\rho}' \mathring{\tau}
 + \tfrac{1}{4} \LinGTr \mathring{\rho}' \mathring{\tau}
 + \dot{\tilde{\rho}}' \mathring{\tau} \label{eq:ThopFGBetaLinEq1},\\
\thopBG \dot{\tilde{\beta}}'={}&\dot{\Psi}_{3}{}
 + \dot{\tilde{\beta}}' \bar{\mathring{\rho}}'
 + \tfrac{1}{2} \LinGTwo \mathring{\rho}' \mathring{\tau}
 -  \dot{\tilde{\sigma}}' \mathring{\tau} \label{eq:ThopFGBetapLinEq1},\\
\thopBG \dot{\tilde{\epsilon}}={}&- \dot{\Psi}_{2}{}
 -  \overline{\dot{\eta}} \mathring{\rho}' \mathring{\tau}
 -  \LinGOne \mathring{\rho}' \mathring{\tau}
 + \dot{\tilde{\beta}} (- \bar{\mathring{\tau}} + \mathring{\tau}')
 + \mathring{\tau} \dot{\tilde{\tau}}'
 + \dot{\tilde{\beta}}' (\mathring{\tau} -  \bar{\mathring{\tau}}') \label{eq:ThopFGEpsilonLinEq1},\\
\thopBG \dot{\tilde{\rho}}={}&- \dot{\Psi}_{2}{}
 + \mathring{\rho} \overline{\dot{\tilde{\rho}}'}
 + \dot{\tilde{\rho}} \bar{\mathring{\rho}}'
 + \overline{\dot{\tilde{\beta}}} \mathring{\tau}
 + \dot{\tilde{\beta}}' \mathring{\tau}
 + 2 \overline{\dot{\eta}} \mathring{\rho}' \mathring{\tau}
 + \tfrac{1}{2} \LinGTwo \mathring{\tau}^2\nonumber\\
& -  \frac{i}{8 \kappa_{1}{}} (4 \dot{\nu} - i \LinGTr) \bigl(\mathring{\Psi}_{2}{} \kappa_{1}{} -  \bar{\mathring{\Psi}}_{2}{} \bar{\kappa}_{1'}{} + 2 \kappa_{1}{} (\mathring{\rho} \mathring{\rho}' -  \mathring{\rho} \bar{\mathring{\rho}}' + \mathring{\tau} \mathring{\tau}')\bigr) \label{eq:ThopFGRhoLinEq1},\\
\thopBG \dot{\tilde{\sigma}}={}&\mathring{\rho}' \dot{\tilde{\sigma}}
 + \mathring{\rho} \overline{\dot{\tilde{\sigma}}'}
 -  \dot{\tilde{\beta}} \mathring{\tau}
 -  \overline{\dot{\tilde{\beta}}'} \mathring{\tau}
 + 2 \dot{\eta} \mathring{\rho}' \mathring{\tau}
 + i \dot{\nu} \mathring{\tau}^2
 -  \tfrac{1}{4} \LinGTr \mathring{\tau}^2\nonumber\\
& + \frac{\LinGTwoDg}{4 \kappa_{1}{}} \bigl(\mathring{\Psi}_{2}{} \kappa_{1}{} -  \bar{\mathring{\Psi}}_{2}{} \bar{\kappa}_{1'}{} + 2 \kappa_{1}{} (\mathring{\rho} \mathring{\rho}' -  \mathring{\rho} \bar{\mathring{\rho}}' + \mathring{\tau} \mathring{\tau}')\bigr) \label{eq:ThopFGSigmaLinEq1},\\
\thopBG \dot{\tilde{\kappa}}={}&- \dot{\Psi}_{1}{}
 -  \dot{\tilde{\epsilon}} \mathring{\tau}
 + \overline{\dot{\tilde{\epsilon}}} \mathring{\tau}
 -  \dot{\tilde{\rho}} \mathring{\tau}
 -  \LinGZero \mathring{\rho}' \mathring{\tau}
 + \overline{\dot{\eta}} \mathring{\tau}^2
 + \LinGOne \mathring{\tau}^2
 + \dot{\tilde{\sigma}} (- \bar{\mathring{\tau}} + \mathring{\tau}')\nonumber\\
& + \frac{1}{2 \kappa_{1}{}} (\dot{\eta} + \LinGOneDg) \bigl(\mathring{\Psi}_{2}{} \kappa_{1}{} -  \bar{\mathring{\Psi}}_{2}{} \bar{\kappa}_{1'}{} + 2 \kappa_{1}{} (\mathring{\rho} \mathring{\rho}' -  \mathring{\rho} \bar{\mathring{\rho}}' + \mathring{\tau} \mathring{\tau}')\bigr)
 + \mathring{\rho} \overline{\dot{\tilde{\tau}}'}
 + \dot{\tilde{\rho}} \bar{\mathring{\tau}}' \label{eq:ThopFGKappaLinEq1},
\end{align}
\end{subequations}
\begin{subequations}
\label{eq:LinBianchi}
\begin{align}
\thoBG \dot{\Psi}_{1}{} -  \edtpBG \dot{\Psi}_{0}{}={}&-3 \mathring{\Psi}_{2}{} \dot{\tilde{\kappa}}
 + 4 \dot{\Psi}_{1}{} \mathring{\rho}
 -  \dot{\Psi}_{0}{} \mathring{\tau}',\\
\thoBG \dot{\Psi}_{2}{} -  \edtpBG \dot{\Psi}_{1}{}={}&3 \dot{\Psi}_{2}{} \mathring{\rho}
 + 3 \mathring{\Psi}_{2}{} \dot{\tilde{\rho}}
 + \tfrac{3}{2} \LinGZero \mathring{\Psi}_{2}{} \mathring{\rho}'
 - 3 \overline{\dot{\eta}} \mathring{\Psi}_{2}{} \mathring{\tau}
 - 3 \LinGOne \mathring{\Psi}_{2}{} \mathring{\tau}
 - 3 \dot{\eta} \mathring{\Psi}_{2}{} \mathring{\tau}'
 - 3 \LinGOneDg \mathring{\Psi}_{2}{} \mathring{\tau}'
 - 2 \dot{\Psi}_{1}{} \mathring{\tau}',\\
\thoBG \dot{\Psi}_{3}{} -  \edtpBG \dot{\Psi}_{2}{}={}&2 \dot{\Psi}_{3}{} \mathring{\rho}
 + 3 \overline{\dot{\eta}} \mathring{\Psi}_{2}{} \mathring{\rho}'
 + \tfrac{3}{2} \LinGTwo \mathring{\Psi}_{2}{} \mathring{\tau}
 - 3i \dot{\nu} \mathring{\Psi}_{2}{} \mathring{\tau}'
 -  \tfrac{3}{4} \LinGTr \mathring{\Psi}_{2}{} \mathring{\tau}'
 - 3 \dot{\Psi}_{2}{} \mathring{\tau}'
 - 3 \mathring{\Psi}_{2}{} \dot{\tilde{\tau}}',\\
\thoBG \dot{\Psi}_{4}{} -  \edtpBG \dot{\Psi}_{3}{}={}&\dot{\Psi}_{4}{} \mathring{\rho}
 + 3 \mathring{\Psi}_{2}{} \dot{\tilde{\sigma}}'
 - 4 \dot{\Psi}_{3}{} \mathring{\tau}',\\
\thopBG \dot{\Psi}_{0}{} -  \edtBG \dot{\Psi}_{1}{}={}&\dot{\Psi}_{0}{} \mathring{\rho}'
 + 3 \mathring{\Psi}_{2}{} \dot{\tilde{\sigma}}
 - 4 \dot{\Psi}_{1}{} \mathring{\tau},\\
\thopBG \dot{\Psi}_{1}{} -  \edtBG \dot{\Psi}_{2}{}={}&3 \dot{\eta} \mathring{\Psi}_{2}{} \mathring{\rho}'
 + 2 \dot{\Psi}_{1}{} \mathring{\rho}'
 + 3i \dot{\nu} \mathring{\Psi}_{2}{} \mathring{\tau}
 -  \tfrac{3}{4} \LinGTr \mathring{\Psi}_{2}{} \mathring{\tau}
 - 3 \dot{\Psi}_{2}{} \mathring{\tau}
 + \tfrac{3}{2} \LinGTwoDg \mathring{\Psi}_{2}{} \mathring{\tau}',\\
\thopBG \dot{\Psi}_{2}{} -  \edtBG \dot{\Psi}_{3}{}={}&3 \dot{\Psi}_{2}{} \mathring{\rho}'
 + 3 \mathring{\Psi}_{2}{} \dot{\tilde{\rho}}'
 - 2 \dot{\Psi}_{3}{} \mathring{\tau},\\
\thopBG \dot{\Psi}_{3}{} -  \edtBG \dot{\Psi}_{4}{}={}&4 \dot{\Psi}_{3}{} \mathring{\rho}'
 -  \dot{\Psi}_{4}{} \mathring{\tau}.
\end{align}
\end{subequations}

Furthermore, the linearization of the constraint equations \eqref{eq:AlgebraicSpincoeff1}, \eqref{eq:StructureSpincoeff1}, and \eqref{eq:ExtraRicci} are the linearized constraint equations \eqref{eq:extraLinearizedStructure} and \eqref{eq:extraLinearizedRicci}. 
\end{theorem}

One of the central goals for this paper is to construct a gauge condition for the Einstein equation with the property that its linearization is the \LnhCondition{} and, furthermore, to make a more restrictive gauge choice with the property that its linearization is the \ClassicalORG, which has long been studied by, for example, \cite{Chrzanowski, Price:2006ke}. Linearizing the gauge conditions constructed in this paper, one finds they have these desired properties, as stated in the following theorem. 

\begin{theorem}[Linearization of the radiation gauge and the trace condition]
\label{thm:linearizationOfRadiationGauge}
Assume the \backgroundHypotheses. 

\begin{enumerate}
\item The linearization of the \NLnhCondition{} \eqref{eq:NLnh} for a metric $\metric$ is the \LnhCondition{} \eqref{eq:linearnh}. 
\item The linearization of the gauge transformation in theorem \ref{thm:NLORGEnforceability} is the condition $\LinGTr=0$. 
\item The linearization of the combination of the \NLnhCondition{} \eqref{eq:NLnh} and the residual gauge transformation in theorem \ref{thm:NLORGEnforceability} is the \ClassicalORG{} of definition \ref{def:linearORG}. 
\end{enumerate}
\end{theorem}

From the evolution equations \eqref{eq:ThopFGRhopLinEq1} and \eqref{eq:ThopFGInvGTrTildeLinEq1} for $\dot{\tilde{\rho}}'$ and $\LinGTr$ and also from the evolution equation \eqref{eq:ThopFGDiffSpinLinEq1} for $\dot{\nu}$, one obtains the following result.
\begin{theorem}[Invariant subspaces]
Consider the linear system given in theorem \ref{thm:linearizeSectionFOSH}. 
\begin{enumerate}
\item The set $\{(\LinGTr ,\dot{\tilde{\rho}}')=(0,0)\}$ is an invariant subspace. 
\item The set $\{(\LinGTr ,\dot{\tilde{\rho}}',\dot{\nu})=(0,0,0)\}$ is an invariant subspace. 
\end{enumerate} 
\end{theorem}
\begin{remark}
There is an alternative perspective on the first part of the previous theorem based on linearizing the equations in section \ref{s:FOSH} subject to the gauge transformation \ref{thm:NLORGEnforceability}. The previous theorem states that the set $\{(\LinGTr ,\dot{\tilde{\rho}}')=(0,0)\}$ is an invariant subspace for the linearization of the evolution equations in section \ref{s:FOSH}. Alternatively, instead of first linearizing and then restricting to a subspace, one can first restrict the nonlinear system in section \ref{s:FOSH} via the gauge transformation in theorem \ref{thm:NLORGEnforceability} and then linearize. In the latter case, one obtains from the linearization of the gauge condition in theorem \ref{thm:NLORGEnforceability} that $\LinGTr=0$ and hence, via equations \eqref{eq:ThopFGRhopLinEq1} and \eqref{eq:algebraicrho1}, also $\dot{\tilde{\rho}}'=0$. 
\end{remark}

The invariant subspaces $\{(\LinGTr ,\dot{\tilde{\rho}}')=(0,0)\}$ and $\{(\LinGTr ,\dot{\tilde{\rho}}',\dot{\nu})=(0,0,0)\}$ are also stable, which can be seen by the following argument. 
The transport equations governing $\LinGTr$ and $\dot{\tilde{\rho}}'$ are
\begin{subequations} \label{eq:6.2}
\begin{align}
\thopBG \LinGTr={}&-2 (\dot{\tilde{\rho}}' + \overline{\dot{\tilde{\rho}}'}), \label{eq:5.6a}\\
\thopBG \dot{\tilde{\rho}}'={}&2 \mathring{\rho}' \dot{\tilde{\rho}}' . 
\end{align}
\end{subequations}
One can treat these ODEs using the method for proving decay of solutions of transport equations that was introduced in \cite{Andersson:2019dwi}. Qualitatively, the argument proceeds as follows. One works in Boyer-Lindquist coordinates and wishes to prove decay of a variable in $t$ for fixed $r$, assuming that the solution decays rapidly on the initial surface $\{t=0\}$ as $r\rightarrow\infty$. For a transport equation of the form $\thop\varphi=0$, one has that the value of $\varphi$ at $(t_1,r_1,\omega_1)$ is equal to the value of $\varphi$ at the intersection of the initial hypersurface $\{t=0\}$ with null geodesic tangent to $\vecN$ going through $(t_1,r_1,\omega_1)$, which occurs at $t_0=0$ and $r_0-r_1$ is bounded above and below by positive multiples of $t_1$ for $t_1>1$. Thus, $|\varphi|$ decays in $t$ because the initial data decays in $r$. Similarly, for an equation of the form $\thop\varphi=c_1\mathring{\rho}'\varphi$, one can introduce an integrating factor $\mathring{\rho}'{}^{c_2}$, and the growth or decay arising from the change in value of this integrating factor can be more than compensated for if the decay of the initial data is sufficiently fast. Furthermore, for an inhomogeneous equation of the form $\thop\varphi=c_1\mathring{\rho}'\varphi+\vartheta$, applying the integrating factor and integrating, the contribution from integrating $\vartheta$ is like $t$ (the length of the integration along the geodesic) times the maximum of $|\vartheta|$ (the maximum on the geodesic), but if this also decays in $t+r$, then the additional factor of $t$ from the integration can be dominated by the decay in $t+r$, although the decay of $\varphi$ will be one power worse than that of $\vartheta$. Applying this method schematically, one sees that if $\dot{\tilde{\rho}}'$ decays rapidly as $r\rightarrow\infty$ on the initial hypersurface $\{t=0\}$, then it will also decay rapidly as $t\rightarrow\infty$ at fixed $r$. Integrating the transport equation for $\LinGTr$, one obtains that $\LinGTr$ also decays rapidly (although not quite as rapidly). In future work, we will investigate the quantitative behaviour. In doing so, we note that we will have at our disposal the diffeomorphism that allows us to set $\InvGTrTilde$ to vanish quadratically, which one might plausibly expect to allow one to show that such a diffeomorphism could be chosen so that $\InvGTrTilde$ and $\tilde{\rho}'$ vanish much more rapidly than the other metric and connection coefficients. As shown in \cite{Andersson:2019dwi}, these methods for proving decay apply not only in Boyer-Lindquist coordinates at fixed $r$ as $t\rightarrow\infty$, but also in hyperboloidal coordinates that allow for precise estimates near null infinity; such estimates are likely to be crucial for controlling nonlinear terms in the Einstein equation. The stability of the invariant subspace $\{(\LinGTr ,\dot{\tilde{\rho}}',\dot{\nu})=(0,0,0)\}$ follows from a similar analysis of equation \eqref{eq:ThopFGDiffSpinLinEq1}.

\subsection{Comparison with previous work}
\label{sec:comparisonOfLinearization}
In this subsection, we compare the linearization of the results in this paper to the results in \cite{Andersson:2019dwi}. 

This comparison is somewhat complicated by the different order in which certain operations are performed. The Einstein equation is inherently a tensorial equation, whereas the techniques we use in this paper to prove symmetric hyperbolicity and the techniques used in \cite{Andersson:2019dwi} to prove decay are both for systems of scalars. Thus, by projecting on tetrads, we convert from tensorial equations to systems of spin-weighted scalar equations, a process which we refer to as scalarization. In this section, we are interested in linearization. We are also interested in imposing either the nonlinear radiation gauge or its linearization, and further imposing the residual gauge of theorem \ref{thm:NLORGEnforceability} or its linearization $\LinGTr=0$. Therefore, we are applying scalarization, linearization, the (nonlinear or linear) radiation gauge, and the (nonlinear or linear) trace condition. So far, in this paper, we have applied the radiation gauge first, scalarized second, and then linearized. (The comments at the end of the previous subsection show that we obtain equivalent results whether we linearize first and then restrict to $\LinGTr=0$ or, alternatively, apply the gauge transformation in theorem \ref{thm:NLORGEnforceability} and then linearize.) In contrast, in \cite{Andersson:2019dwi} we linearize the Einstein equation first, scalarize second, and then impose the linearized radiation gauge and the linearized trace condition. Although it is possible to justify the switching of the order of imposing linearization, scalarization, and the imposition of the gauge conditions, the following theorem provides a more direct comparison. Because \cite{Andersson:2019dwi} linearizes before scalarizing, there is no need to introduce a foreground frame, so the linearized differential Lorentz transformations $(\dot{\eta},\dot{\nu})$ have no analogue in \cite{Andersson:2019dwi}. 

The following theorem provides a relation between the linearized equations in theorem \ref{thm:linearizeSectionFOSH} and those in \cite{Andersson:2019dwi}. 

\begin{theorem}[Comparison with linearization in \cite{Andersson:2019dwi}]
\label{thm:comparisonWithLinearizedGravityPaper}
Assume the \backgroundHypotheses. 

Using table \ref{table:comparison} to identify the linearized variables in \cite{Andersson:2019dwi} with the linear combination of linearized variables from subsection \ref{sec:linearizeThisPaper} restricted to the invariant subspace $(\LinGTr,\dot{\tilde{\rho}}',\dot{\nu})=(0,0,0)$, we have
that the full linearized sytem in theorem \ref{thm:linearizeSectionFOSH} including the constraint equations is equivalent to 
the system in \cite[Lemma A.1]{Andersson:2019dwi} together with the transport equation \eqref{eq:translatedEta} for $\dot{\eta}$. In particular the system in theorem \ref{thm:linearizeSectionFOSH} implies the system in \cite[Lemma A.1]{Andersson:2019dwi}.
\end{theorem} 

\begin{table}[tb]
\begin{center}
\begin{tabular}{|l|l|}
\hline
\cite{Andersson:2019dwi} & This paper\\
\hline
$G_{00'}$ & $\LinGZero$\\
$G_{10'}$ & $\LinGOne$\\
$G_{20'}$ &$\LinGTwo$\\
$G_{01'}$ & $\LinGOneDg$\\
$G_{02'}$ &$\LinGTwoDg$\\
$\tilde{\beta}$&$\dot{\tilde{\beta}}
 + \dot{\eta} \mathring{\rho}'
 + \tfrac{1}{2} \LinGOneDg \mathring{\rho}'$\\
$\tilde{\beta}'$&$\dot{\tilde{\beta}}'$\\
$\tilde{\epsilon}$&$\dot{\tilde{\epsilon}}
 + \dot{\eta} \mathring{\tau}'
 + \tfrac{1}{2} \LinGOneDg \mathring{\tau}'$\\
$\tilde{\kappa}$&$\dot{\tilde{\kappa}}
 -  \dot{\eta} \mathring{\rho}
 -  \LinGOneDg \mathring{\rho}
 + \tfrac{1}{2} \LinGZero \mathring{\tau}
 + \thoBG \dot{\eta}
 + \tfrac{1}{2} \thoBG \LinGOneDg$\\
$\tilde{\rho}$&$\dot{\tilde{\rho}}
 -  \overline{\dot{\eta}} \mathring{\tau}
 + \edtpBG \dot{\eta}
 + \tfrac{1}{2} \edtpBG \LinGOneDg$\\
$\tilde{\sigma}$&$
  \dot{\tilde{\sigma}}
- \tfrac{1}{2} \LinGTwoDg \mathring{\rho}
 -  \dot{\eta} \mathring{\tau}
 + \edtBG \dot{\eta}
 + \tfrac{1}{2} \edtBG \LinGOneDg$\\
$\tilde{\sigma}'$&$
\dot{\tilde{\sigma}}'
- \tfrac{1}{2} \LinGTwo \mathring{\rho}'$\\
$\tilde{\tau}'$&$
 \dot{\tilde{\tau}}'
- \overline{\dot{\eta}} \mathring{\rho}'
 -  \LinGOne \mathring{\rho}'$\\
$\vartheta \Psi_{4}{}$&$\dot{\Psi}_{4}$\\
$\vartheta \Psi_{3}{}$&$\dot{\Psi}_{3}$\\
$\vartheta \Psi_{2}{}$&$\dot{\Psi}_{2}$\\
$\vartheta \Psi_{1}{}$&$\dot{\Psi}_{1} -3 \dot{\eta} \mathring{\Psi}_{2}{}
 - \tfrac{3}{2} \LinGOneDg \mathring{\Psi}_{2}{}$\\
$\vartheta \Psi_{0}{}$&$\dot{\Psi}_{0}$\\
\hline
\end{tabular} 
\end{center}
\caption{Relation between linearized variables in \cite{Andersson:2019dwi} and in section \ref{sec:linearizeThisPaper}.} 
\label{table:comparison}
\end{table}

\begin{proof}
Under the change of variables given in table \ref{table:comparison} and applying the background GHP commutators to eliminate all second-order derivatives on the right-hand side, the system in theorem \ref{thm:linearizeSectionFOSH} restricted to $(\LinGTr,\dot{\tilde{\rho}}',\dot{\nu})=(0,0,0)$ is equivalent to the system \eqref{eq:translatedLinThop}-\eqref{eq:translatedLinBianchi}. Using the relations \eqref{eq:translatedEta}, \eqref{eq:translatedLinStructure} and \eqref{eq:translatedLinRicci} in the other equations one finds that the full system \eqref{eq:translatedLinThop}-\eqref{eq:translatedLinBianchi} is equivalent to the system consisting of \eqref{eq:translatedEta} for $\dot{\eta}$ together with the system in \cite[Lemma A.1]{Andersson:2019dwi} for the linearized variables in \cite{Andersson:2019dwi}. 
\end{proof}

\begin{remark}
In principle, restricting to the subset $(\LinGTr,\dot{\tilde{\rho}}')=(0,0)$ would work too, but this would make table \ref{table:comparison} more complicated. A non-zero $\dot{\nu}$ would just correspond to a spin rotation of the foreground frame with respect to the background frame.
\end{remark}

\begin{remark}
While the full system of evolution and constraint equations in theorem \ref{thm:linearizeSectionFOSH} is equivalent to the full system of evolution and constraint equations in \cite[Lemma A.1]{Andersson:2019dwi}, the evolution systems alone are not equivalent. This is due to the fact that different combinations of the constraint equations were added to the evolution system. These different combinations were added to get symmetric hyperbolicity in this paper and to have a convenient hierarchy for proving decay estimates in \cite{Andersson:2019dwi}. 
\end{remark}

\begin{remark}
It follows from the previous theorem and the results of \cite{Andersson:2019dwi} that the metric coefficients $(\LinGZero,\LinGOne,\LinGTwo)$ satisfy the decay estimates given in \cite{Andersson:2019dwi}. The decay of the remaining connection and curvature components in \cite{Andersson:2019dwi} can be computed from the decay of the metric coefficients. 
\label{rmk:linearizedDecay}
\end{remark}

Remark \ref{rmk:linearizedDecay} gives decay for all the variables except $\dot{\eta}$. As previously noted, the quantity $\dot{\eta}$ has no analogue in \cite{Andersson:2019dwi}, but $\dot{\eta}$ can be shown to converge to a limit. First, observe that $\dot\varsigma=\LinGTr/4$ and $\dot\varsigma^\#=-\LinGTr/4$, which can be treated as either being zero in the invariant subspace or as converging rapidly to zero because of the stability of the invariant subspace. Note that the background values of $\varsigma$ and $\varsigma^\#$ are both $1$.
The linearization of equation \eqref{eq:algebraicbeta1} is 
\begin{align}
\dot{\tilde{\beta}} -  \overline{\dot{\tilde{\beta}}'}={}&- \dot{\eta} \mathring{\rho}'
 -  \LinGOneDg \mathring{\rho}'
 + \dot{\eta} \bar{\mathring{\rho}}'
 + \LinGOneDg \bar{\mathring{\rho}}'
 + \tfrac{1}{2} \LinGTwoDg \mathring{\tau}'
 -  \overline{\dot{\tilde{\tau}}'}
 + i \dot{\nu} \bar{\mathring{\tau}}'
 -  \tfrac{1}{4} \LinGTr \bar{\mathring{\tau}}'.
\end{align}
Substituting this formula into equation \eqref{eq:ThopFGDiffLLinEq1}, one finds that the equations 
\eqref{eq:ThopFGDiffLLinEq1} becomes
\begin{align}
\thopBG \dot{\eta}={}&
 \dot{\eta} \bar{\mathring{\rho}}'
- \LinGOneDg (\mathring{\rho}' -  \bar{\mathring{\rho}}')
 + \tfrac{1}{2} \LinGTwoDg (\bar{\mathring{\tau}} + \mathring{\tau}')
 -  \overline{\dot{\tilde{\tau}}'}
 + i \dot{\nu} (\mathring{\tau} + \bar{\mathring{\tau}}')
 -  \tfrac{1}{4} \LinGTr (\mathring{\tau} + \bar{\mathring{\tau}}').
\end{align}
Assuming all the other linearized quantities on the right of the transport equation decay, then one obtains that $\dot{\eta}$ also decays as explained at the end of section \ref{sec:linearizeThisPaper}. An alternative treatment of $\dot{\eta}$ appears in the following subsection.

\subsection{Smallness of \texorpdfstring{$\tilde\beta$}{beta tilde}}
\label{sec:SmallnessBeta}

In this section, we show that it is possible to choose initial data for $\eta$ (or its linearization) so that $\tilde{\beta}$ vanishes to linear order in both the linear and nonlinear settings. While it may seem natural to choose initial data with $\dot{\eta}=0$, this does not propagate, even when $\LinGTr=0$.

First, consider the linearized setting. The previous subsection argues that we can choose $\LinGTr=0$, $\dot{\tilde{\rho}}'=0$ and $\dot{\nu}=0$. With these choices \eqref{eq:ThopFGBetaLinEq1} gives a homogeneous evolution equation for $\dot{\tilde{\beta}}$. Thus, if $\dot{\tilde{\beta}}$ can be chosen to be initially zero, it remains so. Assuming that the linearized versions of equations \eqref{eq:ThopInvGTilde1}, \eqref{eq:algebraicbeta1}, and \eqref{eq:ethnuTobeta} hold, it follows that the vanishing of $\dot{\tilde{\beta}}$ is equivalent to each of the following
\begin{subequations}
\label{eq:smallBEtaLinearEquations}
\begin{align}
\dot{\eta}={}&- \tfrac{1}{4} \mathring{\rho}'^{-1} (\thopBG{} + 3 \mathring{\rho}' - 2 \bar{\mathring{\rho}}')\LinGOneDg
 + \tfrac{1}{4} \mathring{\rho}'^{-1} (\edtpBG{} -  \bar{\mathring{\tau}} + \mathring{\tau}')\LinGTwoDg,\label{eq:smallBeta:eta}\\
\dot{\tilde{\tau}}'={}&\LinGOne \mathring{\rho}'
 + (\tfrac{1}{2} -  \tfrac{1}{4} \kappa_{1}{}^{-1} \bar{\kappa}_{1'}{}) (\thopBG{} - 2 \mathring{\rho}' -  \bar{\mathring{\rho}}')\LinGOne
 + \tfrac{1}{4} \kappa_{1}{}^{-1} \bar{\kappa}_{1'}{} (\edtBG{} -  \mathring{\tau} -  \bar{\mathring{\tau}}')\LinGTwo,\\
\dot{\tilde{\beta}}'={}&\tfrac{1}{4} (\thopBG{} + 2 \mathring{\rho}' -  \bar{\mathring{\rho}}')\LinGOne
 + \tfrac{1}{4} (\edtBG{} + \mathring{\tau} -  \bar{\mathring{\tau}}')\LinGTwo.
\end{align}
\end{subequations}
From the first of these, we see that $\dot{\eta}$ can be chosen initially so that $\dot{\tilde{\beta}}$ vanishes initially. Recall that the initial data for the metric, spin coefficients, and curvature components must satisfy constraint equations, so they cannot be all freely specified. Although equation \eqref{eq:ThopInvGTilde1} is an evolution equation, equations \eqref{eq:algebraicbeta1} and \eqref{eq:ethnuTobeta} can be viewed as constraints on the initial data, and these impose constraints in the linearization. Due to the fact that the evolution equation for $\dot{\tilde{\beta}}$ is homogeneous, it follows that $\dot{\tilde{\beta}}$ will remain zero, and hence that these three equations \eqref{eq:smallBEtaLinearEquations} will remain valid. Equation \eqref{eq:smallBeta:eta} can be used to estimate $\dot{\eta}$. 

Now consider the nonlinear case. In this case, given initial data for the metric and its derivatives, one is free to choose initial data for $\nu$ and $\eta$ in the frame gauge. Once initial data for $\nu$ has been chosen, it is possible to compute $\tilde{\sigma}'$ and $\tilde{\rho}'$ via equation \eqref{eq:rhoTildePrimeSigmaTildePrime} purely in terms of quantities defined with respect to the background tetrad and $\nu$, without having specified $\eta$. Thus, for example, one may choose the initial value for $\eta$ so that 
\begin{align}
\eta ={}&- \frac{3 \overline{G_{2}} G_{1} \varsigma^{\#}{}^3}{4 \varsigma^2 \mathring{\rho}'} (\mathring{\rho}' -  \bar{\mathring{\rho}}')
 +  \frac{\overline{G_{1}} \varsigma^{\#}{}^2}{2 \varsigma \mathring{\rho}'} (2 - 3 \varsigma^{\#}{} \varsigma) (\mathring{\rho}' -  \bar{\mathring{\rho}}')
 + \frac{\overline{G_{2}} \varsigma^{\#}{}^2 \tilde{\sigma}'}{8 \varsigma^3 \mathring{\rho}'} (\overline{G_{2}} G_{1} + 2 \overline{G_{1}} \varsigma^2)\nonumber\\
& -  \frac{\varsigma^{\#}{}^2 \overline{\tilde{\sigma}'}}{4 \varsigma \mathring{\rho}'} (\overline{G_{1}} G_{2} + 2 G_{1} \varsigma^2)
 -  \frac{\overline{G_{2}} \varsigma^{\#}{}}{4 \mathring{\rho}'} (\bar{\mathring{\tau}} -  \mathring{\tau}')
 + \frac{1}{4 \mathring{\rho}'} (2 \varsigma^{\#}{} \varsigma^2 - \varsigma^{\#}{} -  \varsigma ) (\mathring{\tau} -  \bar{\mathring{\tau}}')\nonumber\\
& -  \frac{\varsigma^{\#}{} (\thopBG{} + \mathring{\rho}' + 2 \tilde{\rho}')\overline{G_{1}}}{4 \mathring{\rho}'}
 - \frac{\varsigma^{\#}{}^2 \edtBG \slashed{G}_{}}{8 \varsigma \mathring{\rho}'} (1 - 2 \varsigma^{\#}{} \varsigma)
 +  \frac{G_{2} \varsigma^{\#}{}^2 \edtBG \overline{G_{2}}}{16 \varsigma^3 \mathring{\rho}'} (1 - 2 \varsigma^{\#}{} \varsigma)\nonumber\\
& -  \frac{\overline{G_{2}} \varsigma^{\#}{}^2 \edtBG G_{2}}{16 \varsigma^3 \mathring{\rho}'} (1 + 2 \varsigma^{\#}{} \varsigma)
 -  \frac{\overline{G_{2}} \varsigma^{\#}{}^3 \edtpBG \slashed{G}_{}}{8 \varsigma^2 \mathring{\rho}'}
 + \frac{\varsigma^{\#}{}^3 \edtpBG \overline{G_{2}}}{4 \mathring{\rho}'}
 + \frac{\overline{G_{2}}{}^2 \varsigma^{\#}{}^3 \edtpBG G_{2}}{16 \varsigma^4 \mathring{\rho}'} .
\end{align}
With this choice, and trivial initial data for $\nu$, it follows from equations \eqref{eq:ThopInvGTilde1}, \eqref{eq:algebraicbeta1}, and \eqref{eq:ethnuTobeta}, that the initial data for $\tilde{\beta}$, $\tilde{\tau}'$ and $\tilde{\beta}'$ takes the form
\begin{subequations}
\label{eq:smallBEtaNonLinearEquations}
\begin{align}
\tilde{\beta}={}&- \eta \tilde{\rho}'
 + \frac{\bar{\eta} \overline{G_{2}} \bar{\mathring{\rho}}'}{2 \varsigma^2} ,\\
\tilde{\tau}'={}&\bar{\eta} (\mathring{\rho}' + \tilde{\rho}')
 + \frac{\varsigma^{\#}{}}{4 \varsigma^2} (\overline{G_{1}} G_{2} + 2 G_{1} \varsigma^2) (\mathring{\rho}' + 2 \tilde{\rho}')
 + \tfrac{1}{2} G_{1} \varsigma^{\#}{} (\mathring{\rho}' -  \bar{\mathring{\rho}}')
 + \eta \tilde{\sigma}'
 + \frac{\varsigma^{\#}{} \tilde{\sigma}'}{2 \varsigma^2} (\overline{G_{2}} G_{1} + 2 \overline{G_{1}} \varsigma^2)\nonumber\\
& + \frac{G_{2} \mathring{\tau}}{4 \varsigma^2} (\varsigma^{\#}{} + \varsigma)
 + \tfrac{1}{2} (\varsigma^{\#}{} -  \varsigma) \bar{\mathring{\tau}}
 -  \tfrac{1}{2} (2 -  \varsigma^{\#}{} -  \varsigma) \mathring{\tau}'
 + \frac{G_{2} \bar{\mathring{\tau}}'}{4 \varsigma^2} (\varsigma^{\#}{} -  \varsigma)
 + \frac{G_{2} \varsigma^{\#}{} \thopBG \overline{G_{1}}}{4 \varsigma^2}
 + \tfrac{1}{2} \varsigma^{\#}{} \thopBG G_{1},\\
\tilde{\beta}'={}&\frac{\eta G_{2} \mathring{\rho}'}{2 \varsigma^2}
 -  \frac{\varsigma^{\#}{}^2}{2 \varsigma^3} (\overline{G_{1}} G_{2} + 2 G_{1} \varsigma^2) (\mathring{\rho}' -  \bar{\mathring{\rho}}')
 -  \bar{\eta} (\mathring{\rho}' + \tilde{\rho}' -  \bar{\mathring{\rho}}')
 + (1 -  \varsigma^{\#}{}) \mathring{\tau}'
 + \tilde{\tau}'
 -  \frac{G_{2} \varsigma^{\#}{} \bar{\mathring{\tau}}'}{2 \varsigma^2} .
\end{align}
\end{subequations}
In particular, $\tilde{\beta}$ vanishes quadratically. The choice of $\eta$ is not unique, in that there are other choices of $\eta$ for which $\tilde{\beta}$ also vanishes quadratically. These equations will not propagate under the evolution, although $\tilde{\beta}$ will remain quadratically small for evolution under equation \eqref{eq:thopbeta}. 

\subsection*{Acknowledgements}
A significant portion of the work was done while the authors were in residence at Institut Mittag-Leffler in Djursholm, Sweden during the fall of 2019, supported by the Swedish Research Council under grant no. 2016-06596. S. M. also acknowledges the support by the ERC grant ERC-2016 CoG 725589 EPGR.
The authors are grateful to Steffen Aksteiner and Bernard Whiting for helpful discussions. 
The authors are grateful to the anonymous referees for suggestions to improve the clarity of the paper. 

\appendix

\section{GHP formalism as a gauge or principal bundle theory}
\label{sec:worry} 

To begin, we recall the definition of a tetrad. 
\begin{definition}
Let $(\mathcal{M},\metric)$ be a $1+3$ dimensional, Lorentzian manifold with an orientation and time orientation. 

A \defn{real null tetrad} at each point is defined to consist of a pair of distinct null vectors $\vecL$ and $\vecN$ satisfying $\metric(\vecL, \vecN)=1$ and an orthonormal basis $(e_1,e_2)$ for the plane orthogonal to $\vecL$ and $\vecN$. A \defn{complex null tetrad} is defined to consist of a quadruple of elements of the complexification of the tangent space $(\vecL,\vecN,\vecM,\vecMb)$ such that they are null vectors and $\vecL,\vecN,(\vecM+\vecMb)/\sqrt{2},(\vecM-\vecMb)/(i\sqrt{2})$ is a real null tetrad. 

A real null tetrad is defined to be \defn{oriented} if $(\vecL,\vecN,e_1,e_2)$ is an oriented basis, and a complex null tetrad is oriented if the corresponding real null tetrad is. 

In this paper, unless otherwise specified, a \defn{tetrad} is understood to mean an oriented complex null tetrad. 

Given an ordered pair of (distinct and future-pointing) null vectors $(\vecLSpecified,\vecNSpecified)$ a tetrad $(\vecL,\vecN,\vecM,\vecMb)$ is defined to be \defn{aligned with $\vecLSpecified$ and $\vecNSpecified$} if $\vecL$ is a positive multiple of $\vecLSpecified$ and $\vecN$ is a positive multiple of $\vecNSpecified$ and defined to be \defn{constructed from $\vecLSpecified$ and $\vecNSpecified$} if $\vecL=\vecLSpecified$ and $\vecN=\vecNSpecified$. 

A \defn{local tetrad} is a smooth map from an open subset of $\mathcal{M}$ taking values, at each point, in the set of null tetrads at that point. 
\end{definition}

An important aspect of the GHP formalism \cite{GHP} is that it is designed specifically to handle the situation where there is a pair of null directions that is naturally singled out rather than a choice of global tetrad. This is particularly important where there is a globally defined pair of ingoing and outgoing null vectors but there is no globally defined tetrad that is aligned with this choice. The nonexistence of such a tetrad is most clearly visible in the Schwarzschild space-time, where a hypothetical $(\vecM+\vecMb)/\sqrt{2}$ and $(\vecM-\vecMb)/(i\sqrt{2})$ would specify a global basis for the tangent space of the spheres orthogonal to the radial ingoing and outgoing vectors, but no such global basis can exist, since it is known that the $2$-spheres do not have any globally non-vanishing vector fields. 

Nonetheless, some specification of local tetrads is required, and, to explain this, it is useful to use the language of principal-$G$ bundles or, equivalently, gauge theory. Regarding the patching of these local tetrads, since \cite[p269]{PenroseRindler} simply states ``this idea can be made mathematically more precise in the language of fiber bundles \ldots, but we need not elaborate on it here'', we briefly summarize the situation in this appendix. 

We begin by considering sets of null tetrads. Consider a pair of null vectors at a point, $\vecL$ and $\vecN$ that satisfy $\metric(\vecL,\vecN)=1$. One can choose an oriented orthonormal basis $(e_1,e_2)$ for the plane orthogonal to the plane spanned by $\vecL$ and $\vecN$. Any oriented real null tetrad aligned with $\vecL$ and $\vecN$ is of the form $(\lambda \vecL,\lambda^{-1}\vecN,\cos\varphi e_1+\sin\varphi e_2,\cos\varphi e_2-\sin\varphi e_1)$ with $\lambda\in(0,\infty),\varphi\in\Reals$ and hence uniquely specified by $\lambda e^{i\varphi}\in\Complex^*$, the set of invertible elements in $\Complex$. Similarly, any oriented real null tetrad constructed from $\vecL$ and $\vecN$ is of the form $(\vecL,\vecN,\cos\varphi e_1+\sin\varphi e_2,\cos\varphi e_2-\sin\varphi e_1)$ with $\varphi\in\Reals$ and hence uniquely specified by $e^{i\varphi}\in\Circle$. Real null tetrads are in one-to-one correspondence with (complex) null tetrads, by taking $\vecM=2^{-1/2}(e_1+ie_2)$. Now consider a pair of globally defined null vector fields, also denoted $\vecL$ and $\vecN$. In the language of principal-$G$ bundles, the set of tetrads aligned with $\vecL$ and $\vecN$ is a principal-$\Complex^*$ bundle, and the set of tetrads constructed from $\vecL$ and $\vecN$ is a principal-$\Circle$ bundle. In the language of gauge theory, these sets have $\Complex^*$ and $\Circle$ gauge groups respectively.

We now consider the notion of GHP scalar, which can be stated in various languages. The GHP scalars that typically arise in this paper can be viewed as $\Complex$-valued contractions of a tensor with elements of a local null tetrad or their derivatives; hence, they can be viewed as $\Complex$-valued functions on the the bundle of null tetrads or the jet bundles over it. An important type of GHP scalars are those that are properly weighted. In perhaps the simplest language, a GHP scalar is defined to be \textbf{properly weighted} if there are scalars $(b,s)$ such that if the tetrad $(\vecL,\vecN,\vecM,\vecMb)$ is transformed to $(\lambda\vecL,\lambda^{-1}\vecN,e^{i\varphi}\vecM,e^{-i\varphi}\vecMb)$ then the GHP scalar is transformed from $\alpha$ to $\lambda^b e^{is\varphi}\alpha$. The exponents $b$ and $s$ are the boost and spin weight. This definition can be expressed as being a function on the frame bundle and transforming equivariantly, as being a section of an associated complex line bundle for the null tetrad bundle, or as a gauge field associated with the set of null tetrads. Roughly speaking, these different characterizations are like characterizing tensors, on the one hand, by how their components transform under a change of basis, or, on the other hand, as tensor products of copies of the tangent and cotangent space. The GHP scalars that are not properly weighted but which arise in the standard presentation (that is $\beta,\beta',\epsilon,\epsilon'$) can be viewed as connection coefficients for a connection on these associated complex line bundles. For properly weighted scalars, it is conventional to use the $(p,q)$ weights such that $b=\frac12(p+q)$ and $s=\frac12(p-q)$. 

Although it doesn't arise in this paper, GHP spinors with noninteger boost and spin weight can be defined by exploiting spinor structures instead of just tensorial structures.

\section{GHP Commutators}\label{sec:commutators}
With our gauge choice we have the following commutator relations for the foreground operators acting on a field $\varphi$ with weight $(p,q)$ with respect to background spin and boost transformations. These are verified using the definition of the operators in terms of the background operators.
\begin{subequations}
\begin{align}
\thop \tho \varphi ={}&\tho \thop \varphi
 -  (\bar{\mathring{\tau}} -  \mathring{\tau}' -  \tilde{\tau}') \edt \varphi
 -  (\mathring{\tau} -  \overline{\tilde{\tau}'} -  \bar{\mathring{\tau}}') \edtp \varphi\nonumber\\
& -  \Bigl(p \bigl(- \mathring{\Psi}_{2}{} -  \tilde{\Psi}_{2}{} + \mathring{\tau} (\mathring{\tau}' + \tilde{\tau}')\bigr) + q \bigl(- \bar{\mathring{\Psi}}_{2}{} -  \overline{\tilde{\Psi}_{2}{}} + \bar{\mathring{\tau}} (\overline{\tilde{\tau}'} + \bar{\mathring{\tau}}')\bigr)\Bigr) \varphi ,\\
\thop \edt \varphi ={}&\edt \thop \varphi
 -  \mathring{\tau} \thop \varphi
 + (\mathring{\rho}' + \tilde{\rho}') \edt \varphi
 + \overline{\tilde{\sigma}'} \edtp \varphi
 + \bigl(- p (\mathring{\rho}' + \tilde{\rho}') \mathring{\tau} + q (\overline{\tilde{\Psi}_{3}{}} -  \overline{\tilde{\sigma}'} \bar{\mathring{\tau}})\bigr) \varphi ,\\
\thop \edtp \varphi ={}&\edtp \thop \varphi
 -  \bar{\mathring{\tau}} \thop \varphi
 + \tilde{\sigma}' \edt \varphi
 + (\overline{\tilde{\rho}'} + \bar{\mathring{\rho}}') \edtp \varphi
 + \bigl(p (\tilde{\Psi}_{3}{} -  \tilde{\sigma}' \mathring{\tau}) -  q (\overline{\tilde{\rho}'} + \bar{\mathring{\rho}}') \bar{\mathring{\tau}}\bigr) \varphi ,\\
\tho \edt \varphi ={}&\edt \tho \varphi
 -  (\overline{\tilde{\tau}'} + \bar{\mathring{\tau}}') \tho \varphi
 -  \tilde{\kappa} \thop \varphi
 + (\overline{\tilde{\rho}} + \bar{\mathring{\rho}}) \edt \varphi
 + \tilde{\sigma} \edtp \varphi\nonumber\\
& -  \Bigl(p \bigl(\tilde{\Psi}_{1}{} + \tilde{\kappa} (\mathring{\rho}' + \tilde{\rho}') -  \tilde{\sigma} (\mathring{\tau}' + \tilde{\tau}')\bigr) + q \bigl(\overline{\tilde{\kappa}} \overline{\tilde{\sigma}'} -  (\overline{\tilde{\rho}} + \bar{\mathring{\rho}}) (\overline{\tilde{\tau}'} + \bar{\mathring{\tau}}')\bigr)\Bigr) \varphi ,\\
\tho \edtp \varphi ={}&\edtp \tho \varphi
 -  (\mathring{\tau}' + \tilde{\tau}') \tho \varphi
 -  \overline{\tilde{\kappa}} \thop \varphi
 + \overline{\tilde{\sigma}} \edt \varphi
 + (\mathring{\rho} + \tilde{\rho}) \edtp \varphi\nonumber\\
& -  \Bigl(p \bigl(\tilde{\kappa} \tilde{\sigma}' -  (\mathring{\rho} + \tilde{\rho}) (\mathring{\tau}' + \tilde{\tau}')\bigr) + q \bigl(\overline{\tilde{\Psi}_{1}{}} + \overline{\tilde{\kappa}} (\overline{\tilde{\rho}'} + \bar{\mathring{\rho}}') -  \overline{\tilde{\sigma}} (\overline{\tilde{\tau}'} + \bar{\mathring{\tau}}')\bigr)\Bigr) \varphi ,\\
\edtp \edt \varphi ={}&\edt \edtp \varphi
 -  (- \mathring{\rho}' -  \tilde{\rho}' + \overline{\tilde{\rho}'} + \bar{\mathring{\rho}}') \tho \varphi
 -  (\mathring{\rho} + \tilde{\rho} -  \overline{\tilde{\rho}} -  \bar{\mathring{\rho}}) \thop \varphi\nonumber\\*
& -  \Bigl(- p \bigl(- \mathring{\Psi}_{2}{} -  \tilde{\Psi}_{2}{} -  (\mathring{\rho} + \tilde{\rho}) (\mathring{\rho}' + \tilde{\rho}') + \tilde{\sigma} \tilde{\sigma}'\bigr) + q \bigl(- \bar{\mathring{\Psi}}_{2}{} -  \overline{\tilde{\Psi}_{2}{}} -  (\overline{\tilde{\rho}} + \bar{\mathring{\rho}}) (\overline{\tilde{\rho}'} + \bar{\mathring{\rho}}') + \overline{\tilde{\sigma}} \overline{\tilde{\sigma}'}\bigr)\Bigr) \varphi .
\end{align}
\end{subequations}

\section{Expressions for the differential connection}\label{sec:gamma}
The background frame components of the differential connection $\widetilde\Gamma^{a}{}_{bc}$ are given by the following expressions together with their complex conjugates.
\begin{subequations}
\label{eq:GammaComponents}
\begin{align}
\widetilde\Gamma^{\mathring{l}}{}_{\mathring{l} \mathring{l}}={}&\overline{G_{1}} (- \bar{\mathring{\tau}} + \mathring{\tau}')
 + G_{1} (- \mathring{\tau} + \bar{\mathring{\tau}}')
 -  \tfrac{1}{2} \thopBG G_{0},\\
\widetilde\Gamma^{\mathring{l}}{}_{\mathring{l} \mathring{n}}={}&0,\\
\widetilde\Gamma^{\mathring{l}}{}_{\mathring{l} \vecMBackground}={}&- \tfrac{1}{2} \overline{G_{2}} \bar{\mathring{\tau}}
 + \tfrac{1}{2} \overline{G_{2}} \mathring{\tau}'
 + \tfrac{1}{4} \slashed{G}_{} (\mathring{\tau} -  \bar{\mathring{\tau}}')
 -  \tfrac{1}{2} (\thopBG{} -  \mathring{\rho}')\overline{G_{1}},\\
\widetilde\Gamma^{\mathring{l}}{}_{\mathring{n} \mathring{n}}={}&0,\\
\widetilde\Gamma^{\mathring{l}}{}_{\mathring{n} \vecMBackground}={}&0,\\
\widetilde\Gamma^{\mathring{l}}{}_{\vecMBackground \vecMBackground}={}&- \tfrac{1}{2} (\thopBG{} - 2 \mathring{\rho}')\overline{G_{2}},\\
\widetilde\Gamma^{\mathring{l}}{}_{\vecMBackground \vecMbBackground}={}&\tfrac{1}{4} (\thopBG{} -  \mathring{\rho}' -  \bar{\mathring{\rho}}')\slashed{G}_{},\\
\widetilde\Gamma^{\mathring{n}}{}_{\mathring{l} \mathring{l}}={}&\tfrac{1}{2} G^{\#}_{1} (-2 (\thoBG{} -  \bar{\mathring{\rho}})\overline{G_{1}} + (\edtBG{} - 2 \bar{\mathring{\tau}}')G_{0})
 + \tfrac{1}{2} \overline{G^{\#}_{1}} (-2 (\thoBG{} -  \mathring{\rho})G_{1} + (\edtpBG{} - 2 \mathring{\tau}')G_{0})
 + \tfrac{1}{2} \thoBG G_{0}\nonumber\\*
& + G^{\#}_{0} \bigl(\overline{G_{1}} (- \bar{\mathring{\tau}} + \mathring{\tau}') + G_{1} (- \mathring{\tau} + \bar{\mathring{\tau}}') -  \tfrac{1}{2} \thopBG G_{0}\bigr),\\
\widetilde\Gamma^{\mathring{n}}{}_{\mathring{l} \mathring{n}}={}&G_{1} \mathring{\tau}
 + \overline{G_{1}} \bar{\mathring{\tau}}
 + \tfrac{1}{4} G^{\#}_{1} \bigl(-2 \overline{G_{2}} (\bar{\mathring{\tau}} + \mathring{\tau}') + \slashed{G}_{} (\mathring{\tau} + \bar{\mathring{\tau}}') - 2 (\thopBG{} -  \mathring{\rho}')\overline{G_{1}}\bigr)\nonumber\\*
& + \tfrac{1}{4} \overline{G^{\#}_{1}} \bigl(\slashed{G}_{} (\bar{\mathring{\tau}} + \mathring{\tau}') - 2 G_{2} (\mathring{\tau} + \bar{\mathring{\tau}}') - 2 (\thopBG{} -  \bar{\mathring{\rho}}')G_{1}\bigr)
 + \tfrac{1}{2} \thopBG G_{0},\\
\widetilde\Gamma^{\mathring{n}}{}_{\mathring{l} \vecMBackground}={}&\overline{G_{1}} \bar{\mathring{\rho}}
 + \tfrac{1}{4} G^{\#}_{0} \bigl(-2 \overline{G_{2}} (\bar{\mathring{\tau}} -  \mathring{\tau}') + \slashed{G}_{} (\mathring{\tau} -  \bar{\mathring{\tau}}') - 2 (\thopBG{} -  \mathring{\rho}')\overline{G_{1}}\bigr)
  -  \tfrac{1}{2} G^{\#}_{1} (2 \overline{G_{1}} \bar{\mathring{\tau}}' + \thoBG \overline{G_{2}})\nonumber\\*
& -  \tfrac{1}{4} \overline{G^{\#}_{1}} \bigl(2 G_{0} (\mathring{\rho}' -  \bar{\mathring{\rho}}') -  (\thoBG{} -  \mathring{\rho} + \bar{\mathring{\rho}})\slashed{G}_{} + 2 (\edtBG{} + \bar{\mathring{\tau}}')G_{1} - 2 (\edtpBG{} -  \mathring{\tau}')\overline{G_{1}}\bigr)
+ \tfrac{1}{2} \edtBG G_{0},\\
\widetilde\Gamma^{\mathring{n}}{}_{\mathring{n} \mathring{n}}={}&0,\\
\widetilde\Gamma^{\mathring{n}}{}_{\mathring{n} \vecMBackground}={}&\tfrac{1}{2} \overline{G_{2}} (\bar{\mathring{\tau}}
 -   \mathring{\tau}')
 -  \tfrac{1}{4} \slashed{G}_{} (\mathring{\tau} -  \bar{\mathring{\tau}}')
 + \tfrac{1}{2} (\thopBG{} + \mathring{\rho}')\overline{G_{1}}
 + \tfrac{1}{4} \overline{G^{\#}_{1}} (\thopBG{} + \mathring{\rho}' -  \bar{\mathring{\rho}}')\slashed{G}_{}
 -  \tfrac{1}{2} G^{\#}_{1} \thopBG \overline{G_{2}},\\
\widetilde\Gamma^{\mathring{n}}{}_{\vecMBackground \vecMBackground}={}&- \tfrac{1}{2} (\thoBG{} - 2 \bar{\mathring{\rho}})\overline{G_{2}}
 -  \tfrac{1}{2} G^{\#}_{0} (\thopBG{} - 2 \mathring{\rho}')\overline{G_{2}}
 + (\edtBG{} -  \bar{\mathring{\tau}}')\overline{G_{1}}
 -  \tfrac{1}{2} G^{\#}_{1} \edtBG \overline{G_{2}}\nonumber\\*
& + \tfrac{1}{2} \overline{G^{\#}_{1}} \bigl(2 \overline{G_{1}} (- \mathring{\rho}' + \bar{\mathring{\rho}}') + \edtBG \slashed{G}_{} + \edtpBG \overline{G_{2}}\bigr),\\
\widetilde\Gamma^{\mathring{n}}{}_{\vecMBackground \vecMbBackground}={}&\tfrac{1}{2} G_{0} \mathring{\rho}'
 + \tfrac{1}{2} G_{0} \bar{\mathring{\rho}}'
 + \tfrac{1}{4} (\thoBG{} -  \mathring{\rho} -  \bar{\mathring{\rho}})\slashed{G}_{}
 + \tfrac{1}{4} G^{\#}_{0} (\thopBG{} -  \mathring{\rho}' -  \bar{\mathring{\rho}}')\slashed{G}_{}
 + \tfrac{1}{2} (\edtBG{} -  \bar{\mathring{\tau}}')G_{1}
\nonumber\\*
& + \tfrac{1}{2} (\edtpBG{} -  \mathring{\tau}')\overline{G_{1}}
 -  \tfrac{1}{2} \overline{G^{\#}_{1}} (2 G_{1} \mathring{\rho}' + \edtBG G_{2})
 -  \tfrac{1}{2} G^{\#}_{1} (2 \overline{G_{1}} \bar{\mathring{\rho}}' + \edtpBG \overline{G_{2}}),\\
\widetilde\Gamma^{\vecMBackground}{}_{\mathring{l} \mathring{l}}={}&G_{1} \mathring{\rho}
 -  G_{0} \mathring{\tau}'
 + G^{\#}_{2} ((\thoBG{} -  \bar{\mathring{\rho}})\overline{G_{1}} -  \tfrac{1}{2} (\edtBG{} - 2 \bar{\mathring{\tau}}')G_{0})
 -  \thoBG G_{1}
  + \tfrac{1}{2} \edtpBG G_{0}\nonumber\\*
& + G^{\#}_{1} \bigl(\overline{G_{1}} (\bar{\mathring{\tau}} -  \mathring{\tau}') + G_{1} (\mathring{\tau} -  \bar{\mathring{\tau}}') + \tfrac{1}{2} \thopBG G_{0}\bigr)
 + \tfrac{1}{4} \slashed{G}^{\#} (2 G_{1} \mathring{\rho} - 2 G_{0} \mathring{\tau}' - 2 \thoBG G_{1} + \edtpBG G_{0}),\\
\widetilde\Gamma^{\vecMBackground}{}_{\mathring{l} \mathring{n}}={}&
 -  \tfrac{1}{2} G_{2} \mathring{\tau}
 + \tfrac{1}{4} \slashed{G}_{} (\bar{\mathring{\tau}} + \mathring{\tau}')
 -  \tfrac{1}{2} G_{2} \bar{\mathring{\tau}}'
 + \tfrac{1}{4} G^{\#}_{2} \bigl(2 \overline{G_{2}} (\bar{\mathring{\tau}} + \mathring{\tau}') -  \slashed{G}_{} (\mathring{\tau} + \bar{\mathring{\tau}}') + 2 (\thopBG{} -  \mathring{\rho}')\overline{G_{1}}\bigr)\nonumber\\*
& + \tfrac{1}{8} \slashed{G}^{\#} \bigl(2 G_{1} \bar{\mathring{\rho}}' + \slashed{G}_{} \bar{\mathring{\tau}} + \slashed{G}_{} \mathring{\tau}' - 2 G_{2} (\mathring{\tau} + \bar{\mathring{\tau}}') - 2 \thopBG G_{1}\bigr)
 -  \tfrac{1}{2} \thopBG G_{1}
 + \tfrac{1}{2} G_{1} \bar{\mathring{\rho}}',\\
\widetilde\Gamma^{\vecMBackground}{}_{\mathring{l} \vecMBackground}={}&- \tfrac{1}{4} \slashed{G}_{} (\mathring{\rho} -  \bar{\mathring{\rho}})
 -  \tfrac{1}{2} G_{0} \mathring{\rho}'
 + \tfrac{1}{2} G_{0} \bar{\mathring{\rho}}'
 -  \tfrac{1}{2} \overline{G_{1}} \mathring{\tau}'
 -  \tfrac{1}{2} G_{1} \bar{\mathring{\tau}}'
 -  \tfrac{1}{2} \edtBG G_{1}
 + \tfrac{1}{2} \edtpBG \overline{G_{1}}\nonumber\\*
& + \tfrac{1}{4} G^{\#}_{1} \bigl(2 \overline{G_{2}} (\bar{\mathring{\tau}} -  \mathring{\tau}') + \slashed{G}_{} (- \mathring{\tau} + \bar{\mathring{\tau}}') + 2 (\thopBG{} -  \mathring{\rho}')\overline{G_{1}}\bigr)
 + \tfrac{1}{4} \thoBG \slashed{G}_{}
 + \tfrac{1}{2} G^{\#}_{2} (2 \overline{G_{1}} \bar{\mathring{\tau}}' + \thoBG \overline{G_{2}})\nonumber\\*
& -  \tfrac{1}{8} \slashed{G}^{\#} \bigl(\slashed{G}_{} \mathring{\rho} -  \slashed{G}_{} \bar{\mathring{\rho}} + 2 G_{0} (\mathring{\rho}' -  \bar{\mathring{\rho}}') + 2 \overline{G_{1}} \mathring{\tau}' + 2 G_{1} \bar{\mathring{\tau}}' -  \thoBG \slashed{G}_{} + 2 \edtBG G_{1} - 2 \edtpBG \overline{G_{1}}\bigr),\\
\widetilde\Gamma^{\vecMBackground}{}_{\mathring{l} \vecMbBackground}={}&- G_{1} \mathring{\tau}'
 + \tfrac{1}{4} G^{\#}_{1} \bigl(\slashed{G}_{} (- \bar{\mathring{\tau}} + \mathring{\tau}') + 2 G_{2} (\mathring{\tau} -  \bar{\mathring{\tau}}') + 2 (\thopBG{} -  \bar{\mathring{\rho}}')G_{1}\bigr)
  -  \tfrac{1}{4} \slashed{G}^{\#} (2 G_{1} \mathring{\tau}' + \thoBG G_{2})\nonumber\\*
& -  \tfrac{1}{4} G^{\#}_{2} \bigl(2 G_{0} (\mathring{\rho}' -  \bar{\mathring{\rho}}') + (\thoBG{} + \mathring{\rho} -  \bar{\mathring{\rho}})\slashed{G}_{} + 2 (\edtBG{} -  \bar{\mathring{\tau}}')G_{1} - 2 (\edtpBG{} + \mathring{\tau}')\overline{G_{1}}\bigr)
 -  \tfrac{1}{2} \thoBG G_{2},\\
\widetilde\Gamma^{\vecMBackground}{}_{\mathring{n} \mathring{n}}={}&0,\\
\widetilde\Gamma^{\vecMBackground}{}_{\mathring{n} \vecMBackground}={}&\tfrac{1}{4} \slashed{G}_{} (\mathring{\rho}' -  \bar{\mathring{\rho}}')
 + \tfrac{1}{4} \thopBG \slashed{G}_{}
 + \tfrac{1}{8} \slashed{G}^{\#} \bigl(\slashed{G}_{} (\mathring{\rho}' -  \bar{\mathring{\rho}}') + \thopBG \slashed{G}_{}\bigr)
 + \tfrac{1}{2} G^{\#}_{2} \thopBG \overline{G_{2}},\\
\widetilde\Gamma^{\vecMBackground}{}_{\mathring{n} \vecMbBackground}={}&- \tfrac{1}{4} G^{\#}_{2} (\thopBG{} -  \mathring{\rho}' + \bar{\mathring{\rho}}')\slashed{G}_{}
 -  \tfrac{1}{2} \thopBG G_{2}
 -  \tfrac{1}{4} \slashed{G}^{\#} \thopBG G_{2},\\
\widetilde\Gamma^{\vecMBackground}{}_{\vecMBackground \vecMBackground}={}&\overline{G_{1}} (\bar{\mathring{\rho}}' - \mathring{\rho}')
 + \tfrac{1}{2} G^{\#}_{1} (\thopBG{} - 2 \mathring{\rho}')\overline{G_{2}}
 + \tfrac{1}{2} \edtBG \slashed{G}_{}
 + \tfrac{1}{2} G^{\#}_{2} \edtBG \overline{G_{2}}
 + \tfrac{1}{2} \edtpBG \overline{G_{2}}\nonumber\\*
& + \tfrac{1}{4} \slashed{G}^{\#} \bigl(\edtBG \slashed{G}_{} + \edtpBG \overline{G_{2}} -2 \overline{G_{1}} (\mathring{\rho}' -  \bar{\mathring{\rho}}')\bigr),\\
\widetilde\Gamma^{\vecMBackground}{}_{\vecMBackground \vecMbBackground}={}&- G_{1} \mathring{\rho}'
 -  \tfrac{1}{4} G^{\#}_{1} (\thopBG{} -  \mathring{\rho}' -  \bar{\mathring{\rho}}')\slashed{G}_{}
 -  \tfrac{1}{2} \edtBG G_{2}
 -  \tfrac{1}{4} \slashed{G}^{\#} (2 G_{1} \mathring{\rho}' + \edtBG G_{2})
 + \tfrac{1}{2} G^{\#}_{2} (2 \overline{G_{1}} \bar{\mathring{\rho}}' + \edtpBG \overline{G_{2}}),\\
\widetilde\Gamma^{\vecMBackground}{}_{\vecMbBackground \vecMbBackground}={}&\tfrac{1}{2} G^{\#}_{1} (\thopBG{} - 2 \bar{\mathring{\rho}}')G_{2}
 -  \tfrac{1}{2} G^{\#}_{2} \bigl(2 G_{1} (\mathring{\rho}' -  \bar{\mathring{\rho}}') + \edtBG G_{2} + \edtpBG \slashed{G}_{}\bigr)
 -  \tfrac{1}{2} \edtpBG G_{2}
 -  \tfrac{1}{4} \slashed{G}^{\#} \edtpBG G_{2}.
\end{align}
\end{subequations}

\section{Constraint equations arising from the structure and Ricci relations}\label{sec:ExtraStructureEqs}
The remaining structure equations take the form
\begin{subequations}
\label{eq:StructureSpincoeff1}
\begin{align}
i \edt \nu ={}&\tilde{\beta}
 + \overline{\tilde{\beta}'}
 -  \overline{\tilde{G}_{1}{}} (\mathring{\rho}'
 -  \bar{\mathring{\rho}}')
 -  \eta \bar{\mathring{\rho}}'
 + \frac{\varsigma}{2 \varsigma^{\#}{}} (\bar{\eta} \overline{\tilde{G}_{2}{}}
 + 2 \eta \varsigma^{\#}{}^2) (\mathring{\rho}'
 + \bar{\mathring{\rho}}')
 + \frac{\overline{\tilde{G}_{2}{}} \varsigma \edt \tilde{G}_{2}{}}{4 \varsigma^{\#}{}^3}
 + \frac{\edt \varsigma^{\#}{}}{\varsigma^{\#}{}}\nonumber\\
& + \frac{\varsigma \edtp \overline{\tilde{G}_{2}{}}}{2 \varsigma^{\#}{}},
\label{eq:ethnuTobeta}\\
i \tho \nu ={}&\tilde{\epsilon}
 -  \overline{\tilde{\epsilon}}
 -  \frac{\tilde{\slashed{G}} \varsigma}{4 \varsigma^{\#}{}} (\mathring{\rho}
 + \tilde{\rho}
 -  \overline{\tilde{\rho}}
 -  \bar{\mathring{\rho}})
 -  \tfrac{1}{2} \tilde{G}_{0}{} (\mathring{\rho}'
 -  \bar{\mathring{\rho}}')
 -  \frac{\varsigma}{2 \varsigma^{\#}{}} (\eta \tilde{G}_{1}{}
 -  \bar{\eta} \overline{\tilde{G}_{1}{}}) (\mathring{\rho}'
 + \bar{\mathring{\rho}}')\nonumber\\
&
 + \frac{\varsigma \mathring{\tau}'}{4 e^{i \nu} \varsigma^{\#}{}^2} (\bar{\eta} \overline{\tilde{G}_{2}{}} + 2 \eta \varsigma^{\#}{}^2)
 -  \frac{e^{i \nu} \varsigma \bar{\mathring{\tau}}'}{4 \varsigma^{\#}{}^2} (\eta \tilde{G}_{2}{} + 2 \bar{\eta} \varsigma^{\#}{}^2)
 + \frac{\eta \varsigma}{2 \varsigma^{\#}{}} (\overline{\tilde{\beta}} -  \tilde{\beta}' + \mathring{\tau}' + \tilde{\tau}')
 \nonumber\\
& -  \frac{\bar{\eta} \varsigma}{2 \varsigma^{\#}{}} (\tilde{\beta} -  \overline{\tilde{\beta}'} + \overline{\tilde{\tau}'} + \bar{\mathring{\tau}}')
 -  \frac{\varsigma (\edt{} -  \overline{\tilde{\tau}'} -  \bar{\mathring{\tau}}')\tilde{G}_{1}{}}{2 \varsigma^{\#}{}}
 + \frac{\varsigma (\edtp{} -  \mathring{\tau}' -  \tilde{\tau}')\overline{\tilde{G}_{1}{}}}{2 \varsigma^{\#}{}}
 + \frac{\overline{\tilde{G}_{2}{}} \varsigma \tho \tilde{G}_{2}{}}{8 \varsigma^{\#}{}^3}\nonumber\\
& -  \frac{\tilde{G}_{2}{} \varsigma \tho \overline{\tilde{G}_{2}{}}}{8 \varsigma^{\#}{}^3}
 -  \frac{\tilde{G}_{2}{} \varsigma \tilde{\sigma}}{2 \varsigma^{\#}{}}
  + \frac{\overline{\tilde{G}_{2}{}} \varsigma \overline{\tilde{\sigma}}}{2 \varsigma^{\#}{}} ,
\label{eq:thoetaToepsilon}\\
\edt \eta ={}&- \tfrac{1}{2} \overline{\tilde{G}_{2}{}} \mathring{\rho}
 + \overline{\tilde{G}_{2}{}} \overline{\tilde{\rho}}
 + \tfrac{1}{2} \overline{\tilde{G}_{2}{}} \bar{\mathring{\rho}}
 + \frac{\varsigma}{\varsigma^{\#}{}} \bigl((\bar{\eta} -  \tilde{G}_{1}{}) (\eta -  \overline{\tilde{G}_{1}{}}) \overline{\tilde{G}_{2}{}}
 + \tfrac{1}{2} (\eta -  \overline{\tilde{G}_{1}{}})^2 (2 + \tilde{\slashed{G}})\bigr) (\mathring{\rho}'
 -  \bar{\mathring{\rho}}')\nonumber\\
& + \frac{\overline{\tilde{G}_{2}{}} \varsigma^2}{4 \varsigma^{\#}{}^2} \bigl((\eta^2 -  \overline{\tilde{G}_{1}{}}^2) \tilde{G}_{2}{}
 + (\bar{\eta}^2 -  \tilde{G}_{1}{}^2) \overline{\tilde{G}_{2}{}}
 + (\eta \bar{\eta} -  \tilde{G}_{1}{} \overline{\tilde{G}_{1}{}}) (2 + \tilde{\slashed{G}})
 -  \frac{\tilde{G}_{0}{} \varsigma^{\#}{}^2}{\varsigma^2} \bigr) (\mathring{\rho}'
 + \bar{\mathring{\rho}}')\nonumber\\
& -  \tfrac{1}{2} (2
 + \tilde{\slashed{G}}) \tilde{\sigma}
 -  \frac{\eta \overline{\tilde{G}_{2}{}}}{2 e^{i \nu} \varsigma^{\#}{}} (\bar{\mathring{\tau}}
 -  \mathring{\tau}')
 + e^{i \nu} \eta \varsigma^{\#}{} (\mathring{\tau}
 -  \bar{\mathring{\tau}}')
 + (\eta
 -  \overline{\tilde{G}_{1}{}}) (\overline{\tilde{\tau}'}
 + \bar{\mathring{\tau}}')
 -  \tfrac{1}{2} \tho \overline{\tilde{G}_{2}{}}\nonumber\\
& + \edt \overline{\tilde{G}_{1}{}}
 -  \frac{\varsigma^2}{4 \varsigma^{\#}{}^2} (\bar{\eta}
 -  \tilde{G}_{1}{}) \bigl((2 + \tilde{\slashed{G}}) \edt \overline{\tilde{G}_{2}{}}
 - 2 \overline{\tilde{G}_{2}{}} (\edt \tilde{\slashed{G}} + \edtp \overline{\tilde{G}_{2}{}})\bigr)
\nonumber\\
& -  \frac{\varsigma^2}{4 \varsigma^{\#}{}^2} (\eta
  -  \overline{\tilde{G}_{1}{}}) \bigl(2 \tilde{G}_{2}{} \edt \overline{\tilde{G}_{2}{}}
 -  (2 + \tilde{\slashed{G}}) (\edt \tilde{\slashed{G}} + \edtp \overline{\tilde{G}_{2}{}})\bigr),\\
\edtp \eta -  \edt \bar{\eta}={}&- \tilde{\rho} + \overline{\tilde{\rho}}
 -  \frac{1}{2 e^{i \nu} \varsigma^{\#}{}} (\bar{\eta} \overline{\tilde{G}_{2}{}}
 + 2 \eta \varsigma^{\#}{}^2) (\bar{\mathring{\tau}}
 + \mathring{\tau}')
  + \frac{e^{i \nu}}{2 \varsigma^{\#}{}} (\eta \tilde{G}_{2}{}
 + 2 \bar{\eta} \varsigma^{\#}{}^2) (\mathring{\tau} + \bar{\mathring{\tau}}') 
\nonumber\\
& + \frac{\varsigma}{2 \varsigma^{\#}{}} \Bigl(
  \frac{\tilde{G}_{0}{} \varsigma^{\#}{}^2}{\varsigma^2} 
+(\overline{\tilde{G}_{1}{}}^2 - \eta^2) \tilde{G}_{2}{}
 + (\tilde{G}_{1}{}^2 - \bar{\eta}^2 ) \overline{\tilde{G}_{2}{}}
  + (\tilde{G}_{1}{} \overline{\tilde{G}_{1}{}} - \eta \bar{\eta}) (2 + \tilde{\slashed{G}})
 \Bigr) (\mathring{\rho}' -  \bar{\mathring{\rho}}')
 \nonumber\\
&
 +  \Bigl(\frac{\varsigma^{\#}{}}{\varsigma} -  1\Bigr) (\mathring{\rho} -  \bar{\mathring{\rho}})
 + \eta (\mathring{\tau}'  + \tilde{\tau}')
 -  \bar{\eta} (\overline{\tilde{\tau}'} + \bar{\mathring{\tau}}') ,\\
\edtp \eta + \edt \bar{\eta}={}&
 - \tfrac{1}{2} (2 + \tilde{\slashed{G}}) (\tilde{\rho}+ \overline{\tilde{\rho}})
 - \frac{\varsigma}{\varsigma^{\#}{}} \bigl((\eta -  \overline{\tilde{G}_{1}{}})^2 \tilde{G}_{2}{}
 - (\bar{\eta} -  \tilde{G}_{1}{})^2 \overline{\tilde{G}_{2}{}}\bigr) (\mathring{\rho}' -  \bar{\mathring{\rho}}')
 +\frac{\varsigma^2}{4 \varsigma^{\#}{}^2}(2 + \tilde{\slashed{G}})
\nonumber\\
& \qquad\times \Bigl((\overline{\tilde{G}_{1}{}}^2 - \eta^2) \tilde{G}_{2}{}
 + (\tilde{G}_{1}{}^2 - \bar{\eta}^2) \overline{\tilde{G}_{2}{}}
 + (\tilde{G}_{1}{} \overline{\tilde{G}_{1}{}}- \eta \bar{\eta}) (2 + \tilde{\slashed{G}})
 + \frac{\tilde{G}_{0}{} \varsigma^{\#}{}^2}{\varsigma^2} \Bigr) (\mathring{\rho}' + \bar{\mathring{\rho}}')\nonumber\\
&
 + \tilde{G}_{2}{} \tilde{\sigma}
 + \overline{\tilde{G}_{2}{}} \overline{\tilde{\sigma}}
 + (\eta -  \overline{\tilde{G}_{1}{}}) (\mathring{\tau}' + \tilde{\tau}')
 + (\bar{\eta} -  \tilde{G}_{1}{}) (\overline{\tilde{\tau}'} + \bar{\mathring{\tau}}')
 + \tfrac{1}{2} \tho \tilde{\slashed{G}}
 + \edt \tilde{G}_{1}{}
 + \edtp \overline{\tilde{G}_{1}{}}\nonumber\\
& + \frac{\varsigma^2 \edt \tilde{G}_{2}{}}{2 \varsigma^{\#}{}^2} \bigl(2 (\tilde{G}_{1}{} - \bar{\eta}) \overline{\tilde{G}_{2}{}}
 + (\overline{\tilde{G}_{1}{}}- \eta) (2 + \tilde{\slashed{G}})\bigr)
 + \frac{1}{2 e^{i \nu} \varsigma^{\#}{}} (2 \eta \varsigma^{\#}{}^2- \bar{\eta} \overline{\tilde{G}_{2}{}}) (\bar{\mathring{\tau}} -  \mathring{\tau}')
\nonumber\\
& + \frac{\varsigma^2 \edtp \overline{\tilde{G}_{2}{}}}{2 \varsigma^{\#}{}^2} \bigl(2 (\overline{\tilde{G}_{1}{}}- \eta) \tilde{G}_{2}{}
  + (\tilde{G}_{1}{} - \bar{\eta}) (2 + \tilde{\slashed{G}})\bigr)
  + \frac{e^{i \nu}}{2 \varsigma^{\#}{}} (2 \bar{\eta} \varsigma^{\#}{}^2 - \eta \tilde{G}_{2}{}) (\mathring{\tau} -  \bar{\mathring{\tau}}'),\\
\tho \eta ={}&- \eta \tilde{\epsilon}
 -  \eta \overline{\tilde{\epsilon}}
 -  \tilde{\kappa}
 + \overline{\tilde{G}_{1}{}} \overline{\tilde{\rho}}
 -  \tfrac{1}{2} \overline{\tilde{G}_{1}{}} (\mathring{\rho}
 -  \bar{\mathring{\rho}})
 + \tfrac{1}{2} \eta (\mathring{\rho}
 + \bar{\mathring{\rho}})
 + \tfrac{1}{4} \tilde{G}_{0}{} (\eta
 -  \overline{\tilde{G}_{1}{}}) (\mathring{\rho}' + \bar{\mathring{\rho}}')
 + \tilde{G}_{1}{} \tilde{\sigma}
\nonumber\\
&  + \bigl((\eta -  \overline{\tilde{G}_{1}{}})^2 \tilde{G}_{2}{}
 + (\bar{\eta} -  \tilde{G}_{1}{})^2 \overline{\tilde{G}_{2}{}}
 + (\bar{\eta} -  \tilde{G}_{1}{}) (\eta -  \overline{\tilde{G}_{1}{}}) (2 + \tilde{\slashed{G}})\bigr) (\frac{e^{i \nu} \varsigma^2 \mathring{\tau}}{2 \varsigma^{\#}{}}
 -  \frac{\overline{\tilde{G}_{2}{}} \varsigma^2 \bar{\mathring{\tau}}}{4 e^{i \nu} \varsigma^{\#}{}^3})\nonumber\\
& -  \frac{\tilde{G}_{0}{} \overline{\tilde{G}_{2}{}}}{4 e^{i \nu} \varsigma^{\#}{}} (\bar{\mathring{\tau}}
 + \mathring{\tau}')
 + \tfrac{1}{2} e^{i \nu} \tilde{G}_{0}{} \varsigma^{\#}{} (\mathring{\tau} + \bar{\mathring{\tau}}')
  + \tfrac{1}{2} (\edt
 - 2 \overline{\tilde{\tau}'}
 - 2 \bar{\mathring{\tau}}')\tilde{G}_{0}{}
\nonumber\\
& + \bigl(2 (\bar{\eta} -  \tilde{G}_{1}{}) \overline{\tilde{G}_{2}{}}
 + (\eta -  \overline{\tilde{G}_{1}{}}) (2 + \tilde{\slashed{G}})\bigr) \Bigl(\frac{\varsigma}{4 \varsigma^{\#}{}} (\mathring{\rho} -  \bar{\mathring{\rho}})
 -  \frac{3 \tilde{G}_{0}{} \varsigma}{8 \varsigma^{\#}{}} (\mathring{\rho}' -  \bar{\mathring{\rho}}')\nonumber\\
&\qquad + \frac{\varsigma^2}{8 \varsigma^{\#}{}^2} \bigl(
2 \overline{\tilde{G}_{2}{}} \overline{\tilde{\sigma}} -2 \tilde{G}_{2}{} \tilde{\sigma}  + (\tho{} -  \mathring{\rho} -  \tilde{\rho} + \overline{\tilde{\rho}} + \bar{\mathring{\rho}})\tilde{\slashed{G}} - 2 (\edt{} -  \overline{\tilde{\tau}'} -  \bar{\mathring{\tau}}')\tilde{G}_{1}{} 
\nonumber\\
&\qquad+ 2 (\edtp{} -  \mathring{\tau}' -  \tilde{\tau}')\overline{\tilde{G}_{1}{}}\bigr)\Bigr)
-  \frac{\varsigma^2 \tho \overline{\tilde{G}_{2}{}}}{4 \varsigma^{\#}{}^2} \bigl(2 (\eta -  \overline{\tilde{G}_{1}{}}) \tilde{G}_{2}{}
+ (\bar{\eta} -  \tilde{G}_{1}{}) (2 + \tilde{\slashed{G}})\bigr)
\nonumber\\
& + \bigl((\eta^2 -  \overline{\tilde{G}_{1}{}}^2) \tilde{G}_{2}{}
 + (\bar{\eta}^2 -  \tilde{G}_{1}{}^2) \overline{\tilde{G}_{2}{}}
 + (\eta \bar{\eta} -  \tilde{G}_{1}{} \overline{\tilde{G}_{1}{}}) (2 + \tilde{\slashed{G}})\bigr) 
 \nonumber\\
&\qquad\times\Bigl(\frac{\varsigma^3}{8 \varsigma^{\#}{}^3} \bigl(2 (\bar{\eta} -  \tilde{G}_{1}{}) \overline{\tilde{G}_{2}{}} + (\eta -  \overline{\tilde{G}_{1}{}}) (2 + \tilde{\slashed{G}})\bigr) (\mathring{\rho}' -  \bar{\mathring{\rho}}')
 + \frac{\varsigma^2}{4 \varsigma^{\#}{}^2} (\overline{\tilde{G}_{1}{}} - \eta ) (\mathring{\rho}' + \bar{\mathring{\rho}}')
\nonumber\\
&
\qquad\qquad + \frac{\overline{\tilde{G}_{2}{}} \varsigma^2 \mathring{\tau}'}{4 e^{i \nu} \varsigma^{\#}{}^3}
 -  \frac{e^{i \nu} \varsigma^2 \bar{\mathring{\tau}}'}{2 \varsigma^{\#}{}} \Bigr).
\end{align}
\end{subequations}
The remaining Ricci relations take the form
\begin{subequations}
\label{eq:ExtraRicci}
\begin{align}
\edtp \tilde{\rho}' -  \edt \tilde{\sigma}'={}&- \tilde{\Psi}_{3}{}
 + \frac{\varsigma^{\#}{} \mathring{\rho}'}{\varsigma} (\tilde{G}^{\#}_1{} -  \eta \tilde{G}^{\#}_2{}) (\bar{\mathring{\rho}}' -  \mathring{\rho}')
 -  \bar{\eta} \mathring{\rho}' \bigl(2 \varsigma^{\#}{} \varsigma (\mathring{\rho}' -  \bar{\mathring{\rho}}') + \bar{\mathring{\rho}}'\bigr)
 -  \frac{e^{i \nu} \tilde{G}^{\#}_2{} \varsigma^{\#}{} \mathring{\tau}}{2 \varsigma^2} (2 \mathring{\rho}' -  \bar{\mathring{\rho}}')
 \nonumber\\*
&+ \bigl(2 \mathring{\rho}' + \tilde{\rho}' -  \overline{\tilde{\rho}'} -  \bar{\mathring{\rho}}' + \frac{\varsigma^{\#}{}}{e^{i \nu}} (-2 \mathring{\rho}' + \bar{\mathring{\rho}}')\bigr) \mathring{\tau}'
 + (2 \mathring{\rho}' + \tilde{\rho}' -  \overline{\tilde{\rho}'} -  \bar{\mathring{\rho}}') \tilde{\tau}',\\
\tho \tilde{\sigma}' -  \edtp \tilde{\tau}'={}&(\mathring{\rho}' + \tilde{\rho}') \overline{\tilde{\sigma}}
 + (\mathring{\rho} + \tilde{\rho}) \tilde{\sigma}'
 - \Bigl(1 - \frac{\varsigma^{\#}{}}{e^{i \nu}}\Bigr) \mathring{\tau}'^2
 + \bar{\eta} \mathring{\rho}' (\mathring{\tau}' - \bar{\mathring{\tau}} )
 -  \tilde{\tau}'^2
 -  \mathring{\tau}' (\overline{\tilde{\beta}} + \tilde{\beta}' + 2 \tilde{\tau}')\nonumber\\*
& + \frac{e^{i \nu} \tilde{G}^{\#}_2{} \varsigma^{\#}{} \mathring{\edt}' \mathring{\tau}}{2 \varsigma^2} ,\\
\tho \tilde{\rho}' -  \edt \tilde{\tau}'={}&- \tilde{\Psi}_{2}{}
 -  (\tilde{\epsilon} + \overline{\tilde{\epsilon}} -  \overline{\tilde{\rho}}) \mathring{\rho}'
 + \tfrac{1}{2} \bigl(\tilde{G}^{\#}_0{} + \eta^2 \tilde{G}^{\#}_2{} + \bar{\eta}^2 \overline{\tilde{G}}^{\#}_2{} -  \eta \bar{\eta} (2 + \tilde{\slashed{G}}^{\#}{})\bigr) \mathring{\rho}'^2
 + (\overline{\tilde{\rho}} + \bar{\mathring{\rho}}) \tilde{\rho}'
 + \tilde{\sigma} \tilde{\sigma}'\nonumber\\
& + \frac{\overline{\tilde{G}}^{\#}_2{} \varsigma^{\#}{} \mathring{\tau}'^2}{2 e^{i \nu} \varsigma^2}
 -  \overline{\tilde{G}}^{\#}_1{} \varsigma^{\#}{} \mathring{\rho}' \bigl(\frac{e^{i \nu} \tilde{G}^{\#}_2{} \mathring{\tau}}{\varsigma^2} + \frac{1}{e^{i \nu}} (\bar{\mathring{\tau}} + \mathring{\tau}')\bigr)
 -  \tfrac{1}{2} \tilde{G}^{\#}_1{} \varsigma^{\#}{} \mathring{\rho}' \bigl(4 e^{i \nu} \mathring{\tau} + \frac{\overline{\tilde{G}}^{\#}_2{}}{e^{i \nu} \varsigma^2} (\bar{\mathring{\tau}} + \mathring{\tau}')\bigr)\nonumber\\
& + \tfrac{1}{2} \bar{\eta} \mathring{\rho}' \bigl(-2 \overline{\tilde{G}}^{\#}_1{} \mathring{\rho}' - 4 e^{i \nu} \varsigma \mathring{\tau} + \frac{\overline{\tilde{G}}^{\#}_2{}}{e^{i \nu} \varsigma} (\bar{\mathring{\tau}} + \mathring{\tau}')\bigr)
 -  (1 -  e^{i \nu} \varsigma^{\#}{}) \mathring{\edt}' \mathring{\tau}
  -  \tilde{\tau}' \overline{\tilde{\tau}'}
 -  \tilde{\tau}' \bar{\mathring{\tau}}'\nonumber\\
& + \eta \mathring{\rho}' \bigl(- \tilde{G}^{\#}_1{} \mathring{\rho}' + \frac{e^{i \nu} \tilde{G}^{\#}_2{} \mathring{\tau}}{\varsigma} -  \bar{\mathring{\tau}} + \mathring{\tau}' -  \frac{\varsigma}{e^{i \nu}} (\bar{\mathring{\tau}} + \mathring{\tau}')\bigr)
 + \mathring{\tau}' (\tilde{\beta} + \overline{\tilde{\beta}'} -  \overline{\tilde{\tau}'}),\\
\tho \tilde{\rho} -  \edtp \tilde{\kappa}={}&\tilde{\epsilon} \mathring{\rho}
 + \overline{\tilde{\epsilon}} \mathring{\rho}
 + 2 \mathring{\rho} \tilde{\rho}
 + \tilde{\rho}^2
 + \tilde{\sigma} \overline{\tilde{\sigma}}
 -  \overline{\tilde{\kappa}} \mathring{\tau}
 -  \frac{e^{i \nu} \mathring{\tau}}{2 \varsigma^2} (\overline{\tilde{G}}^{\#}_1{} \tilde{G}^{\#}_2{} \varsigma^{\#}{} -  \eta \tilde{G}^{\#}_2{} \varsigma + 2 \tilde{G}^{\#}_1{} \varsigma^{\#}{} \varsigma^2 + 2 \bar{\eta} \varsigma^3) (\mathring{\rho} -  \bar{\mathring{\rho}})\nonumber\\
& -  \frac{\mathring{\rho} \mathring{\tau}'}{e^{i \nu} \varsigma^2} (\tilde{G}^{\#}_1{} \overline{\tilde{G}}^{\#}_2{} \varsigma^{\#}{} -  \bar{\eta} \overline{\tilde{G}}^{\#}_2{} \varsigma + 2 \overline{\tilde{G}}^{\#}_1{} \varsigma^{\#}{} \varsigma^2 + 2 \eta \varsigma^3)
 -  \tilde{\kappa} (\mathring{\tau}' + \tilde{\tau}')\nonumber\\
& + \tfrac{1}{2} (-2 \eta \bar{\eta} + \tilde{G}^{\#}_0{} - 2 \eta \tilde{G}^{\#}_1{} - 2 \bar{\eta} \overline{\tilde{G}}^{\#}_1{} + \eta^2 \tilde{G}^{\#}_2{} + \bar{\eta}^2 \overline{\tilde{G}}^{\#}_2{} -  \eta \bar{\eta} \tilde{\slashed{G}}^{\#}{}) \thoBG \mathring{\rho}',\\
\tho \tilde{\sigma} -  \edt \tilde{\kappa}={}&\tilde{\Psi}_{0}{}
 + (\mathring{\rho} + \tilde{\rho} + \overline{\tilde{\rho}} + \bar{\mathring{\rho}}) \tilde{\sigma}
 -  \tilde{\kappa} (\mathring{\tau} + \overline{\tilde{\tau}'} + \bar{\mathring{\tau}}'),\\
\tho \tilde{\beta} -  \edt \tilde{\epsilon}={}&
  \tilde{\Psi}_{1}{}
 - \eta \mathring{\Psi}_{2}{}
 + \tilde{\beta} (\overline{\tilde{\rho}} + \bar{\mathring{\rho}})
 + \tilde{\kappa} \mathring{\rho}'
 -  \frac{\varsigma^{\#}{}}{\varsigma} (\eta + \overline{\tilde{G}}^{\#}_1{} -  \bar{\eta} \overline{\tilde{G}}^{\#}_2{} + \tfrac{1}{2} \eta \tilde{\slashed{G}}^{\#}{}) (\mathring{\Psi}_{2}{} + \mathring{\rho} \mathring{\rho}')
 + \tilde{\kappa} \tilde{\rho}'\nonumber\\
& + \tfrac{1}{2} e^{i \nu} \tilde{G}^{\#}_0{} \varsigma^{\#}{} \mathring{\rho}' \mathring{\tau}
 + \frac{\overline{\tilde{G}}^{\#}_2{} \varsigma^{\#}{} \mathring{\rho} \mathring{\tau}'}{2 e^{i \nu} \varsigma^2}
 + \eta \mathring{\tau} \mathring{\tau}'
 -  \tilde{\sigma} (\tilde{\beta}' + \mathring{\tau}' + \tilde{\tau}')
 -  \tilde{\epsilon} (\overline{\tilde{\tau}'} + \bar{\mathring{\tau}}')\nonumber\\
& + \frac{e^{i \nu} \mathring{\rho}' \mathring{\tau}}{2 \varsigma^2} \bigl(\eta \overline{\tilde{G}}^{\#}_1{} \tilde{G}^{\#}_2{} \varsigma^{\#}{} + \bar{\eta} (\bar{\eta} \overline{\tilde{G}}^{\#}_2{} - 2 \overline{\tilde{G}}^{\#}_1{}) \varsigma^{\#}{} \varsigma^2 - (1 - \varsigma^{\#}{} \varsigma) (\eta^2 \tilde{G}^{\#}_2{} \varsigma - 4 \eta \bar{\eta} \varsigma^3)\bigr),\\
\tho \tilde{\beta}' + \edtp \tilde{\epsilon}={}&\bar{\eta} \mathring{\Psi}_{2}{}
 + \tilde{\beta}' (\mathring{\rho} + \tilde{\rho})
 -  \frac{\varsigma^{\#}{}}{\varsigma} (\bar{\eta} + \tilde{G}^{\#}_1{} -  \eta \tilde{G}^{\#}_2{} + \tfrac{1}{2} \bar{\eta} \tilde{\slashed{G}}^{\#}{}) (\mathring{\Psi}_{2}{} + \mathring{\rho} \mathring{\rho}')
 -  \tilde{\beta} \overline{\tilde{\sigma}}
 -  \tilde{\kappa} \tilde{\sigma}'\nonumber\\
& -  \frac{e^{i \nu} \tilde{G}^{\#}_0{} \tilde{G}^{\#}_2{} \varsigma^{\#}{} \mathring{\rho}' \mathring{\tau}}{4 \varsigma^2}
 +  \frac{e^{i \nu} \varsigma^{\#}{} \mathring{\rho}' \mathring{\tau}}{4 \varsigma^2} (\eta \tilde{G}^{\#}_2{} - 2 \bar{\eta} \varsigma^2) (2 \tilde{G}^{\#}_1{} - \eta \tilde{G}^{\#}_2{} + 2 \bar{\eta} \varsigma^2)
 + \Bigl(1 -  \frac{\varsigma^{\#}{}}{e^{i \nu}}\Bigr) \mathring{\rho} \mathring{\tau}'\nonumber\\
& + \tilde{\rho} \mathring{\tau}'
 -  \bar{\eta} \mathring{\tau} \mathring{\tau}'
 + \mathring{\rho} \tilde{\tau}'
 + \tilde{\rho} \tilde{\tau}'
 + \tilde{\epsilon} (\mathring{\tau}' + \tilde{\tau}'),\\
\edt \tilde{\beta}' + \edtp \tilde{\beta}={}&\tilde{\Psi}_{2}{}
 + \tilde{\rho} \mathring{\rho}'
 + \Bigl(1 -  \frac{\varsigma^{\#}{}}{\varsigma}\Bigr) (\mathring{\Psi}_{2}{} + \mathring{\rho} \mathring{\rho}')
 + (\mathring{\rho} + \tilde{\rho}) \tilde{\rho}'
 + \frac{\varsigma^{\#}{} \tilde{\epsilon}}{\varsigma} (\mathring{\rho}' -  \bar{\mathring{\rho}}')
 -  \tilde{\sigma} \tilde{\sigma}'\nonumber\\
& + \frac{e^{i \nu} \varsigma^{\#}{} \mathring{\rho}' \mathring{\tau}}{2 \varsigma^2} (\eta \tilde{G}^{\#}_2{} - 2 \bar{\eta} \varsigma^2),\\
\edt \tilde{\rho} -  \edtp \tilde{\sigma}={}&- \tilde{\Psi}_{1}{}
 + \tilde{\beta} \mathring{\rho}
 -  \overline{\tilde{\beta}'} \mathring{\rho}
 -  \tilde{\kappa} (\mathring{\rho}' + \tilde{\rho}' -  \overline{\tilde{\rho}'} -  \bar{\mathring{\rho}}')
 + (\tilde{\rho} -  \overline{\tilde{\rho}}) \mathring{\tau}
 +  (1 - e^{i \nu} \varsigma^{\#}{}) (\mathring{\rho} -  \bar{\mathring{\rho}}) \mathring{\tau}\nonumber\\
& -  \frac{\overline{\tilde{G}}^{\#}_2{} \varsigma^{\#}{} \mathring{\rho} \mathring{\tau}'}{e^{i \nu} \varsigma^2}
 -  \eta \thoBG \mathring{\rho}',
\end{align}
\end{subequations}
where the background derivatives of background spin coefficients are given by \eqref{eq:edtpBGtauBG} and
\begin{align}
\thoBG \mathring{\rho}'={}&- \tfrac{1}{2} \mathring{\Psi}_{2}{}
 -  \frac{\bar{\mathring{\Psi}}_{2}{} \bar{\kappa}_{1'}{}}{2 \kappa_{1}{}}
 + \mathring{\rho} \mathring{\rho}'
 -  \mathring{\tau} \bar{\mathring{\tau}}
 + \mathring{\tau} \mathring{\tau}'.
\end{align}

\section{Linearized constraint equations}
The linearization of the constraint equations \eqref{eq:AlgebraicSpincoeff1}, \eqref{eq:StructureSpincoeff1} are equivalent to
\begin{subequations}
\label{eq:extraLinearizedStructure}
\begin{align}
\dot{\tilde{\kappa}}'={}&0,\qquad
\dot{\tilde{\epsilon}}'={}0,\quad
\dot{\tilde{\tau}}={}0,\\
\dot{\tilde{\rho}}' -  \overline{\dot{\tilde{\rho}}'}={}&\tfrac{1}{2} \LinGTr (- \mathring{\rho}'
 + \bar{\mathring{\rho}}'),\\
\dot{\tilde{\beta}} -  \overline{\dot{\tilde{\beta}}'}={}&- \dot{\eta} \mathring{\rho}'
 -  \LinGOneDg \mathring{\rho}'
 + \dot{\eta} \bar{\mathring{\rho}}'
 + \LinGOneDg \bar{\mathring{\rho}}'
 + \tfrac{1}{2} \LinGTwoDg \mathring{\tau}'
 -  \overline{\dot{\tilde{\tau}}'}
 + i \dot{\nu} \bar{\mathring{\tau}}'
 -  \tfrac{1}{4} \LinGTr \bar{\mathring{\tau}}',\\
\dot{\tilde{\rho}} -  \overline{\dot{\tilde{\rho}}}={}&- \tfrac{1}{2} \LinGTr (\mathring{\rho}
 -  \bar{\mathring{\rho}})
 -  \tfrac{1}{2} \LinGZero (\mathring{\rho}'
 -  \bar{\mathring{\rho}}')
 + \overline{\dot{\eta}} \mathring{\tau}
 -  \dot{\eta} \bar{\mathring{\tau}}
 + \edtBG \overline{\dot{\eta}}
 -  \edtpBG \dot{\eta},\\
\dot{\tilde{\kappa}}={}&\dot{\eta} \mathring{\rho}
 + \LinGOneDg \mathring{\rho}
 -  \LinGOneDg \bar{\mathring{\rho}}
 -  \tfrac{1}{2} \LinGZero (\mathring{\tau}
 -  \bar{\mathring{\tau}}')
 -  \thoBG \dot{\eta}
 -  \tfrac{1}{2} \edtBG \LinGZero,\\
\dot{\tilde{\epsilon}} -  \overline{\dot{\tilde{\epsilon}}}={}&- \tfrac{1}{4} \LinGTr (\mathring{\rho}
 -  \bar{\mathring{\rho}})
 -  \tfrac{1}{2} \LinGZero (\mathring{\rho}'
 -  \bar{\mathring{\rho}}')
 -  \dot{\eta} \mathring{\tau}'
 -  \tfrac{1}{2} \LinGOneDg \mathring{\tau}'
 + \overline{\dot{\eta}} \bar{\mathring{\tau}}'
 + \tfrac{1}{2} \LinGOne \bar{\mathring{\tau}}'
 + i \thoBG \dot{\nu}
 -  \tfrac{1}{2} \edtBG \LinGOne\nonumber\\
& + \tfrac{1}{2} \edtpBG \LinGOneDg,\\
\dot{\tilde{\beta}}'={}&- \tfrac{1}{2} \overline{\dot{\eta}} \mathring{\rho}'
 + \tfrac{1}{2} i \dot{\nu} \mathring{\tau}'
 + \tfrac{1}{8} \LinGTr \mathring{\tau}'
 + \tfrac{1}{2} \dot{\tilde{\tau}}'
 -  \tfrac{1}{4} \LinGTwo \bar{\mathring{\tau}}'
 + \tfrac{1}{4} \edtBG \LinGTwo
 -  \tfrac{1}{2} i \edtpBG \dot{\nu}
 + \tfrac{1}{8} \edtpBG \LinGTr,\\
\dot{\tilde{\rho}}={}&- \tfrac{1}{4} \LinGTr (\mathring{\rho}
 -  \bar{\mathring{\rho}})
 -  \tfrac{1}{2} \LinGZero \mathring{\rho}'
 + \overline{\dot{\eta}} \mathring{\tau}
 + \tfrac{1}{2} \LinGOneDg \mathring{\tau}'
 + \tfrac{1}{2} \LinGOne \bar{\mathring{\tau}}'
 -  \tfrac{1}{4} \thoBG \LinGTr
 -  \tfrac{1}{2} \edtBG \LinGOne
 -  \edtpBG \dot{\eta}
 -  \tfrac{1}{2} \edtpBG \LinGOneDg,\\
\dot{\tilde{\sigma}}={}&\tfrac{1}{2} \LinGTwoDg (\mathring{\rho}
 -  \bar{\mathring{\rho}})
 + \dot{\eta} \mathring{\tau}
 + \LinGOneDg \bar{\mathring{\tau}}'
 + \tfrac{1}{2} \thoBG \LinGTwoDg
 -  \edtBG \dot{\eta}
 -  \edtBG \LinGOneDg.
\end{align}
\end{subequations}
The linearization of the Ricci relations \eqref{eq:ExtraRicci} are
\begin{subequations}
\label{eq:extraLinearizedRicci}
\begin{align}
\edtpBG \dot{\tilde{\rho}}' -  \edtBG \dot{\tilde{\sigma}}'={}&- \dot{\Psi}_{3}{}
 -  \LinGOne \mathring{\rho}' (\mathring{\rho}' -  \bar{\mathring{\rho}}')
 -  \overline{\dot{\eta}} \mathring{\rho}' (2 \mathring{\rho}' -  \bar{\mathring{\rho}}')
 -  \tfrac{1}{2} \LinGTwo (2 \mathring{\rho}' -  \bar{\mathring{\rho}}') \mathring{\tau}
 + \dot{\tilde{\rho}}' \mathring{\tau}'
 -  \overline{\dot{\tilde{\rho}}'} \mathring{\tau}'\nonumber\\
& + i \dot{\nu} (2 \mathring{\rho}' -  \bar{\mathring{\rho}}') \mathring{\tau}'
 + \tfrac{1}{4} \LinGTr (2 \mathring{\rho}' -  \bar{\mathring{\rho}}') \mathring{\tau}'
 + 2 \mathring{\rho}' \dot{\tilde{\tau}}'
 -  \bar{\mathring{\rho}}' \dot{\tilde{\tau}}',\\
\thoBG \dot{\tilde{\sigma}}' -  \edtpBG \dot{\tilde{\tau}}'={}&\mathring{\rho}' \overline{\dot{\tilde{\sigma}}}
 + \mathring{\rho} \dot{\tilde{\sigma}}'
 -  \overline{\dot{\eta}} \mathring{\rho}' (\bar{\mathring{\tau}} -  \mathring{\tau}')
 -  \overline{\dot{\tilde{\beta}}} \mathring{\tau}'
 -  \dot{\tilde{\beta}}' \mathring{\tau}'
 - i \dot{\nu} \mathring{\tau}'^2
 -  \tfrac{1}{4} \LinGTr \mathring{\tau}'^2
 - 2 \mathring{\tau}' \dot{\tilde{\tau}}'
 + \tfrac{1}{2} \LinGTwo \edtpBG \mathring{\tau},\\
\thoBG \dot{\tilde{\rho}}' -  \edtBG \dot{\tilde{\tau}}'={}&- \dot{\Psi}_{2}{}
 -  \dot{\tilde{\epsilon}} \mathring{\rho}'
 -  \overline{\dot{\tilde{\epsilon}}} \mathring{\rho}'
 + \overline{\dot{\tilde{\rho}}} \mathring{\rho}'
 + \tfrac{1}{2} \LinGZero \mathring{\rho}'^2
 + \bar{\mathring{\rho}} \dot{\tilde{\rho}}'
 - 2 \overline{\dot{\eta}} \mathring{\rho}' \mathring{\tau}
 - 2 \LinGOne \mathring{\rho}' \mathring{\tau}
 - 2 \dot{\eta} \mathring{\rho}' \bar{\mathring{\tau}}
 + \dot{\tilde{\beta}} \mathring{\tau}'
 + \overline{\dot{\tilde{\beta}}'} \mathring{\tau}'\nonumber\\
& + \tfrac{1}{2} \LinGTwoDg \mathring{\tau}'^2
 -  \LinGOneDg \mathring{\rho}' (\bar{\mathring{\tau}} + \mathring{\tau}')
 -  \mathring{\tau}' \overline{\dot{\tilde{\tau}}'}
 -  \dot{\tilde{\tau}}' \bar{\mathring{\tau}}'
 + i \dot{\nu} \edtpBG \mathring{\tau}
 -  \tfrac{1}{4} \LinGTr \edtpBG \mathring{\tau},\\
\thoBG \dot{\tilde{\rho}} -  \edtpBG \dot{\tilde{\kappa}}={}&\dot{\tilde{\epsilon}} \mathring{\rho}
 + \overline{\dot{\tilde{\epsilon}}} \mathring{\rho}
 + 2 \mathring{\rho} \dot{\tilde{\rho}}
 -  \overline{\dot{\tilde{\kappa}}} \mathring{\tau}
 -  (\overline{\dot{\eta}} + \LinGOne) (\mathring{\rho} -  \bar{\mathring{\rho}}) \mathring{\tau}
 -  \dot{\tilde{\kappa}} \mathring{\tau}'
 - 2 \dot{\eta} \mathring{\rho} \mathring{\tau}'
 - 2 \LinGOneDg \mathring{\rho} \mathring{\tau}'
 + \tfrac{1}{2} \LinGZero \thoBG \mathring{\rho}',\\
\thoBG \dot{\tilde{\sigma}} -  \edtBG \dot{\tilde{\kappa}}={}&\dot{\Psi}_{0}{}
 + \mathring{\rho} \dot{\tilde{\sigma}}
 + \bar{\mathring{\rho}} \dot{\tilde{\sigma}}
 -  \dot{\tilde{\kappa}} (\mathring{\tau} + \bar{\mathring{\tau}}'),\\
\thoBG \dot{\tilde{\beta}} -  \edtBG \dot{\tilde{\epsilon}}={}&\dot{\Psi}_{1}{}
 + \dot{\tilde{\beta}} \bar{\mathring{\rho}}
 + \dot{\tilde{\kappa}} \mathring{\rho}'
 -  \LinGOneDg (\mathring{\Psi}_{2}{} + \mathring{\rho} \mathring{\rho}')
 + \tfrac{1}{2} \LinGZero \mathring{\rho}' \mathring{\tau}
 + \tfrac{1}{2} \LinGTwoDg \mathring{\rho} \mathring{\tau}'
 -  \dot{\tilde{\sigma}} \mathring{\tau}'
 -  \dot{\eta} (2 \mathring{\Psi}_{2}{} + \mathring{\rho} \mathring{\rho}' -  \mathring{\tau} \mathring{\tau}')\nonumber\\
& -  \dot{\tilde{\epsilon}} \bar{\mathring{\tau}}',\\
\thoBG \dot{\tilde{\beta}}' + \edtpBG \dot{\tilde{\epsilon}}={}&\dot{\tilde{\beta}}' \mathring{\rho}
 -  \LinGOne (\mathring{\Psi}_{2}{} + \mathring{\rho} \mathring{\rho}')
 + \dot{\tilde{\epsilon}} \mathring{\tau}'
 + i \dot{\nu} \mathring{\rho} \mathring{\tau}'
 + \tfrac{1}{4} \LinGTr \mathring{\rho} \mathring{\tau}'
 + \dot{\tilde{\rho}} \mathring{\tau}'
 -  \overline{\dot{\eta}} (\mathring{\rho} \mathring{\rho}' + \mathring{\tau} \mathring{\tau}')
 + \mathring{\rho} \dot{\tilde{\tau}}',\\
\edtBG \dot{\tilde{\beta}}' + \edtpBG \dot{\tilde{\beta}}={}&\dot{\Psi}_{2}{}
 + \dot{\tilde{\rho}} \mathring{\rho}'
 + \tfrac{1}{2} \LinGTr (\mathring{\Psi}_{2}{} + \mathring{\rho} \mathring{\rho}')
 + \mathring{\rho} \dot{\tilde{\rho}}'
 + \dot{\tilde{\epsilon}} (\mathring{\rho}' -  \bar{\mathring{\rho}}')
 -  \overline{\dot{\eta}} \mathring{\rho}' \mathring{\tau},\\
\edtBG \dot{\tilde{\rho}} -  \edtpBG \dot{\tilde{\sigma}}={}&- \dot{\Psi}_{1}{}
 + \dot{\tilde{\beta}} \mathring{\rho}
 -  \overline{\dot{\tilde{\beta}}'} \mathring{\rho}
 -  \dot{\tilde{\kappa}} (\mathring{\rho}' -  \bar{\mathring{\rho}}')
 + \dot{\tilde{\rho}} \mathring{\tau}
 -  \overline{\dot{\tilde{\rho}}} \mathring{\tau}
 + \tfrac{1}{4} (-4i \dot{\nu} + \LinGTr) (\mathring{\rho} -  \bar{\mathring{\rho}}) \mathring{\tau}
 -  \LinGTwoDg \mathring{\rho} \mathring{\tau}'
 -  \dot{\eta} \thoBG \mathring{\rho}'.
\end{align}
\end{subequations}

\section{Translation between this paper and \texorpdfstring{\cite{Andersson:2019dwi}}{[1]}}
If we impose
\begin{align}
\LinGTr={}&0,&
\dot{\tilde{\rho}}'={}&0,&
\dot{\nu}={}&0,
\end{align}
it is possible to translate to the variables in \cite{Andersson:2019dwi} using the relations in table \ref {table:comparison} as well as background GHP commutator relations to eliminate all second-order derivatives.
In particular, equations \eqref{eq:ThopFGDiffLLinEq1}, \eqref{eq:ThopLinG}, and \eqref{eq:ThopLinConnection} translate to
\begin{subequations}
\label{eq:translatedLinThop}
\begin{align}
\thopBG \dot{\eta}={}&- \overline{\tilde{\beta}'}
 + \tilde{\beta}
 -  \tfrac{1}{2} G_{01'} \mathring{\rho}'
 + \tfrac{1}{2} G_{02'} \bar{\mathring{\tau}},
\label{eq:translatedEta}\\
\thopBG G_{20'}={}&G_{20'} \bar{\mathring{\rho}}'
 + 2 \tilde{\sigma}',\\
\thopBG G_{10'}={}&G_{10'} \bar{\mathring{\rho}}'
 -  G_{20'} \mathring{\tau}
 + 2 \tilde{\tau}',\\
\thopBG G_{00'}={}&-2 \tilde{\epsilon}
 - 2 \overline{\tilde{\epsilon}}
 - 2 G_{10'} \mathring{\tau}
 - 2 G_{01'} \bar{\mathring{\tau}}
 + G_{01'} \mathring{\tau}'
 + G_{10'} \bar{\mathring{\tau}}',\\
\thopBG \tilde{\sigma}'={}&\vartheta \Psi_{4}{}
 + \bar{\mathring{\rho}}' \tilde{\sigma}',\\
\thopBG \tilde{\tau}'={}&\vartheta \Psi_{3}{}
 + \tilde{\beta}' \mathring{\rho}'
 -  \overline{\tilde{\beta}} \mathring{\rho}'
 -  \tfrac{1}{2} G_{10'} \mathring{\rho}' \bar{\mathring{\rho}}'
 -  \tilde{\sigma}' \mathring{\tau}
 + \tfrac{1}{2} G_{20'} \mathring{\rho}' \bar{\mathring{\tau}}'
 + \tilde{\sigma}' \bar{\mathring{\tau}}'
 -  \mathring{\rho}' \tilde{\tau}',\\
\thopBG \tilde{\beta}'={}&\vartheta \Psi_{3}{}
 + \tilde{\beta}' \bar{\mathring{\rho}}'
 -  \tilde{\sigma}' \mathring{\tau},\\
\thopBG \tilde{\beta}={}&- \overline{\tilde{\beta}'} \mathring{\rho}'
 + 2 \tilde{\beta} \mathring{\rho}'
 + \mathring{\rho}' \overline{\tilde{\tau}'},\\
\thopBG \tilde{\epsilon}={}&- \vartheta \Psi_{2}{}
 + \tilde{\beta}' \mathring{\tau}
 -  \tilde{\beta} \bar{\mathring{\tau}}
 -  \overline{\tilde{\beta}'} \mathring{\tau}'
 + 2 \tilde{\beta} \mathring{\tau}'
 -  \tilde{\beta}' \bar{\mathring{\tau}}'
 + \mathring{\tau} \tilde{\tau}'
 + \mathring{\tau}' \overline{\tilde{\tau}'},\\
\thopBG \tilde{\rho}={}&- \vartheta \Psi_{2}{}
 + \bar{\mathring{\rho}}' \tilde{\rho}
 + 2 \tilde{\beta}' \mathring{\tau}
 + \overline{\tilde{\beta}'} \bar{\mathring{\tau}}
 -  \tilde{\beta} \bar{\mathring{\tau}}
 -  \bar{\mathring{\tau}} \overline{\tilde{\tau}'}
 -  \edtpBG\overline{\tilde{\beta}'}
 + \edtpBG\tilde{\beta}
 + \edtpBG\overline{\tilde{\tau}'},\\
\thopBG \tilde{\sigma}={}&\tfrac{1}{2} G_{02'} \mathring{\Psi}_{2}{}
 + \mathring{\rho}' \tilde{\sigma}
 + \overline{\tilde{\beta}'} \mathring{\tau}
 - 3 \tilde{\beta} \mathring{\tau}
 -  \mathring{\tau} \overline{\tilde{\tau}'}
 -  \edtBG\overline{\tilde{\beta}'}
 + \edtBG\tilde{\beta}
 + \edtBG\overline{\tilde{\tau}'},\\
\thopBG \tilde{\kappa}={}&\tfrac{1}{2} G_{01'} \mathring{\Psi}_{2}{}
 -  \vartheta \Psi_{1}{}
 + \overline{\tilde{\beta}'} \mathring{\rho}
 -  \tilde{\beta} \mathring{\rho}
 -  \tfrac{1}{2} G_{01'} \mathring{\rho} \mathring{\rho}'
 - 2 \tilde{\epsilon} \mathring{\tau}
 -  \tilde{\rho} \mathring{\tau}
 -  \tilde{\sigma} \bar{\mathring{\tau}}
 + \tfrac{1}{2} G_{02'} \mathring{\rho} \mathring{\tau}'
 + \tilde{\sigma} \mathring{\tau}'
 + \tilde{\rho} \bar{\mathring{\tau}}'
 -  \mathring{\rho} \overline{\tilde{\tau}'}\nonumber\\*
& -  \thoBG \overline{\tilde{\beta}'}
 + \thoBG \tilde{\beta}
 + \thoBG \overline{\tilde{\tau}'}.
\end{align}
\end{subequations}
The linearized structure equations \eqref{eq:extraLinearizedStructure} are equivalent to
\begin{subequations}
\label{eq:translatedLinStructure}
\begin{align}
\tilde{\kappa}'={}&0,\qquad
\tilde{\epsilon}'={}0,\\
\tilde{\tau}={}&- \tfrac{1}{2} G_{01'} \mathring{\rho}'
 + \tfrac{1}{2} G_{02'} \mathring{\tau}',\\
\tilde{\beta} - \overline{\tilde{\beta}'}={}&- \tfrac{1}{2} G_{01'} \mathring{\rho}'
 + \tfrac{1}{2} G_{02'} \mathring{\tau}'
 -  \overline{\tilde{\tau}'},\\
\tilde{\rho} -  \overline{\tilde{\rho}}={}&- \tfrac{1}{2} G_{00'} \mathring{\rho}'
 + \tfrac{1}{2} G_{00'} \bar{\mathring{\rho}}'
 -  \tfrac{1}{2} \edtBG G_{10'}
 + \tfrac{1}{2} \edtpBG G_{01'},\\
\tilde{\kappa}={}&- G_{01'} \bar{\mathring{\rho}}
 + \tfrac{1}{2} G_{00'} \bar{\mathring{\tau}}'
 + \tfrac{1}{2} \thoBG G_{01'}
 -  \tfrac{1}{2} \edtBG G_{00'},\\
\tilde{\epsilon} -  \overline{\tilde{\epsilon}}={}&- \tfrac{1}{2} G_{00'} \mathring{\rho}'
 + \tfrac{1}{2} G_{00'} \bar{\mathring{\rho}}'
 -  \tfrac{1}{2} \edtBG G_{10'}
 + \tfrac{1}{2} \edtpBG G_{01'},\\
\tilde{\beta}'={}&\tfrac{1}{2} G_{10'} \mathring{\rho}'
 -  \tfrac{1}{4} G_{20'} \bar{\mathring{\tau}}'
 + \tfrac{1}{2} \tilde{\tau}'
 + \tfrac{1}{4} \edtBG G_{20'},\\
\tilde{\rho}={}&- \tfrac{1}{2} G_{00'} \mathring{\rho}'
 + \tfrac{1}{2} G_{01'} \mathring{\tau}'
 + \tfrac{1}{2} G_{10'} \bar{\mathring{\tau}}'
 -  \tfrac{1}{2} \edtBG G_{10'},\\
\tilde{\sigma}={}&- \tfrac{1}{2} G_{02'} \bar{\mathring{\rho}}
 + G_{01'} \bar{\mathring{\tau}}'
 + \tfrac{1}{2} \thoBG G_{02'}
 -  \tfrac{1}{2} \edtBG G_{01'}.
\end{align}
\end{subequations}

The linearized Ricci relations in equations \eqref{eq:extraLinearizedRicci} are equivalent to  
\begin{subequations}
\label{eq:translatedLinRicci}
\begin{align}
\vartheta \Psi_{3}{}={}&2 \vartheta \Psi_{3}{}
 -  G_{10'} \mathring{\rho}'^2
 -  \tfrac{1}{2} G_{20'} \bar{\mathring{\rho}}' \mathring{\tau}
 - 2 \mathring{\rho}' \tilde{\tau}'
 + \bar{\mathring{\rho}}' \tilde{\tau}'
 -  \tfrac{1}{2} \mathring{\rho}' \edtBG G_{20'}
 -  \edtBG \tilde{\sigma}',\\
\vartheta \Psi_{2}{}={}&- \tilde{\epsilon} \mathring{\rho}'
 -  \overline{\tilde{\epsilon}} \mathring{\rho}'
 + \tfrac{1}{2} G_{00'} \mathring{\rho}'^2
 + \mathring{\rho}' \overline{\tilde{\rho}}
 + \tfrac{1}{2} G_{10'} \bar{\mathring{\rho}}' \mathring{\tau}
 + \overline{\tilde{\beta}'} \mathring{\tau}'
 + \tilde{\beta} \mathring{\tau}'
 -  G_{01'} \mathring{\rho}' \mathring{\tau}'
 + \tfrac{1}{2} G_{02'} \mathring{\tau}'^2\nonumber\\
&-  \bar{\mathring{\tau}}' \tilde{\tau}'
  -  \mathring{\tau}' \overline{\tilde{\tau}'}
 + \tfrac{1}{2} \mathring{\rho}' \edtBG G_{10'}
 + \edtBG \tilde{\tau}',\\
\vartheta \Psi_{2}{}={}&- \tilde{\epsilon} \mathring{\rho}'
 + \tilde{\epsilon} \bar{\mathring{\rho}}'
 -  \mathring{\rho}' \tilde{\rho}
 + \edtBG \tilde{\beta}'
 + \edtpBG \tilde{\beta},\\
\vartheta \Psi_{1}{}={}&- \tilde{\beta} \bar{\mathring{\rho}}
 -  \tilde{\kappa} \mathring{\rho}'
 + \tilde{\sigma} \mathring{\tau}'
 + \tilde{\epsilon} \bar{\mathring{\tau}}'
 + \thoBG \tilde{\beta}
 -  \edtBG \tilde{\epsilon},\\
\vartheta \Psi_{0}{}={}&- \tfrac{1}{2} G_{02'} \mathring{\rho} \bar{\mathring{\rho}}
 -  \mathring{\rho} \tilde{\sigma}
 -  \bar{\mathring{\rho}} \tilde{\sigma}
 + \tilde{\kappa} \mathring{\tau}
 + \tilde{\kappa} \bar{\mathring{\tau}}'
 -  \tfrac{1}{2} G_{00'} \mathring{\tau} \bar{\mathring{\tau}}'
 -  \tfrac{1}{2} \mathring{\tau} \thoBG G_{01'}
 + \tfrac{1}{2} \mathring{\rho} \thoBG G_{02'}
 + \thoBG \tilde{\sigma}\nonumber\\
& + \tfrac{1}{2} \mathring{\tau} \edtBG G_{00'}
 -  \tfrac{1}{2} \mathring{\rho} \edtBG G_{01'}
 -  \edtBG \tilde{\kappa},\\
\thoBG \tilde{\sigma}' -  \edtpBG \tilde{\tau}'={}&\mathring{\rho} \tilde{\sigma}'
 + \mathring{\rho}' \overline{\tilde{\sigma}}
 -  \tilde{\beta}' \mathring{\tau}'
 -  \overline{\tilde{\beta}} \mathring{\tau}'
 -  \tfrac{1}{2} G_{10'} (2 \mathring{\rho}' + \bar{\mathring{\rho}}') \mathring{\tau}'
 + \tfrac{1}{2} G_{20'} (\mathring{\Psi}_{2}{} + \mathring{\rho} \mathring{\rho}' + \mathring{\tau}' \bar{\mathring{\tau}}')
 - 2 \mathring{\tau}' \tilde{\tau}'\nonumber\\
& -  \tfrac{1}{2} \mathring{\rho}' \thoBG G_{20'}
 + \tfrac{1}{2} \mathring{\rho}' \edtpBG G_{10'},\\
\thoBG \tilde{\rho} -  \edtpBG \tilde{\kappa}={}&\tilde{\epsilon} \mathring{\rho}
 + \overline{\tilde{\epsilon}} \mathring{\rho}
 + 2 \mathring{\rho} \tilde{\rho}
 -  \overline{\tilde{\kappa}} \mathring{\tau}
 -  \tfrac{1}{2} G_{10'} (2 \mathring{\rho} -  \bar{\mathring{\rho}}) \mathring{\tau}
 -  \tilde{\kappa} \mathring{\tau}'
 -  \tfrac{1}{2} G_{01'} \mathring{\rho} \mathring{\tau}'
 -  \tfrac{1}{2} G_{00'} (\mathring{\Psi}_{2}{} -  \bar{\mathring{\rho}} \mathring{\rho}' -  \mathring{\tau} \mathring{\tau}')\nonumber\\
& + \tfrac{1}{2} \mathring{\tau} \thoBG G_{10'}
 -  \tfrac{1}{2} \mathring{\tau} \edtpBG G_{00'}
 + \tfrac{1}{2} \mathring{\rho} \edtpBG G_{01'},\\
\thoBG \tilde{\beta}' + \edtpBG \tilde{\epsilon}={}&- G_{10'} \mathring{\Psi}_{2}{}
 + \tilde{\beta}' \mathring{\rho}
 + \tilde{\epsilon} \mathring{\tau}'
 + \tilde{\rho} \mathring{\tau}'
 + \mathring{\rho} \tilde{\tau}',\\
\edtBG \tilde{\rho} -  \edtpBG \tilde{\sigma}={}&- \vartheta \Psi_{1}{}
 -  \overline{\tilde{\beta}'} \mathring{\rho}
 + \tilde{\beta} \mathring{\rho}
 -  \tilde{\kappa} \mathring{\rho}'
 -  \tfrac{1}{2} G_{01'} (\mathring{\Psi}_{2}{} + \mathring{\rho} \mathring{\rho}' - 2 \bar{\mathring{\rho}} \mathring{\rho}')
 + \tilde{\kappa} \bar{\mathring{\rho}}'
 + \tfrac{1}{2} G_{00'} (\mathring{\rho}' -  \bar{\mathring{\rho}}') \mathring{\tau}
 + \tilde{\rho} \mathring{\tau}\nonumber\\
& -  \overline{\tilde{\rho}} \mathring{\tau}
 + (\mathring{\rho} -  \bar{\mathring{\rho}}) \thopBG \dot{\eta}
 + \tfrac{1}{2} (\mathring{\rho} -  \bar{\mathring{\rho}}) \thopBG G_{01'}
 + \tfrac{1}{2} \mathring{\tau} \edtBG G_{10'}
 -  \tfrac{1}{2} \mathring{\tau} \edtpBG G_{01'}
 + \tfrac{1}{2} \mathring{\rho} \edtpBG G_{02'}.
\end{align}
\end{subequations}
The linearized Bianchi equations \eqref{eq:LinBianchi} translate to
\begin{subequations}
\label{eq:translatedLinBianchi}
\begin{align}
\thoBG \vartheta \Psi_{1}{} -  \edtpBG \vartheta \Psi_{0}{}={}&-3 \mathring{\Psi}_{2}{} \tilde{\kappa}
 -  \tfrac{3}{2} G_{01'} \mathring{\Psi}_{2}{} \mathring{\rho}
 + 4 \vartheta \Psi_{1}{} \mathring{\rho}
 + \tfrac{3}{2} G_{00'} \mathring{\Psi}_{2}{} \mathring{\tau}
 -  \vartheta \Psi_{0}{} \mathring{\tau}',\\
\thoBG \vartheta \Psi_{2}{} -  \edtpBG \vartheta \Psi_{1}{}={}&3 \vartheta \Psi_{2}{} \mathring{\rho}
 + \tfrac{3}{2} G_{00'} \mathring{\Psi}_{2}{} \mathring{\rho}'
 + 3 \mathring{\Psi}_{2}{} \tilde{\rho}
 - 3 G_{10'} \mathring{\Psi}_{2}{} \mathring{\tau}
 -  \tfrac{3}{2} G_{01'} \mathring{\Psi}_{2}{} \mathring{\tau}'
 - 2 \vartheta \Psi_{1}{} \mathring{\tau}',\\
\thoBG \vartheta \Psi_{3}{} -  \edtpBG \vartheta \Psi_{2}{}={}&2 \vartheta \Psi_{3}{} \mathring{\rho}
 - 3 G_{10'} \mathring{\Psi}_{2}{} \mathring{\rho}'
 + \tfrac{3}{2} G_{20'} \mathring{\Psi}_{2}{} \mathring{\tau}
 - 3 \vartheta \Psi_{2}{} \mathring{\tau}'
 - 3 \mathring{\Psi}_{2}{} \tilde{\tau}',\\
\thoBG \vartheta \Psi_{4}{} -  \edtpBG \vartheta \Psi_{3}{}={}&\vartheta \Psi_{4}{} \mathring{\rho}
 + \tfrac{3}{2} G_{20'} \mathring{\Psi}_{2}{} \mathring{\rho}'
 + 3 \mathring{\Psi}_{2}{} \tilde{\sigma}'
 - 4 \vartheta \Psi_{3}{} \mathring{\tau}',\\
\thopBG \vartheta \Psi_{0}{} -  \edtBG \vartheta \Psi_{1}{}={}&\tfrac{3}{2} G_{02'} \mathring{\Psi}_{2}{} \mathring{\rho}
 + \vartheta \Psi_{0}{} \mathring{\rho}'
 + 3 \mathring{\Psi}_{2}{} \tilde{\sigma}
 -  \tfrac{3}{2} G_{01'} \mathring{\Psi}_{2}{} \mathring{\tau}
 - 4 \vartheta \Psi_{1}{} \mathring{\tau},\\
\thopBG \vartheta \Psi_{1}{} -  \edtBG \vartheta \Psi_{2}{}={}&- \tfrac{3}{2} G_{01'} \mathring{\Psi}_{2}{} \mathring{\rho}'
 + 2 \vartheta \Psi_{1}{} \mathring{\rho}'
 - 3 \vartheta \Psi_{2}{} \mathring{\tau}
 + \tfrac{3}{2} G_{02'} \mathring{\Psi}_{2}{} \mathring{\tau}'
 - 3 \mathring{\Psi}_{2}{} \thopBG \dot{\eta}
 -  \tfrac{3}{2} \mathring{\Psi}_{2}{} \thopBG G_{01'},\\
\thopBG \vartheta \Psi_{2}{} -  \edtBG \vartheta \Psi_{3}{}={}&3 \vartheta \Psi_{2}{} \mathring{\rho}'
 - 2 \vartheta \Psi_{3}{} \mathring{\tau},\\
\thopBG \vartheta \Psi_{3}{} -  \edtBG \vartheta \Psi_{4}{}={}&4 \vartheta \Psi_{3}{} \mathring{\rho}'
 -  \vartheta \Psi_{4}{} \mathring{\tau}.
\end{align}
\end{subequations}
Observe that the $\dot\eta$ dependence of the equations \eqref{eq:translatedLinRicci} and \eqref{eq:translatedLinBianchi} can be eliminated by using \eqref{eq:translatedEta}.

%
%

%

\bigskip

\end{document}